\documentclass[11pt]{article}
\usepackage{ifthen}
\usepackage{mdwlist}
\usepackage{amsmath,amssymb,amsfonts,amsthm}
\usepackage{bm}
\usepackage{dsfont}
\usepackage[dvipsnames]{xcolor}
\usepackage[colorlinks=true,linkcolor=blue,citecolor=blue,urlcolor=blue]{hyperref}
\usepackage{enumitem}
\usepackage[pdftex]{graphicx}
\usepackage{palatino} %
\usepackage[varg,bigdelims]{newpxmath}
\usepackage[scr=rsfso]{mathalfa}%
\usepackage{xspace}
\usepackage{verbatim}
\usepackage{algorithm}
\usepackage{algpseudocode}

\usepackage{bbold}
\usepackage[margin=1in]{geometry}
\usepackage{cleveref}

\usepackage{thm-restate}
\usepackage{latexsym}
\usepackage{epsfig}
\usepackage[colorinlistoftodos,size=tiny]{todonotes}
\usepackage[normalem]{ulem}
\usepackage{microtype}

\def\withcolors{1}
\def\withnotes{1}

\newtheorem{theorem}{Theorem}

\newtheorem{proposition}[theorem]{Proposition}

\newtheorem{remark}[theorem]{Remark}
\newtheorem{fact}[theorem]{Fact}
\newtheorem{lemma}[theorem]{Lemma}

\newtheorem{corollary}[theorem]{Corollary}

\newtheorem{definition}[theorem]{Definition}

\numberwithin{theorem}{section} 
\numberwithin{assumption}{section}
\numberwithin{tournament}{section}
\numberwithin{problem}{section}
\numberwithin{nontheorem}{section} 
\numberwithin{proposition}{section} 
\numberwithin{observation}{section} 
\numberwithin{remark}{section} 
\numberwithin{fact}{section} 
\numberwithin{lemma}{section} 
\numberwithin{claim}{section} 
\numberwithin{corollary}{section} 
\numberwithin{case}{section} 
\numberwithin{dfn}{section} 
\numberwithin{definition}{section} 
\numberwithin{question}{section} 
\numberwithin{openquestion}{section} 
\numberwithin{res}{section} 

\usepackage{anyfontsize}

\ifnum\withcolors=1
  \newcommand{\gcolor}[1]{{\color{red}#1}}

\else

  \newcommand{\gcolor}[1]{{#1}}

\fi

\ifnum\withnotes=1
  \newcommand{\gnote}[1]{\par\gcolor{\textbf{G: }\sf #1}} %
  \newcommand{\gfootnote}[1]{\footnote{{\bf \gcolor{Gautam}}: {#1}}}

  \newcommand{\todonote}[2][]{}  
	\newcommand{\questionnote}[2][]{}
	\newcommand{\todonotedone}[2][]{}
	\newcommand{\todonoteinline}[2][]{}  
  \newcommand{\marginnote}[1]{}
\else
  \newcommand{\gnote}[1]{}
  \newcommand{\gfootnote}[1]{}
  \newcommand{\todonote}[2][]{\ignore{#2}}
	\newcommand{\questionnote}[2][]{\ignore{#2}}
	\newcommand{\todonotedone}[2][]{\ignore{#2}}
	\newcommand{\todonoteinline}[2][]{\ignore{#2}}
  \newcommand{\marginnote}[1]{\ignore{#1}}
\fi

\usepackage{todonotes}

\let\originalleft\left
\let\originalright\right
\renewcommand{\left}{\mathopen{}\mathclose\bgroup\originalleft}
\renewcommand{\right}{\aftergroup\egroup\originalright}

\newcommand{\ignore}[1]{\leavevmode\unskip} %

\newcommand{\paren}[1]{(#1)}
\newcommand{\Paren}[1]{\left(#1\right)}

\newcommand{\brac}[1]{[#1]}
\newcommand{\Brac}[1]{\left[#1\right]}

\newcommand{\floor}[1]{\lfloor#1\rfloor}

\newcommand{\abs}[1]{\lvert#1\rvert}
\newcommand{\Abs}[1]{\left\lvert#1\right\rvert}

\newcommand{\set}[1]{\{#1\}}
\newcommand{\Set}[1]{\left\{#1\right\}}

\newcommand{\norm}[1]{\lVert#1\rVert}
\newcommand{\Norm}[1]{\left\lVert#1\right\rVert}

\newcommand{\Bignorm}[1]{\Big\lVert#1\Big\rVert}

\newcommand{\normt}[1]{\norm{#1}_2}
\newcommand{\Normt}[1]{\Norm{#1}_2}

\newcommand{\Bignormt}[1]{\Bignorm{#1}_2}
\newcommand{\snormt}[1]{\norm{#1}^2_2}
\newcommand{\Snormt}[1]{\Norm{#1}^2_2}

\newcommand{\iprod}[1]{\langle#1\rangle}

\newcommand{\normsch}[1]{\vert\kern-0.25ex\vert\kern-0.25ex\vert #1 
\vert\kern-0.25ex\vert\kern-0.25ex\vert}
\newcommand{\Normsch}[1]{\left\vert\kern-0.25ex\left\vert\kern-0.25ex\left\vert #1 
\right\vert\kern-0.25ex\right\vert\kern-0.25ex\right\vert}
\newcommand{\Esymb}{\mathbb{E}}

\newcommand{\Psymb}{\mathbb{P}}
\newcommand{\Vsymb}{\mathbb{V}}
\newcommand{\bR}{\mathbb{R}}
\DeclareMathOperator*{\E}{\Esymb}

\DeclareMathOperator*{\Var}{\Vsymb}
\DeclareMathOperator*{\ProbOp}{\Psymb}
\renewcommand{\Pr}{\ProbOp}

\newcommand{\tensor}{\otimes}

\newcommand{\mper}{\,.}
\newcommand{\mcom}{\,,}
\newcommand\bdot\bullet

\DeclareMathOperator{\Ind}{\mathbf 1}

\DeclareMathOperator{\poly}{poly}

\newcommand{\iid}{i.i.d.\xspace}

\newcommand{\N}{\mathbb N}
\newcommand{\R}{\mathbb R}

\newcommand{\cD}{\mathcal D}
\newcommand{\cE}{\mathcal E}

\newcommand{\cH}{\mathcal H}

\newcommand{\cL}{\mathcal L}

\newcommand{\cN}{\mathcal N}
\newcommand{\cO}{\mathcal O}

\newcommand{\cU}{\mathcal U}

\newcommand{\cX}{\mathcal X}

\renewcommand{\leq}{\leqslant}
\renewcommand{\le}{\leqslant}
\renewcommand{\geq}{\geqslant}
\renewcommand{\ge}{\geqslant}
\let\epsilon=\varepsilon
\newcommand\MYcurrentlabel{xxx}
\newcommand{\MYstore}[2]{%
	\global\expandafter \def \csname MYMEMORY #1 \endcsname{#2}%
}
\newcommand{\MYload}[1]{%
	\csname MYMEMORY #1 \endcsname%
}
\newcommand{\MYnewlabel}[1]{%
	\renewcommand\MYcurrentlabel{#1}%
	\MYoldlabel{#1}%
}
\newcommand{\MYdummylabel}[1]{}
\newcommand{\torestate}[1]{%
	\let\MYoldlabel\label%
	\let\label\MYnewlabel%
	#1%
	\MYstore{\MYcurrentlabel}{#1}%
	\let\label\MYoldlabel%
}
\newcommand{\restatetheorem}[1]{%
	\let\MYoldlabel\label
	\let\label\MYdummylabel
	\begin{theorem*}[Restatement of \cref{#1}]
		\MYload{#1}
	\end{theorem*}
	\let\label\MYoldlabel
}
\newcommand{\restatelemma}[1]{%
	\let\MYoldlabel\label
	\let\label\MYdummylabel
	\begin{lemma*}[Restatement of \cref{#1}]
		\MYload{#1}
	\end{lemma*}
	\let\label\MYoldlabel
}
\newcommand{\restateprop}[1]{%
	\let\MYoldlabel\label
	\let\label\MYdummylabel
	\begin{proposition*}[Restatement of \cref{#1}]
		\MYload{#1}
	\end{proposition*}
	\let\label\MYoldlabel
}
\newcommand{\restatefact}[1]{%
	\let\MYoldlabel\label
	\let\label\MYdummylabel
	\begin{fact*}[Restatement of \cref{#1}]
		\MYload{#1}
	\end{fact*}
	\let\label\MYoldlabel
}
\newcommand{\restate}[1]{%
	\let\MYoldlabel\label
	\let\label\MYdummylabel
	\MYload{#1}
	\let\label\MYoldlabel
}

\newcommand{\eps}{\epsilon}

\allowdisplaybreaks
\newcommand*{\normtv}[1]{\Norm{#1}_{\mathrm{TV}}
}

\DeclareMathOperator*{\argmin}{arg\,min}

\makeatletter
\newcommand{\distas}[1]{\mathbin{\overset{#1}{\kern\z@\sim}}}%
\newsavebox{\mybox}\newsavebox{\mysim}
\newcommand{\distras}[1]{%
	\savebox{\mybox}{\hbox{\kern3pt$\scriptstyle#1$\kern3pt}}%
	\savebox{\mysim}{\hbox{$\sim$}}%
	\mathbin{\overset{#1}{\kern\z@\resizebox{\wd\mybox}{\ht\mysim}{$\sim$}}}%
}
\makeatother

\newcommand{\wb}[1]{\ensuremath{\overline{#1}}}

\newcommand{\RNum}[1]{\uppercase\expandafter{\romannumeral #1\relax}}

\newcommand*\diff{\mathop{}\!\mathrm{d}}

\title{Private Mean Estimation \\ with Person-Level Differential Privacy\thanks{Authors are listed in alphabetical order.}}

\author{
~~~~~~ Sushant Agarwal\thanks{Khoury College of Computer Sciences, Northeastern University.  {\tt agarwal.sus@northeastern.edu}.  Supported by NSF award CNS-2247484.}
\and
Gautam Kamath\thanks{Cheriton School of Computer Science, University of Waterloo and Vector Institute. {\tt g@csail.mit.edu}. Supported by a University of Waterloo startup grant, an NSERC Disovery Grant, a Canada CIFAR AI Chair, and an unrestricted gift from Google.}
\and
Mahbod Majid\thanks{
Machine Learning Department, Carnegie Mellon University.
{\tt mahbodm@andrew.cmu.edu}.}~~~~
\and
Argyris Mouzakis\thanks{
Cheriton School of Computer Science, University of Waterloo.
{\tt amouzaki@uwaterloo.ca}.
Supported by a University of Waterloo startup grant, a Canada CIFAR AI Chair, an unrestricted gift from Google, and a scholarship from the Onassis Foundation (Scholarship ID: F ZT 053-1/2023-2024)
}
\and
Rose Silver\thanks{Khoury College of Computer Sciences, Northeastern University.  {\tt silver.r@northeastern.edu}.  Supported by NSF awards CNS-2232692 and CNS-2247484.}
\and
Jonathan Ullman\thanks{Khoury College of Computer Sciences, Northeastern University.  {\tt jullman@ccs.neu.edu}.  Supported by NSF awards CNS-2232692 and CNS-2247484.}
}

\date{\today}

\begin{document}
\pagenumbering{Alph}

\maketitle
\pagenumbering{gobble}
\ifnum\withcolors=0
  \renewcommand{\todo}[1]{{{}}}
\fi
\begin{abstract}
    We study person-level differentially private (DP) mean estimation in the case where each person holds multiple samples.  DP here requires the usual notion of distributional stability when \emph{all} of a person's datapoints can be modified. Informally, if $n$ people each have $m$ samples from an unknown $d$-dimensional distribution with bounded $k$-th moments, we show that
    \[n = \tilde \Theta\left(\frac{d}{\alpha^2 m} + \frac{d  }{ \alpha m^{1/2}  \eps} + \frac{d}{\alpha^{k/(k-1)} m  \eps} + \frac{d}{\eps}\right)\]
    people are necessary and sufficient to estimate the mean up to distance $\alpha$ in $\ell_2$-norm under $\varepsilon$-differential privacy (and its common relaxations). In the multivariate setting, we give computationally efficient algorithms under approximate-DP and computationally inefficient algorithms under pure DP, and our nearly matching lower bounds hold for the most permissive case of approximate DP.  Our computationally efficient estimators are based on the standard clip-and-noise framework, but the analysis for our setting requires both new algorithmic techniques and new analyses.  In particular, our new bounds on the tails of sums of independent, vector-valued, bounded-moments random variables may be of interest.
\end{abstract}

\vfill\newpage
\tableofcontents

\newpage
\pagenumbering{arabic}
\section{Introduction}
\label{sec:intro}
 Mean estimation is the following classical statistical task: given a dataset $X_1, \dots, X_n$ of samples drawn i.i.d.\ from some unknown distribution $\cD$, output an estimate of its mean $\E_{X\sim \cD}[X]$.
Without special care, standard procedures for mean estimation can leak sensitive information about the underlying dataset~\cite{HomerSRDTMPSNC08,DworkSSUV15}, so there is now a line of recent work developing methods for mean estimation under the constraint of \emph{differential privacy} (DP)~\cite{DworkMNS06}.

Informally speaking, a DP algorithm is insensitive to any one person's data, but how we formalize the concept of one person's data can depend on the application.\footnote{Throughout, we use ``person'' to denote an entity for whom their entire dataset must have privacy preserved as the base level of granularity in the DP definition, which is often called \emph{user-level privacy} in the literature.  Depending on the application, the entity could naturally be called a person, individual, user, organization, silo, etc.\ without changing the underlying definition of privacy.}  In many applications, we assume that each person's data corresponds to a single element of the dataset, and thus we formalize DP as requiring that the algorithm be insensitive to changes in one sample $X_i$.  This modeling protects one sample in the dataset, regardless of whether or not that captures all of one person's data, and is often called \emph{item-level DP}.  

However, there are many cases where a person can possess \emph{multiple} samples, and we want to protect privacy of all the samples belonging to a single person.  In these cases, we want to define datasets to be neighbors if they differ in exactly one person's set of samples.  For example, systems for \emph{federated learning}~\cite{McMahanMRHY17,KairouzMABBBBCCCDEEEGGGGGGHHHHHHJJJKKKKKLLMMNOPQRRRSSSSSTVWXXYYYZ21} adopt this modeling of privacy.  More precisely, we study a setting where $n$ people each have $m$ samples\footnote{We assume both $m$ and $n$ are publicly known. For simplicity, we focus on the case where each person has exactly $m$ samples, though our general approach applies to heterogeneous cases where individuals have varying amounts of data.} drawn i.i.d.\ from an unknown distribution $\cD$.  In this case we formalize DP as requiring that the algorithm be insensitive to changing all $m$ samples belonging to one person, which we call \emph{person-level DP}.  Note that this generalizes the initial setting, which corresponds to the case $m=1$, and also that there are now $nm$ total samples available.  

Our goal is to design an estimator for $\E_{X \sim \cD}[X]$, with small error in $\ell_2$ distance, subject to person-level differential privacy.  Our work focuses on two primary questions: What algorithms can we use to achieve this goal, and how many people and samples are required?

\subsection{Results and Techniques}
\label{sec:results}
Our main results are upper and lower bounds on the sample complexity of mean estimation with person-level differential privacy, for both univariate and multivariate distributions with \emph{bounded moments}.  Bounded moments are a standard approach for capturing and quantifying what it means for data to be well behaved, and interpolates between distributions that have heavy tails and distributions that are well concentrated (e.g.\ Gaussians).  Specifically, a distribution $\cD$ over $\R^d$ with mean $\mu$ has bounded $k$-th moment if for all unit vectors $v \in \R^d$, 
$$
\E_{X \sim \cD}[|\langle v, X - \mu \rangle|^k] \leq 1.\footnote{The choice of $1$ on the right-hand side is arbitrary, and can be replaced with any other value by appropriate scaling.}
$$
Note that when $d = 1$ this is the standard $k$-th central moment condition and when $d > 1$ this condition says that we can project the distribution onto any line and obtain a distribution with bounded $k$-th central moment.

In all cases, our lower bounds on sample complexity hold for the most permissive notion of \emph{approximate differential privacy.}  Our upper bounds on sample complexity hold even for the most restrictive notion of \emph{pure differential privacy}.  We complement our upper bounds on the sample complexity with a computationally efficient estimator in low-dimensions, or a computationally efficient estimator in high-dimensions satisfying approximate differential privacy.

As a starting point, we begin by presenting our estimator for the univariate case so we can introduce and discuss techniques and new phenomena that arise in the person-level setting.
\begin{theorem}[Informal, see Corollary~\ref{cor:mean-estimation-bounded}]
There is a computationally efficient person-level $\eps$-DP estimator such that for every distribution $\cD$ over $\R$ with mean $\mu$ and bounded $k$-th moments, the estimator takes $m$ samples per person from
\begin{equation*}
    n = \tilde{O} \Paren{\frac{1}{\alpha^2 m} + \frac{1}{\alpha m^{1/2} \eps} + \frac{1}{\alpha^{k/(k-1)} m \eps} + \frac{1}{\eps}}
\end{equation*}
people\footnote{For simplicity, throughout this introduction we elide any assumptions that $|\mu| \leq R$ and any dependence on this parameter $R$, which is necessary for certain variants of DP but not others.} and, with probability at least $2/3$,\footnote{The constant $2/3$ is arbitrary and can be amplified to any constant via standard techniques, at a cost of a constant factor blowup in sample complexity.} outputs an estimate $\hat\mu$ such that $|\hat\mu - \mu| \leq \alpha$.
\label{thm:one-d-pure}
\end{theorem}
Observe that the required number $n$ of people consists of four terms. As the total number of samples across all people is $nm$, the first term corresponds to the classical sample complexity of the problem, sans privacy constraints. The second and third terms are the most interesting.  Loosely speaking, the second term corresponds to the optimal sample complexity of estimating the mean of a Gaussian in our setting.  Since a Gaussian satisfies bounded $k$-th moments for every $k$, this number of samples is a lower bound in our setting.  The third term can be arranged to say that the total number of samples we need is at least $1/\alpha^{k/(k-1)}\eps$, which is the optimal sample complexity for estimating a $k$-th moment distribution under item-level privacy~\cite{BarberD14,KamathSU20}, and is also a lower bound for our setting.  The final term comes from a \emph{coarse estimation step} in the algorithm, whose discussion we defer to the technical sections.  This term is also necessary as any non-trivial DP algorithm requires data from at least $1/\eps$ people.

Next, we consider the multivariate case, where we give an efficient estimator that satisfies approximate differential privacy.\footnote{We note that a similar algorithm can be designed under the stronger notion of \emph{concentrated DP}~\cite{DworkR16,BunS16}, differing primarily in the coarse estimation step (and the corresponding term in the sample complexity).}
\begin{theorem}[Informal, see Theorem~\ref{thm:approx-dp-upperbound}]
\label{thm:approx-dp-upperbound-informal}
There is a computationally efficient $(\eps,\delta)$-DP estimator such that for every distribution $\cD$ over $\R^d$ with mean $\mu$ and bounded $k$-th moments, the estimator takes $m$ samples per person from
\begin{equation*}
    n = \tilde{O}\Paren{\frac{d}{\alpha^2 m} + \frac{d \log^{1/2}(1/\delta)}{\alpha m^{1/2} \eps} + \frac{d\log^{1/2}(1/\delta)}{\alpha^{k/(k-1)}m\eps} + \frac{d\log^{1/2}(1/\delta) + d^{1/2} \log^{3/2}(1/\delta)}{\eps}}
\end{equation*}
people and, with probability at least $2/3$, outputs an estimate $\hat\mu$ such that $\|\hat\mu - \mu\|_2 \leq \alpha$.
\label{thm:high-d-cdp}
\end{theorem}
Observe that, in the first three terms, the sample complexity matches the sample complexity of the univariate case, scaled up by a factor of $d$, and, as we show, all three of these terms are optimal up to polylogarithmic factors.  The final term does not decrease with $m$, meaning we will require $\gtrsim d/\eps$ people no matter how many samples each person has.  However, this final term does not depend on $\alpha$ and becomes dominated when $\alpha$ is sufficiently small, and in this case we provably obtain the optimal sample complexity.

In Section~\ref{sec:to} we give an overview of the algorithms behind Theorems \ref{thm:one-d-pure} and \ref{thm:approx-dp-upperbound-informal} and highlight the main new technical ideas.

Next, we give an algorithm with improved sample complexity and a guarantee of pure DP, at the cost of computational efficiency.
\begin{theorem}[Informal, see Theorem~\ref{thm:mean_est_pure_dp}]
There is a computationally inefficient $\eps$-DP estimator such that for every distribution $\cD$ over $\R^d$ with mean $\left\|\mu\right\| \le R$ and bounded $k$-th moments, the estimator takes $m$-samples per person from
\begin{equation*}
    n = \tilde{O}\Paren{\frac{d}{\alpha^2 m} + \frac{d  }{ \alpha m^{1/2}  \eps} + \frac{d}{\alpha^{k/(k-1)} m  \eps} + \frac{d \log(R)}{\eps}}
\end{equation*}
people and, with probability at least $2/3$, outputs an estimate $\hat\mu$ such that $\|\hat\mu - \mu\|_2 \leq \alpha$.
\label{thm:high-d-pdp}
\end{theorem}

The first three terms in the sample complexity match those in the approximate DP case (Theorem~\ref{thm:high-d-cdp}). The fourth term is again the cost of coarse estimation, and is optimal for pure DP via standard packing lower bounds (see Lemma~\ref{lem:lb-range}).

Our pure-DP for the high-dimensional cases employs the framework of Kamath, Singhal, and Ullman~\cite{KamathSU20}, which reduces multivariate mean estimation to a collection of univariate mean testing problems.  We slightly simplify their framework by viewing the set of testing problems as between a candidate mean and a ``local cover'' of its neighborhood, bringing it more in line with other works on private estimation~\cite{AdenAliAK21}. This reduction allows us to appeal to the techniques we developed for proving Theorem~\ref{thm:one-d-pure}.  We comment that, inheriting the deficiencies of~\cite{KamathSU20}, this algorithm is \emph{not} computationally efficient. Indeed, most algorithms for pure DP estimation in multivariate settings are computationally inefficient, barring a few recent works which rely on semidefinite programming and sum-of-squares optimization~\cite{HopkinsKM22, HopkinsKMN23}. However, the amount of computation and data required by these methods scales poorly as the order of the moment they employ increases.  Since, in general, our algorithms employ higher-order moment information, it is not obvious how to make them computationally efficient, and we leave this open as an interesting open question for future work.\footnote{In particular, even for \emph{item-level DP} where $m=1$, computationally efficient algorithms with nearly optimal sample complexity are only known for the cases of $k=2$ or for Gaussian data.}

Finally, we prove lower bounds on the sample complexity of private mean estimation under $(\eps, \delta)$-DP. Since both pure and concentrated DP imply $(\eps,\delta)$-DP, these lower bounds also apply to algorithms satisfying these variants.

\begin{theorem}[Informal, see Theorem~\ref{thm:main_lb_approx_dp}]
    \label{thm:lower-bound}
    Suppose there is an $(\eps, \delta)$-DP estimator such that for every distribution $\cD$ over $\R^d$ with mean $\mu$ and bounded $k$-th moments, the estimator takes $m$-samples per person from $n$ people and, with probability at least $2/3$, outputs $\|\hat\mu - \mu\|_2 \leq \alpha$.
    Then we have that 
    \[n = \tilde \Omega\left(\frac{d}{\alpha^2 m} + \frac{d}{\alpha m^{1/2} \eps} + \frac{d}{\alpha m^{k/(k-1)} \eps} + \frac{\log (1/\delta)}{\eps}\right).\]
\end{theorem}

This lower bound matches our one-dimensional upper bound (Theorem~\ref{thm:one-d-pure}) in all four terms.
It also matches our multivariate upper bounds (Theorems~\ref{thm:high-d-cdp} and~\ref{thm:high-d-pdp}) in all terms except the last. As previously mentioned, we speculate the former upper bound (for approximate DP) is loose in this fourth term. The latter upper bound for pure DP is necessarily greater in the last term, due to the stronger privacy notion.

The first and third terms follow from a simple reduction from the item-level $m=1$ case: imagine starting with $nm$ samples, splitting them into $n$ batches of $m$ samples, and feeding them into a person-level DP algorithm. If this algorithm required too few people, it would violate known lower bounds for the $m=1$ case. The second term employs a conversion of Levy et al.~\cite{LevySAKKMS21}, which also reduces from the item-level case to person-level case for Gaussian data, based on sufficient statistics. We also give an alternate, direct proof of this term using the fingerprinting technique~\cite{BunUV14,KamathMS22}. The fourth term is again required for any non-trivial DP algorithm.

\subsection{Technical Overview of Theorems \ref{thm:one-d-pure} and \ref{thm:approx-dp-upperbound-informal}}
\label{sec:to}

Our efficient estimators use the following general framework, following \cite{KamathSU20}:
\begin{enumerate}
    \item \textbf{Reduce to item-level privacy.}  For each person $i = 1,\dots,n$, average their $m$ samples to obtain $Y^{(i)}$.  From now on, we will design an algorithm that takes $Y^{(1)},\dots,Y^{(n)}$, satisfies item-level privacy, and solves mean estimation for the distribution of $Y$.  Notice that these random variables now come from the class of \emph{averages of $m$ independent random variables with bounded $k$-th moment,} which is more structured than just bounded $k$-th moment.
    
    \item \textbf{Coarse Estimation.} Start by obtaining a \emph{coarse estimate} $\mu_0$ of the mean satisfying $\| \mu_0 - \mu \|_2 \leq \gamma$ for some suitable $\gamma$ to narrow the search space.

    \item \textbf{Clip-and-Noise.}  Clip each value $Y^{(i)}$ to lie in an $\ell_2$-ball of radius $\rho$ centered at $\mu_0$ for some suitable $\rho$.  Then average these values and add Gaussian noise to this average.  Notice that the noise we have to add to ensure DP will be proportional to the sensitivity of the average, which is bounded by $2\rho/n$ because of clipping.
\end{enumerate}

Although the framework is well known, in the rest of this section we will describe the technical challenges that had to be overcome to apply this general framework in the setting of person-level privacy with data satisfying bounded moments conditions.

For now we will ignore the coarse estimation step and suppress exactly what we need from the coarse estimate $\mu'$, so we can focus on the interplay between the first and third steps.

In the Gaussian case, this method is deceptively simple to analyze.  In the first step, each value $Y$ is itself a Gaussian with smaller variance (i.e., scaled down by a factor of $m$).  Then, in the third step, since Gaussians are very tightly concentrated around their mean, we can clip to a small ball of radius $\smash{\approx d^{1/2}/m^{1/2}}$ and introduce essentially no bias because there is only a very small probability that any individual point gets clipped. Formalizing the rest of this analysis is relatively straightforward.  The case of bounded $k$-th moments is more challenging because the data can have heavy tails, posing a dilemma: either we clip to a relatively small ball and introduce significant bias, or we clip to a relatively large ball and have to add more noise for privacy, leading to an unavoidable bias-variance-tradeoff~\cite{KamathMRSSU23}.  

To quantify this tradeoff and find the optimal balance of bias and variance, we need to understand the bias introduced by clipping the random variables $Y$, as a function of the clipping radius $\rho$.  The standard approach to bounding the bias in one-dimension is more or less to integrate the function $f(t) = \mathbb{P}[|Y-\mu| \geq t]$ over all values of $t \geq \rho$, and thus we can reduce bounding the bias to computing tail bounds.  In \cite{KamathSU20}, which is the case where $m = 1$ and $Y$ is an arbitrary random variable with bounded $k$-th moment, the tightest possible bound on the tails is roughly $1/t^{k}$.  However, things become more complicated in the person-level setting with $m > 1$.

\smallskip\textbf{Technical Challenge 1: Bounding the Bias in One Dimension.}  Using the fact that $Y$ is an average of $m$ independent random variables with $k$-th moment bounded by $1$, we can deduce that $Y$ also has $k$-th moment bounded by roughly $1/m^{k/2}$, which yields a tail bound of roughly $1/m^{k/2} t^{k}$.  However, this tail bound is actually suboptimal in its dependence on $m$.  To see why, observe that as $m \to \infty$, $Y$ converges to a Gaussian by the Central Limit Theorem, which has very thin tails of roughly $\exp(-m t^2)$.  So for large $m$ we'd expect a stronger tail bound that just what we get from assuming a bounded $k$-th moment.  To make this argument rigorous, we apply a strong form of the Central Limit Theorem---\emph{non-uniform Berry-Esseen}~\cite{Michel76}---that gives an optimal tail bound for this class of random variables (see Corollary~\ref{cor:non-uniform-tail-bound}).  Plugging this tail bound into the standard bias arguments and then solving for the optimal tradeoff between bias and variance yields Theorem~\ref{thm:one-d-pure}.

\smallskip\textbf{Technical Challenge 2: Tail Bounds in High Dimensions.} The natural way to generalize this approach to the high-dimensional case involves proving an analogous bound on $\mathbb{P}[\|Y - \mu\|_2 \geq t]$ for $t \geq \rho$.  However, as far as we are aware, the non-uniform Berry-Esseen theorem does not have a generalization to high-dimensional random variables.  To address this challenge, we use a somewhat modified proof strategy that succeeds in high-dimensions and also gives an alternative proof of the optimal tail bound for the one-dimensional case.  Roughly, our novel tail bound states that if $Y$ is an average of $m$ independent random variables with bounded $k$-th moment in $d$ dimensions, then
$$
\forall t \gtrsim \sqrt{\frac{d}{m}} \quad \mathbb{P}\left[ \|Y - \mu\|_2 \geq t \right] \lesssim \exp\left(\frac{-mt^2}{d}\right) + \frac{d^{k/2}}{m^{k-1}t^{k}}.
$$
See Theorem~\ref{thm:highd-tailbound} for a more precise statement. Intuitively, the result says that the probability $\|Y - \mu\|_2 \geq t$ is, roughly, at most the sum of two different probabilities: (1) The probability that a draw from $N(0,\frac{1}{m} \mathbb{I})$ has norm larger than $t$, and (2) the probability that a single random variable $X$ with bounded $k$-th moments has norm larger than $mt$.  Note that (1) is a lower bound on the tails because Gaussians have bounded $k$-th moment for every constant $k$, so $Y-\mu$ might be Gaussian.  Also, (2) is a lower bound on the tails because $Y$ is an average of $m$ random variables with bounded $k$-th moment, and if a single one of these random variables has norm at least $tm$ then there can be a constant probability that the norm of the average is at least $mt$.  The theorem says that these are the two most likely ways the norm can be large.  We believe that bounds on the tails of averages of independent bounded-moments random variables are quite natural and we expect this result will find further applications.

\smallskip\textbf{Technical Challenge 3: Better Bias Bounds in High-Dimensions.}  The natural approach of integrating the function $f(t) = \mathbb{P}[\|Y-\mu\|_2 \geq t]$ for $t \geq \rho$ will yield some bound on the bias introduced by clipping to a ball of radius $\rho$, which has the form $d^{k/2} / (m\rho)^{k-1}$.  However, this bound is not strong enough to prove the sample complexity bounds in Theorem~\ref{thm:approx-dp-upperbound-informal}.  The looseness is because the bound only uses the fact that the random variables have a bound on the $k$-th moment of its norm, which, after suitable rescaling, is a \emph{weaker} condition than assuming that the $k$-th moment is bounded in every individual direction.  We give a novel, tighter bound on the bias introduced by clipping that takes into account the fact that we have control over our random variables in every direction.  An interesting feature of our bound is that it depends heavily on how close the center of the ball used for clipping is to the true mean of the distribution.  Specifically, if $Y$ is a random variable satisfying our assumptions with mean $\mu$, and we clip to a ball of radius $\rho$ around a coarse estimate $\mu_0$ such that $\| \mu_0 - \mu\|_2 \leq \gamma \ll \rho$, then the bias is at most
$$
\frac{d^{(k-1)/2}}{(m\rho)^{k-1}}\left( 1 + \frac{\gamma d^{1/2}}{\rho} \right).
$$
The bound is a bit tricky to interpret.  As a start, assume $\gamma = 0$ so we are clipping to a ball centered around the true mean.  In this case the bias improves over the straightforward bound by a factor of $d^{1/2}$, which is tight and sufficient for our purposes.  We get the same conclusion as long as the ball is only off-center by $\gamma \lesssim \rho/d^{1/2}$.  However, if the ball is off-center by more than $\rho/d^{1/2}$ we only get a weaker bias bound, and we can argue that this is inherent.  Understanding the bias introduced by clipping is fundamental, and so we expect this bias bound will find further applications.

\smallskip\textbf{Technical Challenge 4: Iteratively Improving the Coarse Estimate.}  As the discussion above makes clear, the accuracy of the coarse estimate plays a significant role in determining the bias introduced by clipping.  In order to say that the coarse estimate is accurate enough, we need it to have sufficiently small error $\gamma \lesssim \rho/d^{1/2}$.  In some regimes, we will want to set $\rho \approx d^{1/2}/m^{1/2}$, so we want the coarse estimate to have error roughly $\gamma \approx 1/m^{1/2}$ to be safe.  However, the straightforward way of obtaining a coarse estimate using the approach of \cite{KamathSU20} will not produce such an accurate coarse estimate, and the result will be suboptimal sample complexity that is $\omega(d)$ overall.

To solve this problem, we use a novel iterative approach, inspired by the private preconditioning technique from~\cite{KamathLSU19,BiswasDKU20,CanonneKMUZ20}.  First, we use the simple strategy for producing a coarse estimate with somewhat large error $\gamma_0$.  We then run the clip-and-noise algorithm with $O(d)$ people's data to obtain some new coarse estimate $\mu_1$.  While $\mu_1$ will not itself have the error $\alpha$ that we desire, we can show that $\mu_1$ is a better estimate than $\mu_0$.  Thus, we can restart the clip-and-noise algorithm with a smaller value $\gamma_1$ and get an even better estimate $\mu_2$.  We can analyze this process and show that by iterating this procedure a small number of times we will obtain $\mu_t$ such that $\|\mu_t - \mu\| \leq \gamma_t \approx 1/m^{1/2}$.  As we've discussed above, such a coarse estimate $\mu_t$ is always accurate enough to get a tight bound on the bias of clipping, meaning we can run clip-and-noise a final time to get our last and most accurate estimate.  Combining all these steps yields Theorem~\ref{thm:approx-dp-upperbound-informal}.

\subsection{Related Work}
\label{sec:related}

Mean estimation under differential privacy has been studied extensively, including works which investigate mean estimation for Gaussians~\cite{KarwaV18,BunKSW19,KamathLSU19,BiswasDKU20,DuFMBG20,CaiWZ21,HuangLY21,HopkinsKMN23,BieKS22,BenDavidBCKS23,AumullerLNP23,AlabiKTVZ23,AsiUZ23}, distributions with bounded moments~\cite{BarberD14,KamathSU20,WangXDX20,BrownGSUZ21,KamathLZ22,HopkinsKM22,KuditipudiDH23,BrownHS23}, or both~\cite{BunS19,LiuKKO21,LiuKO22,KamathMRSSU23}. Several works have also focused on the related problem of privately estimating higher-order moments~\cite{AdenAliAK21,KamathMSSU22,AshtianiL22,KothariMV22,KamathMS22,Narayanan23,PortellaH24}. Other works study private mean estimation of arbitrary (bounded) distributions~\cite{BunUV14,SteinkeU15,DworkSSUV15}. However, all of these works focus on differential privacy in the item-level setting (i.e., when neighboring datasets differ in exactly one datapoint).  This is a special case of the more general person-level privacy setting we consider, which corresponds to $m=1$.

Recently, there have been several works focusing on estimation under person-level privacy.  Some focus on mean estimation, but for distributions which are bounded or satisfy certain strong concentration properties~\cite{LevySAKKMS21,GirgisDD22}. These encompass (sub-)Gaussian distributions, but not those with heavier tails that we consider in our work. \cite{NarayananME22} also study mean estimation under a bounded moment condition, but a qualitatively different type of condition than what we consider, and only with a bound on the second moment. Consequently, their work does not capture the rich tradeoff that we do as one varies the number of samples per person $m$ and the number of moments bounded $k$. \cite{CummingsFMT22} focuses on mean estimation in a slightly different setting, where rather than people receiving samples from the distribution of interest, they instead have a latent sample from said distribution and receive several samples from the Bernoulli with that parameter. \cite{GeorgeRST24} study mean estimation of Bernoullis under continual observation~\cite{ChanSS10,DworkNPR10}.  A line of works focuses on generic transformations from algorithms for item-level privacy to algorithm algorithms for person-level privacy~\cite{BunGHILPSS23,GhaziKKMMZ23b}. Even in the one-dimensional case, their results (combined with item-level algorithms of~\cite{KamathSU20}) are unable to recover our results for a couple reasons. First, their transformation only saves a factor of $1/\sqrt{m}$ in the private sample complexity, whereas in several parameter regimes, we save a factor of $1/m$.  Based on details of our analysis, we believe a black-box technique that recovers our result from the item-level case seems unlikely. Furthermore, in general, their transformations are not computationally efficient.  Finally, other works study adjacent learning tasks, such as learning discrete distributions~\cite{LiuSYKR20,ChhorS23,AcharyaLS23}, PAC learning~\cite{GhaziKM21}, and convex optimization~\cite{GhaziKKMMZ23a,BassilyZ23,LiuA24}.

Numerous other statistical estimation tasks have been studied under item-level differential privacy, including learning mixtures of Gaussians~\cite{KamathSSU19,AdenAliAL21,ArbasAL23,AfzaliAL24}, discrete distributions~\cite{DiakonikolasHS15}, graphical models~\cite{ZhangKKW20}, random graph models~\cite{BorgsCS15,BorgsCSZ18a,SealfonU19,ChenDDHLS24}, and median estimation~\cite{AvellaMedinaB19,TzamosVZ20}. For more coverage of the literature on private statistics, see~\cite{KamathU20}.

\subsubsection{Independent work of Zhao et al.~\cite{ZhaoLSLWL24}}
In May 2024, an initial version of this paper and a simultaneous and independent work of Zhao et al.~\cite{ZhaoLSLWL24} were posted on arXiv. Both works study mean estimation under person-level privacy, though considering slightly different notions of what it means for moments to be bounded. We consider a \emph{directional} bound (Definition~\ref{def:bdd-moments}), which roughly says that the moment is bounded in every univariate projection. On the other hand, they consider a \emph{non-directional} bound, which bounds the moment of the $\ell_2$-norm of the random vector, i.e., $\E_{X \sim \cD}\left[\|X - \mu\|_2^k\right]$.%

It is not hard to see that a non-directional bound of $1$ on the $k$-th moment implies a directional bound of $1$ on the $k$-th moment. Similarly, one can show that a directional bound of $1$ implies a non-directional bound of $d^{k/2}$ on the $k$-th moment. Thus, the two moment conditions are polynomially related to each other, and, by rescaling the data, upper and lower bounds for each case have implications for the other. 

Using the latter implication that a directional bound implies a non-directional bound, one can employ algorithms for the non-directional setting in the directional setting. If one uses the algorithm in Theorem 3 of~\cite{ZhaoLSLWL24}, and looks at the implications for the setting with directional moment bounds, one achieves a sample complexity comparable to Theorem 4.1 in the original May 2024 version of this paper. This sample complexity is roughly that of Theorem~\ref{thm:approx-dp-upperbound-informal}, albeit with an extra dimension-dependent factor in the third term. However, by working directly with the directional moment bound as we do, one can achieve the improved sample complexity in Theorem~\ref{thm:approx-dp-upperbound-informal}, which was done subsequent to the initial posting of our two papers. While our algorithms under our directional bound imply algorithms under their non-directional bound, the resulting sample complexities are loose by dimension-dependent factors.

Their work only considers algorithms, and not lower bounds. There may be implications of our lower bound arguments to show lower bounds for their non-directional moment bound setting. 

\section{Preliminaries}
\label{sec:prelim}

\subsection{Privacy Preliminaries}
We say that two datasets $X$ and $X'$ are \textit{neighboring datasets} if they differ in at most one datapoint, i.e. $d_{H}(X,X') \le 1$ where $d_{H}$ denotes the Hamming distance. We introduce the definition of differential privacy (which is also referred to as item-level differential privacy).
\begin{definition}[Item-Level Differential Privacy]
    Let $\mathrm{A} : \mathcal{X}^n \rightarrow \mathcal{Y}$. We say that $\mathrm{A}$ satisfies item-level $(\eps,\delta)$-differential privacy (DP) if, for all neighborhing datasets $X,X' \in \mathcal{X}^n$ and for all $Y \in \mathcal{Y}$,
    \[
        \Pr\Brac{\mathrm{A}(X) \in Y} \le e^{\eps}\Pr\Brac{\mathrm{A}(X') \in Y} + \delta.
    \]
\end{definition}
It is reasonable to consider $\eps, \delta \in [0,1]$. We remark that $(\eps,0)$-DP is commonly referred to as \textit{pure differential privacy}, and $(\eps,\delta)$-DP is commonly referred to as \textit{approximate differential privacy}. In this paper, we focus on a generalized definition of differential privacy known as \textit{person-level differential privacy}. For this definition, we consider datasets in $\mathcal{U}^n$ where $\mathcal{U} = \mathcal{X}^m$. We call each $X^{(i)} = (X^{(i)}_1,\ldots,X^{(i)}_m)\in \mathcal{U}$ a \textit{person} and call each $X^{(i)}_j \in \mathcal{X}$ a \textit{sample} (since in all our results these points will be chosen as independent samples from a probability distribution). In the definition of person-level differential privacy, we say that two datasets $X,X' \in \mathcal{U}^n$ are \textit{neighboring datasets} if they differ in at most one one person's samples--that is, if $d_{H}(X,X') = \sum_{i=1}^n \Ind{(X^{(i)} \ne X'^{(i)})} \le 1$.
\begin{definition}[Person-Level Differential Privacy]
    Let $\mathcal{U} = \mathcal{X}^m$, and let $\mathrm{A} : \mathcal{U}^n \rightarrow \mathcal{Y}$. We say that $\mathrm{A}$ satisfies person-level $(\eps,\delta)$-differential privacy (DP) if, for all neighboring datasets $X,X' \in \mathcal{U}^n$ and for all $Y \in \mathcal{Y}$,
    \[
        \Pr\Brac{\mathrm{A}(X) \in Y} \le e^{\eps}\Pr\Brac{\mathrm{A}(X') \in Y} + \delta.
    \]
\end{definition}
We now introduce some standard lemmas for differentially private algorithms.\footnote{The specific formulations of some of these lemmas are inspired by those of \cite{KamathSU20}.} We first introduce a lemma surrounding a property known as \textit{post-processing}.
\begin{lemma}[Post-Processing]
    Let $\mathsf{A}:\mathcal{X}^n\rightarrow\mathcal{Y}$ be a person-level $(\eps,\delta)$-DP algorithm. Let $\mathsf{B}:\mathcal{Y} \rightarrow \mathcal{Z}$. The mechanism $\mathsf{B}\circ\mathsf{A}$ satisfies person-level $(\eps,\delta)$-DP.
    \label{lem:postprocessing}
\end{lemma}
The next two lemmas deal with a property called \textit{composition}. For these lemmas, we define $\mathsf{A_1},\ldots,\mathsf{A_k}$ to be a sequence of mechanisms: $\mathsf{A_1}:\mathcal{X}^n\rightarrow \mathcal{Y}_1$ and, for all $t \in \{2,\ldots,k\}$, $\mathsf{A_t}:\mathcal{X}^n\times\mathcal{Y}_1\times\cdots\times\mathcal{Y}_{t-1} \rightarrow \mathcal{Y}_t$.
\begin{lemma}[Basic Composition]
\label{lem:basic-composition}
    Suppose $\mathsf{A_1}$ satisfies person-level $(\eps_1,\delta_1)$-DP. Similarly, for each $t \in \{2,\ldots,k\}$, suppose that, under every fixed $y_1 \in \mathcal{Y}_1,\ldots,y_{t-1} \in \mathcal{Y}_{t-1}$, $\mathsf{A_t}(\cdot,y_1,\ldots,y_{t-1})$ satisfies person-level $(\eps_t,\delta_t)$-DP. It follows that the mechanism $\mathsf{A}_k$ which runs each of $\mathsf{A}_1,\ldots,\mathsf{A}_{k-1}$ in sequence satisfies person-level $\Paren{\sum_{i=1}^k\eps_i,\sum_{i=1}^k\delta_i}$-DP.
    \label{lem:privacy-composition}
\end{lemma}
\begin{lemma}[Advanced Composition]
\label{lem:advanced-composition}
    Suppose $\mathsf{A_1},\ldots,\mathsf{A_k}$ are person-level $(\eps_0,\delta_1),\ldots,(\eps_0,\delta_k)$-DP respectively for some $\eps_0 \le 1$. Then, for all $\delta_0 > 0$, the mechanism $\mathsf{A_k}$ which runs each of $\mathsf{A}_1,\ldots,\mathsf{A}_{k-1}$ in sequence satisfies person-level $(\eps,\delta)$-DP for $\eps = \eps_0\sqrt{6k\log(1/\delta_0)}$ and $\delta = \delta_0 + \sum_{t=1}^k \delta_t$.
\end{lemma}

We now state a few definitions and lemmas related to standard private mechanisms.%
\begin{definition}
    Let $f: \mathcal{U}^n \rightarrow \R^d$ be a function, where $\mathcal{U} = \mathcal{X}^m$. The person-level $\ell_2$-sensitivity (denoted $\Delta_{f,2}$) is defined as
    \[
        \Delta_{f,2} \coloneqq \max_{\substack{X,X' \in \mathcal{U}^n\\d_{H}(X,X') \le 1}}\normt{f(X) - f(X')}.
    \]
\end{definition}

\begin{lemma}[Gaussian Mechanism]
    Let $f : \mathcal{U}^n \rightarrow \R^d$ be a function with $\ell_2$-sensitivity $\Delta_{f,2}$, where $\mathcal{U} = \mathcal{X}^m$. Let $S \in \mathcal{U}^n$, and let
    \[
        W \sim \mathcal{N}\Paren{0,\Paren{\frac{\Delta_{f,2}\sqrt{2\ln (2/\delta)}}{\eps}}^2\cdot \mathbb{I}_{d \times d}}.
    \]
    Then the mechanism
    \[
        M(X) = f(X) + W
    \]
    satisfies person-level $(\eps,\delta)$-DP.
    \label{lem:gaussian-mechanism}
\end{lemma}

\begin{lemma}[Private Histograms]
    \label{lem:private-histograms}
    Let $(S_1,\ldots,S_n)$ be samples in some data universe $\Omega$, and let $U = \{h_u\}_{u \subset \Omega}$ be a collection of disjoint histogram buckets over $\Omega$. Then, we have person-level $\eps$-DP and $(\eps,\delta)$-DP histogram algorithms with the following guarantees.
    \begin{itemize}
        \item $\eps$-DP: $\ell_{\infty}$ error $O\Paren{\frac{\log (|U|/\beta}{\eps}}$ with probability at least $1-\beta$; run time $\poly(n,\log(|U|/(\eps\beta)))$,
        \item $(\eps,\delta)$-DP: $\ell_{\infty}$ error $O\Paren{\frac{\log (1/\delta\beta)}{\eps}}$ with probability at least $1-\beta$; run time $\poly(n,\log(|U|/(\eps\beta)))$.
    \end{itemize}
\end{lemma}

\begin{lemma}[Exponential Mechanism~\cite{McSherryT07}]
\label{lem:exp_mech}
The exponential mechanism $M_{\eps, S, \mathrm{Score}}(X)$ takes a dataset $X \in \cX^n$, computes a score ($\mathrm{Score} : \cX^n \times S \to \bR$) for each $p \in S$ with respect to $X$, and outputs $p \in S$ with probability proportional to $\exp\paren{\frac{\eps \cdot \mathrm{Score}(X, p)}{2 \cdot \Delta_{\mathrm{Score}, 1}}}$, where:
\[
    \Delta_{\mathrm{Score},1} \coloneqq \max_{p \in S} \max_{X \sim X' \in \cX^n} \abs{\mathrm{Score}\paren{X, p} - \mathrm{Score}\paren{X', p}}.
\]
It satisfies the following properties:
\begin{enumerate}
    \item $M$ is person-level $\eps$-differentially private.

    \item Let $\mathrm{OPT}_{\mathrm{Score}}\paren{X} = \max\limits_{p \in S} \{\mathrm{Score}\paren{X, p}\}$. Then
    \[
        \Pr\left[\mathrm{Score}(X, M_{\eps, S, \mathrm{Score}}(X)) \le \mathrm{OPT}_{\mathrm{Score}}(X) - \frac{2 \Delta_{\mathrm{Score}, 1}}{\eps} (\ln(|S|) + t)\right] \le e^{-t}.
    \]
\end{enumerate}
\end{lemma}

\subsection{Bounded Moments}
\begin{definition}
\label{def:bdd-moments}

    Let $\cD$ be a distribution over $\R$ with mean $\mu$. The $k$-th moment $\sigma_k(\cD)$ of $\cD$ is defined as 
    \[
        \sigma_k(\cD) \coloneqq \E_{X\sim \cD}\Brac{\Abs{X-\mu}^k}^{1/k}.
    \]
    For $\cD$ over $\R^d$ with mean $\mu$, $\sigma_k(\cD)$ is defined as 
    \[
        \sigma_k(\cD) \coloneqq \sup_{v \in \mathbb{S}^{d-1}}\E_{X\sim\cD}\Brac{\Abs{\langle X-\mu,v\rangle}^k}^{1/k}.
    \]
\end{definition}
\begin{lemma}
    Let $\cD$ be a distribution over $\R$ with mean $\mu$ and $\sigma_k(\cD) \le 1$. Let $X_1,\ldots,X_m \sim \cD$. Let $X = \frac{1}{m}(X_1 + \cdots + X_m)$. It follows that 
    \[
        \E[|X - \mu|^k]^{1/k} = O\Paren{\frac{\sigma_k(\cD)}{\sqrt{m}}}.
    \]
    \label{lem:kth-moment-batches}
\end{lemma}

\section{Private Mean Estimation in One Dimension}\label{sec:one-dim}

In this section, we describe an efficient algorithm for estimating means of univariate distributions with bounded $k$-th moments under person-level differential privacy. The algorithm is based on \textit{Clip-and-Noise} and consists of two phases: 
\begin{enumerate}
    \item Coarse Estimation: Find a coarse estimate of the mean with only a few samples.  
    \item Fine Estimation: Given a coarse estimate of the mean, compute the mean of the batches, truncate them to a radius around the coarse estimate, compute the mean of the truncated means and output the result with additional noise. 
\end{enumerate}

\subsection{Preliminaries}
In our analyses of this section (and also in \Cref{sec:hd}) we will treat $k$ as a constant. We define the following notation to denote that the inequality holds up to multiplicative factors that depend only on $k$.

\begin{definition}[Hiding $k$: $\le_k$]  
We write $f(x) \le_k g(x)$, if and only if there exists a function $h(k) \ge 0$ such that $f(x) \le h(k) g(x)$. 
\end{definition}

We will make use of the Laplace mechanism which uses additive noise. We use the following tail bound on the Laplace distribution to bound the noise term.
\begin{fact}[Laplace Distribution Tail Bound]
\label{fact:laplace-tail-bound}
For a Laplace distribution with parameter b we have that
\begin{equation*}
\Pr\brac{\Abs{\cL\Paren{b}} \ge t b} \le \exp\paren{-t}.
\end{equation*}
\end{fact}

In order to reduce sensitivity we use a truncation operation defined as follows.
\begin{definition}[Clipping Operation $\mathrm{Trunc}_{\ell, r}$]
For two numbers $\ell \le r \in \R$ we define the clipping operation to be
\begin{equation*}
\mathrm{Trunc}_{\ell, r} (x):= 
\begin{cases}
r & x \ge r \\
l & x \le l \\
x & \text{\textit{o.w.}}
\end{cases}.
\end{equation*}
\end{definition}
In the analysis of the algorithm, we use the following fact about truncated random variables.
\begin{lemma}[Variance Decreases After Truncation]
\label{lem:second-moment-after-truncation}
Suppose $X$ is a random variable in $\R$ and $Z = \mathrm{Trunc}_{\ell, r}(X)$, where $\mathrm{Trunc}_{\ell, r}$ is the truncation operation. Then
\begin{equation*}
\E\Brac{\Abs{Z - \E\brac{Z}}^2} \le \E\Brac{\Abs{X - \E\brac{X}}^2}.
\end{equation*}
\end{lemma}
\begin{proof}
First note that for any random variable $Y$ we have that:
\begin{equation*}
    \E\Brac{\Abs{Y - \E\brac{Y}}^2} = \frac{1}{2} \E\Brac{\Abs{Y - Y'}^2},
\end{equation*}
where $Y'$ is an independent copy of $Y$.
Now,
\begin{align*}
\E\Brac{\Abs{Z - \E\brac{Z}}^2}
= 
\frac{1}{2} \E\Brac{\Abs{Z - Z'}^2}
= 
\frac{1}{2} \E\Brac{\Abs{\mathrm{Trunc}_{\ell ,r} (X) - \mathrm{Trunc}_{\ell ,r}(X')}^2}
&\le
\frac{1}{2} \E\Brac{\Abs{X - X'}^2} \\
&= \E\Brac{\Abs{X - \E\brac{X}}^2},
\end{align*}
where the last inequality relies on the fact that the distance between $X$ and $X'$ cannot increase post-truncation.
\end{proof}

Finally, in our analysis we will use the non-uniform Berry-Esseen tail bound for the average of $m$ \iid samples from a distribution with bounded $k$-th moments to bound the bias caused by truncation. See \Cref{sec:recreating-nube} for discussion and proof.

\begin{corollary}[Tail Bound for Averages of Bounded-Moments Variables \cite{Michel76}]
\label{cor:non-uniform-tail-bound}
Let $k\ge 3$. Assume $X$ is a distribution with mean $0$ and $k$-th moment bounded by $1$, and $X_i$'s are $m$ \iid copies of $X$. Then,
\begin{equation*}
\forall t \geq \sqrt{\frac{\paren{k -1} \log m}{m}} \quad \Pr\Brac{\frac{1}{m} \sum_{i=1}^{m} X_i \ge t} \le O\left(\frac{1}{m^{k-1} t^{k}}\right) \mper
\end{equation*}
\end{corollary}

\subsection{Generic Clip and Noise Theorem}

In this section we create a generic theorem that applies the clip-and-noise framework for private mean estimation, given access to tail bounds for a univariate distribution, a coarse estimate of its mean, and an upper bound on its variance. The algorithm is as follows: take $n$ \iid samples from the distribution, take their average, clip the average, and add Laplace noise in order to output an $\eps$-DP mean estimation. \Cref{lem:clipping-bias} bounds the error caused by the clipping operation. \Cref{lem:sampling-bias-noise-term} bounds the error caused by the addition of Laplace noise and the sampling error from the truncated distribution. Finally, putting together the clipping bias, sampling bias, and noise term guarantees, \Cref{thm:coarse-to-fine} gives a guarantee on the accuracy of the clip-and-noise framework, given distributions with arbitrary tail bounds.

\begin{lemma}[Clipping Bias]
\label{lem:clipping-bias}
Let $X \in \R$ be a random variable with finite mean.
For two numbers $l \le r$, let $Z \coloneqq \mathrm{Trunc}_{l, r}(X)$.
Then, we have:
\[
    \Abs{\E\brac{X} - \E\brac{Z}} \le \int\limits_{- \infty}^{\ell} \mathrm{LeftTail}(x) \, dx + \int\limits_r^{+ \infty} \mathrm{RightTail}(x) \, dx,
\]
where $\mathrm{RightTail}(x) \coloneqq \Pr[X \geq x]$ and $\mathrm{LeftTail}(x) \coloneqq \Pr[X \le x]$.
\end{lemma}

\begin{proof}
We observe that $Z = \ell \mathbb{1}\{X < \ell\} + X \mathbb{1}\{\ell \le X \le r\} + r \mathbb{1}\{X > r\}$.
Using the above, as well as a combination of the triangle inequality and Jensen's inequality for the function $\phi(x) = \Abs{x}$, the clipping bias can be upper-bounded by:
\begin{align*}
    \Abs{\E\brac{Z} - \E\brac{X}} &= \Abs{\E\brac{(\ell - X) \mathbb{1}\{X \le \ell \}} + \E\brac{(r - X) \mathbb{1}\{X \geq r \}}} \\
    &\le \E\brac{(\ell - X) \mathbb{1}\{X \le \ell\}} + \E\brac{(X - r) \mathbb{1}\{X \geq r\}}.
\end{align*}
We note now that the individual terms in the above sum involve non-negative random variables, so we have:
\begin{align*}
    \Abs{\E\brac{Z} - \E\brac{X}} &\le \int\limits_0^{+ \infty} \Pr[(\ell - X) \mathbb{1}\{X \le \ell\} \geq x] \, dx + \int\limits_0^{+ \infty} \Pr[(X - r) \mathbb{1}\{X \geq r\} \geq x] \, dx \\
    &= \int\limits_0^{+ \infty} \Pr[X \le \ell - x] \, dx + \int\limits_0^{+ \infty} \Pr[X \geq r + x] \, dx \\
    &= \int\limits_{- \infty}^{\ell} \Pr[X \le x] \, dx + \int\limits_r^{+ \infty} \Pr[X \geq x] \, dx,
\end{align*}
where the last equality relies on the changes of variables $x \to \ell - x$, and $x \to x - r$, respectively.
\end{proof}

Often the form of the tail bounds we have on distributions are of the form $\Pr[X \le \mu - t] \le \cdots$, or $\Pr[X \ge \mu + t] \le \cdots$, where $\mu$ is the mean of $X$. The following corollary adapts \Cref{lem:clipping-bias}
to this setting where we have access to a coarse estimate of the mean.

\begin{corollary}[Clipping Bias Given Coarse Estimate]
\label{cor:clipping-bias-coarse-estimate}
In \Cref{lem:clipping-bias} suppose that we are given access to a coarse estimate $\mu_{\texttt{coarse}}$ such that $\Abs{\mu_{\texttt{coarse}} - \mu} \le u$, and we set $\ell = \mu_{\texttt{coarse}} - \rho, r = \mu_{\texttt{coarse}} + \rho $. Then
\begin{equation*}
\Abs{\E\brac{X} - \E\brac{Z}} \le \int\limits_{\rho - u}^{+\infty}
\Pr\Brac{\Abs{X - \mu} \ge t}
\diff t.
\end{equation*}
\end{corollary}
\begin{proof}
For the first term we have
\begin{align*}
\int\limits_{-\infty}^{\mu_{\texttt{coarse}} - \rho} \mathrm{LeftTail}\paren{x} \diff x 
=
\int\limits_{-\infty}^{\mu_{\texttt{coarse}} - \rho} \Pr\brac{X \le x} \diff x
\le
\int\limits_{-\infty}^{\mu - \rho + u} \Pr\brac{X \le x} \diff x 
&=\int\limits_{0}^{+\infty}
\Pr\brac{X \le \mu - \rho + u - t} \diff t \\
&= 
\int\limits_{0}^{+\infty} 
\Pr\Brac{
X - \mu \le -\paren{\rho - u +t}} \diff t.
\end{align*}
Similarly,
\begin{equation*}
\int\limits_{\mu_{\texttt{coarse}} + \rho}^{+\infty} \mathrm{RightTail}\paren{x} \diff x 
\le
\int\limits_{0}^{+\infty}
\Pr\brac{X - \mu \ge \rho - u + t} \diff t.
\end{equation*}
Putting these together finishes the proof.
\end{proof}

Next, let's analyze the sampling bias for the truncated distribution and the error induced by the additive Laplace noise.

\begin{lemma}[Sampling Bias and Noise Term]
\label{lem:sampling-bias-noise-term}
Suppose $\eps > 0$, and a random variable $Z$ supported on $[\ell , r]$ is given, with variance $\Var\brac{Z}$. Suppose $n$ \iid copies of $Z$, $Z_i$'s are given. Let $2 \rho = r - \ell.$ and  $\hat{\mu} = \frac{Z_1+\cdots +Z_n}{n} + \mathcal{\cL}\paren{\frac{2 \rho}{n \eps}}$, where $\cL$ is the Laplace distribution. Then 
\begin{equation*}
\Pr\Brac{\Abs{\hat{\mu} - \E\Brac{Z}} \ge
\Omega \Paren{
\sqrt{\frac{\Var\brac{Z}}{n \beta}} + \frac{\rho \log\paren{1/\beta}}{n\epsilon}
}
}
\le 
\beta.
\end{equation*}

\end{lemma}
\begin{proof}

Note that $\hat{\mu} - \E\brac{Z} = 
\frac{Z_1 + \cdots + Z_n}{n} - \E\brac{Z} + \mathcal{L}\paren{\frac{\rho}{n \epsilon}}$. 
From Chebyshev's inequality we know 

\begin{equation*}
\forall t\ge 0: \quad \Pr\Brac{\Abs{\frac{Z_1 + \dots Z_n}{n} - \E\Brac{Z}}\ge t} \le \frac{\Var\brac{Z}}{nt^2} \mper
\end{equation*}
From tail bounds for the Laplace distribution (\Cref{fact:laplace-tail-bound})
\begin{equation*}
\forall t \ge 0: \quad
\Pr\Brac{\Abs{\mathcal{L} \Paren{\frac{2 \rho}{n \epsilon}}} \ge t} 
\le \exp\Paren{- \frac{t n \eps}{2 \rho}} \mper
\end{equation*}
Therefore,
\begin{align*}
\Pr\Brac{\Abs{\hat{\mu} - \E\brac{Z}} \ge t }
&\le
\Pr\Brac{\Abs{\frac{Z_1 + \cdots + Z_n}{n} - \E\brac{Z}}\ge t/2} + \Pr \Brac{\Abs{\cL\Paren{\frac{2 \rho}{n \eps}}} \ge t/2 } \\
& \le \frac{4\Var\brac{Z}}{nt^2} + 
\exp\paren{-\frac{tn\epsilon}{4\rho}}
\mper
\end{align*}
Therefore,
\begin{equation*}
\Pr\Brac{\Abs{\hat{\mu}
 -\E\Brac{Z}} \ge 
\sqrt{\frac{8 \Var\brac{Z}}{n \beta}} + \frac{4 \rho \log\paren{2/\beta}}{n\epsilon}} \le \beta \mper
\end{equation*}
\end{proof}

Together \Cref{lem:sampling-bias-noise-term} and \Cref{cor:clipping-bias-coarse-estimate} give us a recipe for estimating the mean for arbitrary distributions with known tail bounds, and variances.

\begin{theorem}[Coarse to Fine Estimate]
\label{thm:coarse-to-fine}
Suppose $X$ is a distribution with variance $\Var\brac{X}$.
Suppose a coarse estimate of $\E\Brac{X}$, the true mean is given with accuracy $u$. Then truncation to a radius of $\rho / 2$ around this coarse estimate gives an $\eps$-dp estimate of $\mu$ with success probability $1 - \beta$, and accuracy
\begin{equation*}
\alpha = \cO\Paren{\sqrt{\frac{\Var\brac{X}}{n\beta}}
+ \frac{\rho \log\paren{1/\beta}}{n \eps}}
+
\int\limits_{\rho - u}^{+\infty}
\Pr\Brac{\Abs{X-\mu} \ge t} \diff t.
\end{equation*}
\end{theorem}
\begin{proof}
Putting together \Cref{lem:sampling-bias-noise-term}, and \Cref{cor:clipping-bias-coarse-estimate}, and noting that variance decreases after truncation by \Cref{lem:second-moment-after-truncation}.
\end{proof}

\newcommand{\ind}[1]{\mathbb{I}{\Paren{#1}}}
\subsection{Coarse Estimation}

In this section we give an algorithm for obtaining a coarse estimate of the mean up to accuracy $O\paren{1/\sqrt{m}}$, under both pure and approximate differential privacy, given access to $n$ people, each taking $m$ samples.

\begin{algorithm}[htb]
    \caption{One-Dimensional Range Estimator}\label{alg:user-level-coarse-estimate}
    \hspace*{\algorithmicindent} \textbf{Input:} Samples $X = \Set{X^{(i)}_j}_{i \in [n], j \in [m]}$ where each $X^{(i)}_j \in \R$. Parameters $\eps, \delta, r, R$.\\
    \hspace*{\algorithmicindent} \textbf{Output:} $\mu_{\texttt{coarse}} \in \R$.
    \begin{algorithmic}[1]
    \Procedure{$\textrm{RangeEstimator}$}{$X$; $\eps,\delta,r,R$}
        \State For each $i \in [n]$, do $S_i = \frac{1}{m}(X^{(i)}_1 + \cdots + X^{(i)}_m)$.
        \State Divide $[-R-2r,R+2r)$ into buckets of width $r$: $[-R-2r,-R-r),\ldots,[-r,0),[0,r),\ldots,[R+r,R+2r)$.
        \State If $\delta = 0$, run $\eps$-DP Histogram, and if $\delta > 0$, run $(\eps, \delta)$-DP Histogram from \Cref{lem:private-histograms} using samples $(S_1,\ldots,S_n)$ over the buckets defined above.
        \State Return $\mu_{\texttt{coarse}} = \frac{a+b}{2}$, where $[a,b]$ is the bucket with the largest count.
    \EndProcedure
    \end{algorithmic}
\end{algorithm}

\begin{theorem}[Person-Level Coarse Estimation]
\label{thm:user-level-coarse-esitmation}
    For all people $i \in [n]$ and number $j \in [m]$ of samples per person, let $X = \{X^{(i)}_j\}_{i,j}$. Let $r > 0$. Then for all $\eps > 0$, $\delta \in [0,1)$, the algorithm $\textrm{RangeEstimator}(X;\cdot,\cdot,\cdot,\cdot)$ satisfies person-level $(\eps, \delta)$-DP. In addition, let $P$ be a distribution over $\R$ with mean $\mu \in [-R,R]$ and $k$-th moment bounded by $1$. Suppose each $X^{(i)}_j \sim P$ \iid. For all $\eps > 0$, $ \beta \in (0,1)$, and $16^{1/k}/\sqrt{m} < r < R$, there exists
    \[
        n_0 = O\Paren{\frac{\log(1/\beta)}{\log (\sqrt{m}r)} + \frac{\log(R/(r\beta))}{\eps}}
    \]
    such that, for all $n \ge n_0$,
    $\mu_{\texttt{coarse}} \leftarrow \textrm{RangeEstimator}(X;\eps,0,r,R)$
    satisfies
    \[
        |\mu - \mu_{\textrm{coarse}}| < 2r
    \]
    with probability at least $1 - \beta$. Moreover, for all $\eps > 0, \delta \in (0,1), \beta \in (0,1)$, and $16^{1/k} / \sqrt{m} < r < R$, there exists
    \[
        n_1 = O\Paren{\frac{\log(1/\beta)}{\log (\sqrt{m}r)} + \frac{\log(1/(\delta\beta))}{\eps}}
    \]
    such that, for all $n \ge n_1$,
    $\mu_{\texttt{coarse}} \leftarrow \textrm{RangeEstimator}(X;\eps,\delta,r,R)$
    satisfies
    \[
        |\mu - \mu_{\textrm{coarse}}| < 2r
    \]
    with probability at least $1 - \beta$.
\end{theorem}

In the proof of \Cref{thm:user-level-coarse-esitmation}, we use the following Chernoff bound:
\begin{theorem}[Chernoff]
    Let $Y_1,\ldots,Y_n$ be $0$-$1$ random variables. Let $Y = Y_1 + \ldots + Y_n$, and let $\mu = \E[Y]$. Then for all $t > 1$,
    $$\Pr\Brac{Y > t\mu} = \frac{1}{t^{\Omega\Paren{t\mu}}}.$$
    \label{thm:chernoff}
\end{theorem}
\begin{proof}[Proof of \Cref{thm:user-level-coarse-esitmation}]
Privacy follows by running the $(\eps, \delta)$-DP Histogram (\Cref{lem:private-histograms}) and by post-processing (\Cref{lem:postprocessing}). We focus on analyzing the accuracy.

Let $[a,b]$ be the bucket such that $\mu_{\textrm{coarse}} \in [a,b]$ (that is, $[a,b]$ is the bucket with the largest noisy count). To show that $\Abs{\mu - \mu_{\texttt{coarse}}} < 2r$, it suffices to show that $\max(\Abs{a-\mu},\Abs{b-\mu}) < 2r$, or equivalently, the bucket with the largest count resides within $\mu \pm 2r$. In proving this, we show that the each of the following two events hold with probability at least $1-\beta/2$: (1) A large fraction of the samples resides within $\mu \pm 2r$, and (2) no bucket will have too large a magnitude of additive Laplace or Gaussian noise.

We start by showing (1). In particular, we show that, with probability at least $1-\beta/2$, a $(1-\frac{1}{16})$-fraction of the samples will lie within 
$\mu \pm r$. For all $i \in [n]$, let $Y_i = \ind{\Abs{X_i-\mu} > r}$, let $p = \Pr[Y_i = 1]$, and let $Y = Y_1 + \ldots + Y_n$. Note that $\E[Y] = np$ and
\[
    p = \Pr[\Abs{X_i-\mu} > r] = \Pr[\Abs{X_i-\mu}^k > r^k] \le \Paren{\frac{1}{\sqrt{m}r}}^k
\]
where the inequality follows from Markov's inequality and \Cref{lem:kth-moment-batches}. To show (1), we want to bound $\Pr\Brac{Y > \frac{n}{16}}$. To do this, we can assume that $p$ is as maximal as possible (i.e. when $p = \Paren{1/(\sqrt{m}r)}^k$) as this maximizes $\Pr\Brac{Y > \frac{n}{16}}$. Thus,
\begin{align*}
    \Pr\Brac{Y > \frac{n}{16}} 
    &= \Pr\Brac{Y > \frac{(\sqrt{m}r)^k\E[Y]}{16}}\\
    &= \Paren{\frac{16}{(\sqrt{m}r)^k}}^{\Omega\Paren{(\sqrt{m}r)^k\E[Y]/16}} \tag{by \Cref{thm:chernoff}}\\
    &= \Paren{\frac{16}{(\sqrt{m}r)^k}}^{\Omega\Paren{n/16}}\\
    &= \Paren{\frac{16}{(\sqrt{m}r)^k}}^{\Omega\Paren{\log(\beta^{-1})/\log ((\sqrt{m}r)^k/16)}}\\
    &= \Paren{\frac{1}{2}}^{\Omega\Paren{\log \beta^{-1}}}\\
    & \le \beta/2.
\end{align*}
Note that there are at most $3$ buckets which satisfy $\max(\Abs{a-\mu},\Abs{b-\mu}) < 2r$. Thus, the largest of these $3$ buckets must have at least $(15n/16)/3 = 5n/16$ points with probability at least $1-\beta/2$.

To show (2), we can directly apply \Cref{lem:private-histograms}. If $n_0 = O\Paren{\log(R/(r\beta))/\eps}$, then the $\ell_{\infty}$-error induced by the Laplace noise is at most $n / 16$ with probability at least $1 - \beta/2$. Similarly, If $n_1 = O\Paren{\log(1/\beta\delta) / \eps}$, 
then the $\ell_{\infty}$-error induced by the Gaussian noise is at most $n / 16$ with probability at least $1 - \beta/2$.
Thus, by a union bound,
\[
    \Pr[|\mu - \mu_{\textrm{coarse}}| > 2r] \le \beta/2 + \beta/2 \le \beta.
\]
\end{proof}

\newcommand{\Sbar}{\overline{S}_m}

\subsection{Applying Clip and Noise to Distributions with Bounded \texorpdfstring{$k$}{k}-th Moments}

In this section we use non-uniform Berry-Esseen and a coarse estimate of the mean to find a fine estimate of the mean.
We assume $\Sbar$ is a random variable corresponding to the mean of $m$ samples, and apply Berry-Esseen to obtain tail bounds for this distribution and then use \Cref{thm:coarse-to-fine}.

\begin{lemma}[Coarse to Fine for Bounded $k$-th Moments]
\label{lem:coarse-to-fine-bounded}
Assume $\rho \ge u + \sqrt{\frac{\paren{k-1} \log m}{m}}$. 
Suppose $X$ is a distribution with mean $\mu$ and $k$-th moment bounded by $1$. Let $\Sbar$ be the random variable corresponding to the mean of $m$ \iid samples from $X$.
Suppose we have access to a coarse estimate of the mean with accuracy $u$. Then there exists an $\eps$-DP estimator of the mean that takes samples from the the sample mean distribution truncates them within a radius of $\rho$ around the coarse estimate, and outputs an estimate of the mean with success probability $1 - \beta$ and accuracy
\begin{equation*}
\alpha  = 
O_k \Paren{\sqrt{\frac{1}{mn\beta}}
+
\frac{\rho \log\paren{1/\beta}}{n \epsilon}
+ 
m^{-k+1} \cdot \paren{\rho - u}^{-k+1}},
\end{equation*}
where $\cO_k$ hides factors that only depend on $k$.
\end{lemma}
\begin{proof}
We aim to apply \Cref{thm:coarse-to-fine}, and note that by \Cref{cor:non-uniform-tail-bound} we have that for all $t \ge \sqrt{\frac{\paren{k-1} \log m}{m}}$,
\begin{equation*}
\Pr\Brac{\Abs{\Sbar - \mu} \ge t} \le_k  m^{-k+1} t^{-k}.
\end{equation*}
Therefore, since $\rho - u \ge \sqrt{\frac{\paren{k-1} \log m}{m}}$, we have that
\begin{align*}
\int\limits_{\rho - u}^{+\infty}
\Pr\Brac{\Abs{\Sbar - \mu} \ge t} \diff t
&\le_k
\int\limits_{\rho - u}^{+\infty}
m^{-k+1} t^{-k} \diff t \\
&= 
m^{-k+1} \int\limits_{\rho - u}^{+\infty} t^{-k} \diff t \\
&=
m^{-k+1} \frac{t^{-k+1}}{-k+1} \Big|_{\rho - u}^{+\infty} \\
&=
m^{-k+1} \frac{\paren{\rho - u}^{-k+1}}{k-1} \mper
\end{align*}
Finally, note that since the $k$-th moment of $X$ is bounded by $1$, its variance is bounded by $1$ as well. Therefore, $\Var\brac{\Sbar} \le 1/m$, and hence
\begin{equation*}
\alpha = O_k \Paren{ 
\sqrt{\frac{1}{mn\beta}}
+
\frac{\rho \log\paren{1/\beta}}{n \epsilon}
+ 
m^{-k+1} \cdot \paren{\rho - u}^{-k+1}}
\mper
\end{equation*}
\end{proof}

Before presenting the main theorem of the section, we state a corollary that is related to the contribution of the bias term in the error rate given in the previous lemma.

\begin{corollary}    
\label{cor:bias_error_batches_1d}
Suppose $X$ is a distribution with mean $\mu$ and $k$-th moment bounded by $1$.
Let $\Sbar$ be the random variable corresponding to the mean of $m$ \iid samples from $X$.
Suppose that we are given access to a coarse estimate $\mu_{\texttt{coarse}}$ such that $\Abs{\mu_{\texttt{coarse}} - \mu} \le u$, and that $Z$ is the result of truncating $\Sbar$ in the interval $[\mu_{\texttt{coarse}} - \rho, \mu_{\texttt{coarse}} + \rho]$, where $\rho - u \geq \sqrt{\frac{\paren{k-1} \log m}{m}}$.
Then, the bias induced by this truncation operation satisfies:
\[
    \Abs{\E[Z] - \E[X]} \le m^{-k+1} \frac{\paren{\rho - u}^{-k+1}}{k-1}.
\]
Thus, if we want the bias to be at most $\alpha$, we need to take $\rho \geq_k u + \frac{1}{\alpha^{\frac{1}{k - 1}} m}$.
\end{corollary}

We now conclude by presenting the main result of this section.

\begin{theorem}[Mean Estimation for Bounded $k$-th Moment Distributions]
\label{cor:mean-estimation-bounded}
Suppose $X$ is a distribution with mean $\mu$ and $k$-th moment bounded by $1$. Then there exists a person-level $\eps$-DP mean estimation algorithm that takes 
\begin{align*}
\tilde{O}_k \Paren{\frac{1}{\alpha^2 m \beta}
+ 
\frac{\log\paren{1/\beta}}{\alpha \sqrt{m} \eps}
+ 
\frac{\log\paren{1/\beta}}{\alpha^{\frac{k}{k-1}} m \eps}
+
\frac{\log \paren{Rm/\beta}}{\eps}
}
\end{align*}
and outputs an estimate of the mean with failure probability $\beta$ and accuracy $\alpha$. Moreover $\tilde{O}_k$ hides multiplicative factors that only depend on $k$ and lower order logarithmic factors in $m$.
\end{theorem}

\begin{proof}
Suppose we have access to a coarse estimate of the mean with accuracy $16\sqrt{1/m}$. Then by \Cref{lem:coarse-to-fine-bounded} we have that there exists an $\eps$-DP estimator of the mean that takes samples from the the sample mean distribution truncates them within a radius of $\rho$ around the coarse estimate, and outputs an estimate of the mean with success probability $\beta$ and accuracy
\begin{equation*}
\alpha  \le_k
\sqrt{\frac{1}{m n \beta}}
+
\frac{\rho \log\paren{1/\beta}}{n \epsilon}
+
m^{-k+1} \cdot \paren{\rho - u}^{-k+1},
\end{equation*}
where $u =16\sqrt{1/m}$. Take $\rho = \Theta\Paren{\sqrt{\frac{\paren{k-1}\log m}{m}} +\Paren{\frac{n\eps}{\log\paren{1/\beta}}}^{1/k} \cdot \frac{1}{m^{1 - 1/k}}}$, hiding a sufficiently large constant to ensure that $\rho - u \ge \sqrt{\frac{\paren{k-1} \log m}{m}}$. Therefore, we have that

\begin{align*}
\alpha &\le_k
\sqrt{\frac{1}{mn \beta}}
+ 
\frac{\log\paren{1/\beta} \sqrt{\log m}}{\sqrt{m} n \eps}
+
\Paren{\frac{\log\paren{1/\beta}}{m n \eps}}^{1 - 1/k}.
\end{align*}

Rearranging the terms and using the coarse estimation sample complexity from \Cref{thm:user-level-coarse-esitmation}, with $r = 16\sqrt{1/m}$ we conclude that there exists a person-level $\eps$-dp mean estimation algorithm that takes 

\begin{align*}
n = \tilde{O}_k \Paren{\frac{1}{\alpha^2 m\beta}
+ 
\frac{\log\paren{1/\beta}}{\alpha \sqrt{m} \eps}
+ 
\frac{\log\paren{1/\beta}}{\alpha^{\frac{k}{k-1}} m \eps}
+
\frac{\log \paren{Rm/\beta}}{\eps}
}
,
\end{align*}
many samples and outputs an estimate of the mean with success probability $\beta$ and accuracy $\alpha$. Moreover $\tilde{O}_k$ hides multiplicative factors that only depend on $k$ and lower order logarithmic factors in $m$.
\end{proof}

\makeatletter
\newcommand{\batch}{%
  \@ifnextchar\bgroup{\batchwitharg}{\batchwithoutarg}%
}

\newcommand{\batchwitharg}[1]{%
  X^{(#1)} %
}

\newcommand{\batchwithoutarg}{%
  X
}

\newcommand{\clipbatch}{%
  \@ifnextchar\bgroup{\clipbatchwithargA}{\clipbatchwithoutarg}%
}

\newcommand{\clipbatchwithargA}[1]{%
  \@ifnextchar\bgroup{\clipbatchwithargB{#1}}{\clipbatchwithonearg{#1}}%
}

\newcommand{\clipbatchwithargB}[2]{%
  Z_{#1}(#2)
}

\newcommand{\clipbatchwithonearg}[1]{%
  Z_{#1}
}

\newcommand{\clipbatchwithoutarg}{%
  Z
}
\makeatother
\newcommand{\clip}[1]{\mathrm{clip}_{\rho,u}\Paren{#1}}
\newcommand{\itclip}{Threaded Clip-and-Noise\xspace}
\newcommand{\itcp}{TCN\xspace}
\section{Mean Estimation in High Dimensions with Approximate-DP}\label{sec:hd}
In this section, we introduce \Cref{alg:hd}, an algorithm for privately and efficiently estimating the mean of distributions with bounded $k$-th moments.

\paragraph{Using Clip-and-Noise in $d$ dimensions.} As in the univariate case, our algorithm for computing an estimate of the mean of a multivariate distribution is as follows: ($1$) compute a coarse estimate to the mean via private histograms, and ($2$) apply the clip-and-noise framework, clipping around the coarse estimate. As in the univariate case, our overall error is proportional to the bias induced by clipping and also the noise added to ensure privacy. To ensure overall error $\alpha$, we set the clipping radius $\rho$ to be as small as possible while still ensuring that the bias is $\lesssim \alpha$. Under this fixed $\rho$, we can guarantee that the noise is also $\lesssim \alpha$ so long as the number $n$ of people is sufficiently large.

\paragraph{Main technical obstacles in $d$ dimensions.} In one dimension, we used a corollary of the Non-Uniform Berry-Esseen Theorem (\Cref{cor:non-uniform-tail-bound}) from \cite{Michel76} in order to bound the bias due to clipping. Unfortunately, this bound does not apply to distributions over $\R^d$ for $d > 1$, and moreover the proof introduced by \cite{Michel76} does not readily extend to higher dimensions. In this section, we prove an analogous concentration bound in high dimensions. 

\begin{restatable}{theorem}{thmhighdtailbound}\label{thm:highd-tailbound}
    Let $k > 2$, and let $\cD$ have mean $0$ and $\sigma_k(\cD) \le 1$. There exists $t_1 = O\Paren{\sqrt{\frac{d\log m}{m}}}$ such that, for all $t \ge t_1$,
    \begin{equation*}
        \Pr\Brac{\normt{X} \ge t} = \tilde{O}\Paren{\frac{d^{k/2}}{m^{k-1}t^k} + e^{-mt^2 / d 
        }}.
    \end{equation*}
\end{restatable}

We use this concentration inequality in order to bound the bias due to Clip-and-Noise. As discussed in \Cref{sec:intro}, the techniques used to bound the bias in the univariate case do not give tight results in the multivariate case. In this section, we introduce \Cref{lem:new-highd-clipping-bias} whose proof circumvents this issue.
\begin{restatable}{theorem}{thmbias}\label{lem:new-highd-clipping-bias}
Let $\cD$ be a distribution over $\R^d$ with mean $\mu$ and $\sigma_k(\cD) \le 1$. Suppose that $X_j \sim \cD$ for all $j \in [m]$, and let $X = \frac{1}{m}\sum_{j=1}^mX_j$. Let $u \in \R^d$, $\rho \ge 0$, and $\gamma = \normt{u - \mu}$. If $\rho \ge t_1$ for some $t_1 = O\Paren{\sqrt{\frac{d\log m}{m}}}$, and if $\gamma \le \sqrt{\frac{d}{m}}$, then,
    \begin{equation}
        \Normt{\mu - \E\Brac{\clip{X}}} 
        = \tilde{O}\Paren{\frac{d^{\frac{k-1}{2}}}{m^{k-1}\rho^{k-1}}\Paren{1 + \gamma \cdot \frac{d^{\frac{1}{2}}}{\rho}}}.
    \end{equation}
\end{restatable}

\paragraph{Main algorithmic obstacle in $d$ dimensions.} Even equipped with a high-dimensional tail bound (\Cref{thm:highd-tailbound}) and bound on the bias (\Cref{lem:new-highd-clipping-bias}), there is an inherent algorithmic obstacle that exists in the multivariate case that does not exist in the univariate case.

Recall that the error of Clip-and-Noise is proportional to both the bias induced by the clipping and the noise added to ensure privacy. Thus, if the overall error is $\alpha$, then each of the bias and the noise must be $\lesssim \alpha$. As seen in \Cref{lem:new-highd-clipping-bias}, the bias is proportional to the error $\gamma$ of the coarse estimate. In order for the bias to be $\lesssim \alpha$, the clipping radius $\rho$ must be large enough to offset the error $\gamma$. However, enlarging $\rho$ also enlarges the sensitivity and thus the noise! The only way to then decrease the sensitivity is to increase the number of available people in the data set. Unfortunately, the $\gamma$ produced by the coarse estimate is too large, ultimately requiring $n$ to be larger than optimal.

To get around this problem, instead of forcing the optimal error $\alpha$ and settling for a sub-optimal $n$, we instead enforce the optimal $n$ and observe a different error $\gamma' > \alpha$. In our parameter setting, while $\gamma'$ is not smaller than $\alpha$, it is \textit{smaller} than $\gamma$; thus, we can actually interpret the clip-and-noise algorithm as boosting the error of the coarse estimate. The insight is that this process is repeatable. In particular, we will run a process called \itclip which does exactly this, and is able to drive the error of the coarse estimate down to optimal.\footnote{\emph{We can get the optimal coarse estimate for free!}} Then, with the optimal coarse estimate, we can do one last application of Clip-and-Noise in order to achieve error $\alpha$ with the optimal sample complexity.

In this section, we work towards introducing \Cref{alg:hd}, an algorithm which composes the processes described above. The guarantees of \Cref{alg:hd} are described in the following theorem.
\begin{restatable}{theorem}{thmapproxdpupperbound}\label{thm:approx-dp-upperbound}
    Let $\cD$ be a distribution over $\R^d$ with mean $\mu$, and let $\sigma_k(\cD) \le 1$ for some $k > 2$. For all $\eps, \delta > 0$, \Cref{alg:hd} satisfies person-level $(\eps,\delta)$-DP. Furthermore, \Cref{alg:hd} is efficient, and for all $\eps, \delta, \alpha \in (0,1)$, there exists 
    \begin{equation}
        \label{eq:sample-complexity-twice-clipped}
        n_0 = \tilde{O}_k\Paren{\frac{d}{\alpha^2 m} + \frac{d\log^{1/2} \paren{1/\delta}}{\alpha m^{1/2}\eps} + \frac{d\log^{1/2} \paren{1/\delta}}{\alpha^{k/(k-1)}m\eps} + \frac{d\log^{1/2} \paren{1/\delta}}{\eps} + \frac{d^{1/2} \log^{3/2}\paren{1/\delta}}{\eps}}
    \end{equation}
    such that, if $n \ge n_0$, then, with probability at least $2/3$, \Cref{alg:hd} outputs $\hat{\mu} \in \R^d$ such that $$\normt{\hat{\mu} - \mu} \le \alpha.$$
\end{restatable}

\paragraph{Organization of \Cref{sec:hd}.} The remainder of this section is organized as follows. In \Cref{subsec:technical-lemmata}, we prove \Cref{thm:highd-tailbound} and \Cref{lem:new-highd-clipping-bias}. In \Cref{subsec:iterative-clip}, we discuss \itclip. In \Cref{sec:full-algo} we give the full algorithm and prove \Cref{thm:approx-dp-upperbound}.

\subsection{Technical Lemmata: Bounding the Bias}\label{subsec:technical-lemmata}
In this section, we prove a concentration bound for distributions in $d$ dimensions with bounded $k$th moments. Using this concentration bound, we analyze the bias induced by Clip-and-Noise.
\subsubsection{Theorem~\ref{thm:highd-tailbound}: Proof Overview and Core Lemmata}\label{sec:highdim-overview}
Let $\cD$ be a distribution over $\R^d$, let $X_1,\ldots,X_m \stackrel{\text{iid}}{\sim} \cD$, and let $X = \frac{1}{m}\sum_{i=1}^m X_i$. We introduce the following theorem which is central in the analysis of the bias:
\thmhighdtailbound*
This theorem is a direct corollary of the following two theorems:

\begin{restatable}{theorem}{thmsmallt}
    Let $k > 2$, and let $\cD$ have mean $0$ and $\sigma_k(\cD) \le 1$. There exists $t_1 = O\Paren{\sqrt{\frac{d\log m}{m}}}$ and $t_2 = \Omega\Paren{\sqrt{d}\log^{\frac{-1}{k-2}}m}$ such that, for all $t \in [t_1,t_2]$,
    \begin{equation*}
        \Pr\Brac{\normt{X} \ge t} = \tilde{O}\Paren{\frac{d^{k/2}}{m^{k-1}t^k} + e^{-mt^2 / d 
        }}.
    \end{equation*}
    \label{thm:thmsmallt}
\end{restatable}

\begin{restatable}{theorem}{thmbigt}
    Let $k > 2$, and let $\cD$ have mean $0$ and $\sigma_k(\cD) \le 1$. There exists $t_3 = O\Paren{\sqrt{d}\log^{1/k} m}$ such that, for all $t \ge t_3$,
    \begin{equation*}
        \Pr\Brac{\normt{X} \ge t} = \tilde{O}\Paren{\frac{d^{k/2}}{m^{k-1}t^k}}.
    \end{equation*}
    \label{thm:thmbigt}
\end{restatable}

In this subsection, we give a proof overview of both \Cref{thm:thmsmallt} and \Cref{thm:thmbigt} (whose combination yields \Cref{thm:highd-tailbound}). As we give the overview, we will also develop some high-level lemmas that will be used in both analyses. Building on these lemmas, we will give the full proofs of the theorems in the following subsections. 

As notation, we say that a random variable $Z \in \R^d$ has a \textit{$t$-deviation} if $\normt{Z} \ge t$. In \Cref{thm:thmsmallt}, we analyze the probability that $X$ has a $t$-deviation for all $\sqrt{d/m} \lesssim t \lesssim \sqrt{d}$ (we call all $t$ in this range \textit{small $t$}). In \Cref{thm:thmbigt}, we analyze the probability that $X$ has a $t$-deviation for all $t \gtrsim \sqrt{d}$ (we call all $t$ in this range \textit{large $t$}). The analyses of the small-$t$ event and the large-$t$ event closely follow the same blueprint, which we describe in this subsection.

To start, let $r_1, r_2 \ge 0$ such that $r_1 < r_2$. We call a sample $X_i$ \textit{light} if $\normt{X_i} < r_1$. We call a sample $X_i$ \textit{moderate} if $\normt{X_i} \in [r_1,r_2)$. We call a sample $X_i$ \textit{heavy} if $\normt{X_i} \ge r_2$. We can group together the light, moderate, and heavy samples as follows:
\begin{align*}
A_1 &= \frac{1}{m} \sum_{i=1}^m X_i \cdot \ind{\normt{X_i} < r_1} \\
A_2 &= \frac{1}{m} \sum_{i=1}^m X_i \cdot \ind{r_1 \le \normt{X_i} < r_2} \\
A_3 &= \frac{1}{m} \sum_{i=1}^m X_i \cdot \ind{\normt{X_i} \ge r_2}
\end{align*}
We can use these groupings to expand the probability that $X$ has a $t$-deviation:
\begin{equation}
    \Pr\Brac{\normt{X} \ge t} 
    \le \Pr\Brac{\normt{A_1} \ge \frac{t}{3}} + \Pr\Brac{\normt{A_2} \ge \frac{t}{3}} + \Pr\Brac{\normt{A_3} \ge \frac{t}{3}}.
    \label{eq:pigeon}
\end{equation}
The proofs of \Cref{thm:thmsmallt} and \Cref{thm:thmbigt} focus on bounding the probabilities that each of $A_1$, $A_2$, and $A_3$ have $(t/3)$-deviations. 

In order to bound these events, we must take care to set $r_1$ and $r_2$ appropriately as functions of $t$. In particular, for both the small-$t$ regime and the large-$t$ regime, we will always set $r_2 = \frac{mt}{3\log m} \approx mt$. That is, we should think of $r_2$ as being a global variable for both the small-$t$ regime and the large-$t$ regime. On the other hand, we should think of $r_1$ as being a local variable which is set differently depending on whether $t$ is small or large (this is actually why we divide $t$ into small $t$ and large $t$ in the first place). In the small-$t$ regime, we set $r_1 = d/t$. In the large-$t$ regime, we set $r_1 = t/3$.

As it turns out, the proof bounding the deviation of the heavy samples is agnostic to whether $t$ is small or large. On the other hand, the proofs bounding the deviation of the light samples will necessarily be different for small $t$ and large $t$. The proofs bounding the deviation of the moderate samples are largely similar for small $t$ and large $t$. We focus for now on proving bounds for the deviations of just the heavy and moderate samples. (In ensuing subsections, we bound the probability that $A_1$ has a $(t/3)$-deviation separately for small $t$ and for large $t$.) 

Throughout the proof, we make extensive use of the following lemma:
\begin{lemma}[Tail bound for the norm of bounded $k$-th moment distribution \cite{Zhu2022}]
\label{lem:tail-bound-norm}
Assume $p$ has its $k$-th moment bounded by $1$ for $k \ge 2$. Then for all $t > 0$,
\begin{equation*}
\Pr\Brac{\normt{X} \ge t} \le d^{k/ 2} t^{-k}.
\end{equation*}    
\end{lemma}

\paragraph{Bounding the deviation of the heavy samples.}
We begin by introducing \Cref{lem:heavy-samples}, which bounds the probability that $A_3$ has a $(t/3)$-deviation. This is the easiest step in the analysis and is used in both the small-$t$ and large-$t$ regimes. As mentioned earlier, we set $r_2 = \frac{mt}{3\log m}$.

\begin{lemma}[Deviation of Heavy Samples]
Let $k \ge 2$, and let the distribution $\cD$ over $\R^d$ have mean $0$ and $k$-th moment bounded by $1$. Then for all $t > 0$,
\begin{equation}
    \Pr\Brac{\normt{A_3} \ge \frac{t}{3}} \le \tilde{O}\Paren{\frac{d^{k/2}}{m^{k-1}t^k}}.
    \label{eq:third-term}
\end{equation}
\label{lem:heavy-samples}
\end{lemma}
\begin{proof}
    Note that
    \begin{align}
        \Pr\Brac{\normt{A_3} \ge \frac{t}{3}} 
        &= \Pr\Brac{\Bignormt{\frac{1}{m}\sum_{\substack{i \\ \normt{X_i} \ge r_2}} X_i} \ge \frac{t}{3}}\\
        &\le\Pr\Brac{\sum_{\substack{i \\ \normt{X_i} \ge r_2}} \normt{X_i} \ge \frac{mt}{3}} \notag \\
        &\le \Pr\Brac{\sum_{\substack{i \\ \normt{X_i} \ge r_2}} \normt{X_i} \ge r_2} \notag \\
        &= \Pr\Brac{\exists i \text{ such that } \normt{X_i} \ge r_2} \notag \\
        &\le m \Pr\Brac{\normt{X_1} \ge r_2}, \label{eq:last_equation}
    \end{align}
    where the last line follows from the Union Bound. We now use \Cref{lem:tail-bound-norm} to bound the probability of the norm of one sample having an $r_2$-deviation:
    \begin{equation}
        \Pr\Brac{\normt{X_1} \ge r_2}
        \le r_2^{-k} d^{k/2}
        = \tilde{O}\Paren{\frac{d^{k/2}}{m^kt^k}},
    \end{equation}
    which combined with \eqref{eq:last_equation} completes the proof.
\end{proof}

\paragraph{Bounding the deviation of the moderate samples.}
We give an overview of how to bound the probability of the moderate samples deviating for both small $t$ and large $t$. This is the most interesting part of the analysis.

We begin by partitioning the moderate samples further. WLOG assume that $r_2/r_1$ is a power of two, and for all integers $\ell \in \set{1, \dots, \log\Paren{r_2/ r_1}}$ we define the interval
\begin{equation*}
B_{\ell} = \Set{i : 2^{-\ell} \cdot r_2 \le \normt{X_i} < 2^{-\ell + 1} \cdot r_2 }.
\end{equation*}
Note that, as $\ell$ gets larger, the interval $B_{\ell}$ shrinks in width and also covers a smaller range of values. Consider the following remark about these intervals.

\begin{remark}
\label{claim:B_ell_exists}
If $\log (r_2/r_1) \le \log m$ and $A_2$ has a $(t/3)$-deviation, then there exists $\ell \in \Set{1, \dots, \log\Paren{r_2/ r_1}}$ such that $|B_{\ell}| \ge 2^{\ell-1}$. 
\end{remark}
\begin{proof}
Suppose for contradiction that $\forall \ell : \abs{B_{\ell}} < 2^{\ell-1}$. Then, 
\begin{align}
    \normt{A_2} 
    &=\Bignormt{\frac{1}{m}\sum_{\substack{i \\ \normt{X_i} \in [r_1,r_2)}} X_i} \notag \\
    &\le \frac{1}{m}\sum_{\substack{i \\ \normt{X_i} \in [r_1,r_2)}} \normt{X_i} \notag \\ 
    &= \frac{1}{m}\sum_{\ell = 1}^{\log\Paren{r_2/ r_1}} \sum_{i \in B_{\ell}} \normt{X_i} \notag \\
    &< \frac{1}{m}\sum_{\ell = 1}^{\log\Paren{r_2/ r_1}} 2^{\ell-1} \cdot \paren{2^{-\ell + 1} \cdot r_2} \notag \\
    &= \frac{\log\Paren{r_2/ r_1} r_2}{m}.
    \label{eq:bucketing-argument}
\end{align}
By assumption on $\log(r_2/r_1) \le \log m$, equation \eqref{eq:bucketing-argument} becomes
\begin{equation*}
\normt{A_2} < \frac{\log\Paren{r_2/ r_1} r_2}{m} \le \frac{\log m}{m} \cdot \frac{mt}{3\log m} \le \frac{t}{3},
\end{equation*}
which is a contradiction.
\end{proof}

We will assume $\log(r_2/r_1) \le \log m$ and use \Cref{claim:B_ell_exists} to expand the probability that $A_2$ has a $(t/3)$-deviation:
\begin{align}
    \Pr\Brac{\normt{A_2} \ge \frac{t}{3}}
    &\le\sum_{\ell = 1}^{\log\Paren{r_2/ r_1}} \Pr\Brac{\abs{B_{\ell}} \ge 2^{\ell-1}} \tag{by \Cref{claim:B_ell_exists}} \notag\\
    &\le \sum_{\ell = 1}^{\log\Paren{r_2/ r_1}} \binom{m}{2^{\ell-1}} \Pr\Brac{\normt{X_i} \ge 2^{-\ell} \cdot r_2}^{2^{\ell-1}} \tag{by Union Bound}\\
    &\le \sum_{\ell = 1}^{\log\Paren{r_2/ r_1}} \Paren{\frac{em}{2^{\ell-1}}}^{2^{\ell-1}} \Paren{\Paren{\frac{2^{\ell} \cdot \sqrt{d}}{r_2}}^k}^{2^{\ell-1}}\tag{by \Cref{lem:tail-bound-norm} and \Cref{lem:binom-upperbound}} \notag\\ 
    &\le \sum_{\ell = 1}^{\log\Paren{r_2/ r_1}} \Paren{\frac{e \cdot m  \cdot d^{k/2} \cdot 2^{\ell k - \ell + 1} }{r_2^k}}^{2^{\ell-1}}\notag\\
    &\le \log(r_2/r_1)\cdot \max_{\ell \in \{1,\ldots,\log(r_2/r_1)\}}\Paren{\frac{e \cdot m  \cdot d^{k/2} \cdot 2^{\ell k - \ell + 1} }{r_2^k}}^{2^{\ell-1}}\label{eq:expanding-moderate-samples}
\end{align}
Since $\log(r_2/r_1) = \tilde{O}(1)$, then the rest of the proof will reduce to showing
\begin{equation}
    \max_{\ell}\Paren{\frac{2\cdot 3^ke\cdot d^{k/2} \cdot 2^{\ell(k - 1)}\log^km }{m^{k-1}t^k}}^{2^{\ell-1}} = \tilde{O}\Paren{\frac{d^{k/2}}{m^{k-1}t^k}}.
    \label{eq:sum-to-bound}
\end{equation}
Showing \eqref{eq:sum-to-bound} will follow from algebraic calculations. To lighten the notation, we will define the function $$f(\ell) = \Theta\Paren{\frac{d^{k/2} \cdot 2^{\ell(k - 1)}\log^km}{m^{k-1}t^k}}^{2^{\ell-1}},$$ where the $\Theta$-notation hides the $2 \cdot 3^k e$ term in \eqref{eq:sum-to-bound}. To bound $\max_{\ell} f(\ell)$, it helps to observe \Cref{lem:f-is-convex} which states that $f(\ell)$ is convex with respect to $\ell$ (see \Cref{sec:highdim-overview-missing} for the proof).
\begin{restatable}{lem}{lemfisconvex}
    Fix any $m,d,t > 0$ and $k \ge 1$. Then the function $f(\ell)$ is convex with respect to $\ell$.
    \label{lem:f-is-convex}
\end{restatable}
By \Cref{lem:f-is-convex}, it follows that 
\begin{equation}
    \max_{\ell} f(\ell) \in \{f(1),f(\log(r_2/r_1)\},
    \label{eq:max-f-ell}
\end{equation}
and so it suffices to show that $f(1)$ and $f(\log(r_2/r_1))$ are bounded. Note that, for both small $t$ and large $t$,
\begin{equation}
  f(1) = \Theta\Paren{\frac{d^{k/2} \cdot 2^{k - 1}\log^km}{m^{k-1}t^k}} = \tilde{O}\Paren{\frac{d^{k/2}}{m^{k-1}t^k}}.
  \label{eq:ell-is-one}
\end{equation}
On the other hand, we must bound $f(\log(r_2/r_1))$ separately for small $t$ and large $t$, as in each analysis we will set $r_1$ differently. We leave this subsection and bookmark our progress with the following lemma:
\begin{lemma}
    If $\log (r_2/r_1) \le \log m$ and $f(\log(r_2/r_1)) = \tilde{O}(d^{k/2}m^{-k+1}t^{-k})$, then $$\Pr\Brac{\normt{A_2} \ge \frac{t}{3}} = \tilde{O}\Paren{\frac{d^{k/2}}{m^{k-1}t^k}}.$$
    \label{lem:missing-piece}
\end{lemma}
\subsubsection{The Small-\texorpdfstring{$t$}{t} Regime: Proving Theorem~\ref{thm:thmsmallt}}
\label{sec:highdim-smallt}
In this section, we prove the following theorem:
\thmsmallt*
Let $r_1 = \frac{d}{t}$, $r_2 = \frac{mt}{3\log m}$, and let $A_1$, $A_2$, and $A_3$ be defined as follows:
\begin{align*}
A_1 &= \frac{1}{m} \sum_{i=1}^m X_i \cdot \ind{\normt{X_i} < r_1} \\
A_2 &= \frac{1}{m} \sum_{i=1}^m X_i \cdot \ind{r_1 \le \normt{X_i} < r_2} \\
A_3 &= \frac{1}{m} \sum_{i=1}^m X_i \cdot \ind{\normt{X_i} \ge r_2}
\end{align*}
As discussed earlier, we must bound the probabilities that $A_1$, $A_2$, and $A_3$ have $(t/3)$-deviations. We have already bounded the heavy samples by \Cref{lem:heavy-samples}. We move on to showing the complete proof of the bound for the moderate samples in \Cref{lem:moderate-samples}:
\begin{lemma}[Deviation of Moderate Samples in Small-$t$ Regime]
Let $k > 2$, $\cD$ have mean $0$ and $k$-th moment bounded by $1$, and $t > 0$. Then, there exists $t_1 = O\Paren{\sqrt{\frac{d\log m}{m}}}$ and $t_2 = \Omega(\sqrt{d}\log^{\frac{-1}{k-2}}m)$ such that, for all $t \in [t_1,t_2]$,
\begin{equation}
    \Pr\Brac{\normt{A_2} \ge \frac{t}{3}} \le \tilde{O}\Paren{\frac{d^{k/2}}{m^{k-1}t^k}}.
    \label{eq:second-term-small-t}
\end{equation}
\label{lem:moderate-samples}
\end{lemma}
\begin{proof}
We seek to apply \Cref{lem:missing-piece} (as described in the previous subsection). Note that there exists $t_2 = \Omega(\sqrt{d}\log^{\frac{-1}{k-2}}m)$ such that, for all $t \le t_2$,
$$\log(r_2/r_1) = \log\Paren{\frac{mt^2}{3d\log m}} \le \log m.$$
Thus, all that remains is proving that $f(\log(r_2/r_1)) = \tilde{O}(d^{k/2}m^{-k+1}t^{-k})$. We defer the proof of this to \Cref{lem:ell-is-large} in \Cref{sec:highdim-smallt-missing}.
\end{proof}

We now introduce \Cref{lem:light-samples}, which shows that the deviation of the light samples is bounded. 
\begin{lemma}[Deviation of Light Samples in Small-$t$ Regime]
Let $k > 2$, and let the distribution $\cD$ over $\R^d$ have mean $0$ and $k$-th moment bounded by $1$. Then, there exists $t_2 = \Omega(\sqrt{d})$ such that, for all positive $t \le t_2$,
\begin{equation}
    \Pr\Brac{\normt{A_1} \ge \frac{t}{3}} \le \exp\Paren{-\Theta\Paren{\frac{mt^2}{d}}}.
    \label{eq:first-term}
\end{equation}
\label{lem:light-samples}
\end{lemma}

\Cref{lem:light-samples} critically uses a Bernstein Inequality for vectors (see, e.g., Lemma 18 of \cite{kohler2017sub}):
\begin{lemma}[Vector Bernstein Inequality]
\label{lem:vector-bernstein}
Let $Z_1, \dots, Z_m$ be independent random vectors in $\R^d$ such that $\E[Z_i] = 0$, $\E[\Snormt{Z_i}] \le \sigma^2$  and $\normt{Z_i} \le r$ for all $i$. Then for all $0 < t < \sigma^2 / r$,
\begin{equation*}
\Pr\Brac{\Bignormt{\sum_{i=1}^m Z_i} \ge mt} \le \exp\Paren{\frac{-mt^2}{8 \sigma^2} + \frac{1}{4}}.
\end{equation*}
\end{lemma}

We now prove \Cref{lem:light-samples}.
\begin{proof}[Proof of \Cref{lem:light-samples}]
    For each $i \in [m]$, set $Y_i = X_i\cdot \ind{\normt{X_i} < r_1}$, and $Z_i = Y_i - \E\Brac{Y_i}$.
    We can rewrite the left-hand side of \eqref{eq:first-term} in terms of $Z_i$ and $\E\Brac{Y_i}$ as follows:
    \begin{align*}
    \Pr\Brac{\normt{A_1} \ge \frac{t}{3}}
    &= \Pr\Brac{\Bignormt{\frac{1}{m}\sum_{\substack{i \\ \normt{X_i} < r_1}} X_i} \ge \frac{t}{3}}\\
    &= \Pr \Brac{\Bignormt{\sum_{\substack{i \\ \normt{X_i} < r_1}} X_i} \ge \frac{mt}{3}} \\
    &= \Pr \Brac{\Bignormt{\sum_{i=1}^m Y_i} \ge \frac{mt}{3}} \\
    &= \Pr \Brac{\Bignormt{\sum_{i= 1}^m Y_i - \E\Brac{Y_i} + \E\Brac{Y_i}} \ge \frac{mt}{3}} \\
    &= \Pr \Brac{\Bignormt{\sum_{i= 1}^m Z_i + \E\Brac{Y_i}} \ge \frac{mt}{3}} \\
    &\le \Pr \Brac{\Bignormt{\sum_{i= 1}^m Z_i} \ge m \Paren{\frac{t}{3} - \normt{\E\Brac{Y_i}}}}.\tag{by the Triangle Inequality}
    \end{align*}
    We have that $\normt{\E\Brac{Y_i}} \le t/6$ by \Cref{lem:bounding-expectation-in-the-tails}, which uses the fact that $\cD$ has a bounded $k$-th moment (see \Cref{sec:highdim-smallt-missing} for the proof). This gives us
    \[
    \Pr\Brac{\normt{A_1} \ge \frac{t}{3}} \le \Pr \Brac{\Bignormt{\sum_{i= 1}^m Z_i} \ge \frac{mt}{6}}.
    \]
    Applying Lemma~\ref{lem:vector-bernstein} to the expression above---using the parameters $\sigma^2 = d$ and $r = 2r_1$---completes the proof. This is done  in detail in \Cref{lem:bernstein-buildup-1} and \Cref{lem:final-berstein} in Section~\ref{sec:highdim-smallt-missing}. In particular, \Cref{lem:bernstein-buildup-1} checks the constraints on the $Z_i$'s, and \Cref{lem:final-berstein} performs the necessary algebraic manipulation to complete the proof. 
\end{proof}

\begin{proof}[Proof of \Cref{thm:thmsmallt}]
    We have that
    \begin{align*}
        \Pr\Brac{\normt{X} \ge t} &\le \Pr\Brac{\normt{A_1} \ge \frac{t}{3}} + \Pr\Brac{\normt{A_2} \ge \frac{t}{3}} + \Pr\Brac{\normt{A_3} \ge \frac{t}{3}}\\
        &\le \tilde{O}\Paren{\frac{d^{k/2}}{m^{k-1}t^k} + e^{-mt^2 / d}}.\tag{by \Cref{lem:light-samples}, \Cref{lem:moderate-samples}, and \Cref{lem:heavy-samples}}
    \end{align*}
\end{proof}
\subsubsection{The Large-\texorpdfstring{$t$}{t} Regime: Proving Theorem~\ref{thm:thmbigt}}\label{sec:highdim-larget}
\thmbigt*
\begin{proof}
    Let $r_1 = \frac{t}{3}$, $r_2=\frac{mt}{3\log m}$ and let $A_1$, $A_2$, and $A_3$ be defined as follows:
    \begin{align*}
        A_1 &= \frac{1}{m} \sum_{i=1}^m X_i \cdot \ind{\normt{X_i} < r_1} \\
        A_2 &= \frac{1}{m} \sum_{i=1}^m X_i \cdot \ind{r_1 \le \normt{X_i} < r_2} \\
        A_3 &= \frac{1}{m} \sum_{i=1}^m X_i \cdot \ind{\normt{X_i} \ge r_2}.
    \end{align*}
    We have that
    \begin{align*}
        \Pr\Brac{\normt{X} \ge t} &\le \Pr\Brac{\normt{A_1} \ge \frac{t}{3}} + \Pr\Brac{\normt{A_2} \ge \frac{t}{3}} + \Pr\Brac{\normt{A_3} \ge \frac{t}{3}}\\
        &= 0 + \Pr\Brac{\normt{A_2} \ge \frac{t}{3}} + \Pr\Brac{\normt{A_3} \ge \frac{t}{3}}\\
        &\le \Pr\Brac{\normt{A_2} \ge \frac{t}{3}} + \tilde{O}\Paren{\frac{d^{k/2}}{m^{k-1}t^k}}.\tag{by \Cref{lem:heavy-samples}}
    \end{align*}
    To bound $\Pr\Brac{\normt{A_2} \ge \frac{t}{3}}$, we seek to apply \Cref{lem:missing-piece}. When $r_1 = t/3$, 
    $$\log(r_2/r_1) = \log\Paren{\frac{m}{\log m}} \le \log m.$$
    Thus, all that remains is proving that $f(\log(r_2/r_1)) = \tilde{O}(d^{k/2}m^{-k+1}t^{-k})$.
    Note that 
    \begin{align*}
	f\Paren{\log \Paren{\frac{m}{\log m}}}
	&= \Theta\Paren{\frac{md^{k/2}\Paren{\frac{m}{\log m}}^{k-1}\log^km}{m^kt^k}}^{\frac{m}{2\log m}}\\
	&= \Theta\Paren{\frac{d^{k/2}\log m}{t^k}}^{\frac{m}{2\log m}}\\
    &\le \tilde{O}\Paren{\frac{d^{k/2}}{m^{k-1}t^k}}. \tag{since $t \ge t_3$} 
\end{align*}
    This completes the proof.
\end{proof}
As an immediate corollary, we have:
\begin{restatable}{corollary}{corhighdtailbound}
    Let $k > 2$, and let $\cD$ have mean $\mu$ and $\sigma_k(\cD) \le 1$. There exists $t_1 = O\Paren{\sqrt{\frac{d\log m}{m}}}$ such that, for all $t \ge t_1$,
    \begin{equation}
        \Pr\Brac{\normt{X - \mu} \ge t} = \tilde{O}\Paren{\frac{d^{k/2}}{m^{k-1}t^k}}.
        \label{eq:highd-tailbound}
    \end{equation}
    \label{cor:highd-tailbound}
\end{restatable}

\begin{proof}
    Without loss of generality, we can assume that $\cD$ has arbitrary mean $\mu$. Thus, to complete the proof, we show that the exponential term in \Cref{thm:highd-tailbound} is dominated by the polynomial term in \Cref{thm:highd-tailbound} for all $t \ge t_0 \sqrt{\frac{kd\log m}{m}}$. Set $t = c \sqrt{\frac{kd\log m}{m}}$ for some $c \ge t_0$. Observe that, for all $c \ge t_0$,
    \begin{align*}
        \exp{\Paren{-\frac{mt^2}{d}}} 
        &= \exp\Paren{-c^2k\log m}
        \le \frac{1}{m^{kc^2}}
        \le \frac{1}{m^k2^{kc^2}}
        \le \frac{1}{m^k2^{k\log c}}
        \le \frac{1}{m^kc^k}\\
        &\le \frac{1}{m^{k/2-1}c^k}\\
        &= \frac{d^{k/2}}{m^{k-1}t^k}.
    \end{align*}
\end{proof}
\subsubsection{Using Corollary~\ref{cor:highd-tailbound} to Bound the Bias}\label{sec:bias}
We now define the operation $\mathrm{clip}_{\rho,u}$. The operation takes as input some $X \in \R^d$ and clips $X$ to the $\ell_2$ ball centered at $u \in \R^d$ with radius $\rho$.
\begin{definition}[Clipping Operation]
\label{def:clipping-operation}
For $\rho \ge 0$ and $u \in \R^d$, the operation $\mathrm{clip}_{\rho,u}: \R^d\rightarrow\R^d$ is defined as follows:
\begin{equation*}
\mathrm{clip}_{\rho,u} (x)\coloneqq
\begin{cases}
x & \normt{x-u} \le \rho \\
u + \rho \cdot \frac{x-u}{\normt{x-u}} & \text{\textit{otherwise}}
\end{cases}.
\end{equation*}
\end{definition}

We now introduce the following theorem, which describes the bias of the statistical estimator that averages $m$ samples from some distribution $\cD$ with bounded $k$-th moments and applies $\mathrm{clip}_{\rho,u}$ to the average.
\thmbias*
\begin{proof}
    Define $Z = \clip{X}$, and define 
    $$v = \frac{\E\Brac{X - \clipbatch}}{\Normt{\E\Brac{X - \clipbatch}}}. $$ 
    We use $v$ to expand the expression for the bias:
    \begin{align}
        &\Normt{\mu - \E\Brac{\clipbatch}} 
        = \Normt{\E\Brac{X-\clipbatch}}
        = \E\Brac{\langle X-\clipbatch, v\rangle} \notag\\
        &\le \int_{\ell=0}^{\infty} \Pr\Brac{\langle X-\clipbatch, v\rangle \ge \ell} \,\mathrm{d}\ell \notag \\
        &= \underbrace{\int_{\ell=0}^{\rho / \sqrt{d} + \gamma} \Pr\Brac{\langle X-\clipbatch, v\rangle \ge \ell} \,\mathrm{d}\ell}_{T_1\coloneqq}  + \underbrace{\int_{\ell=\rho / \sqrt{d}+\gamma}^{\infty} \Pr\Brac{\langle X-\clipbatch, v\rangle \ge \ell} \,\mathrm{d}\ell}_{T_2\coloneqq}. \label{eq:bias-two-sums}
    \end{align}
    The rest of the proof is dedicated to bounding $T_1$ and $T_2$ in \eqref{eq:bias-two-sums}. 
    
    \textbf{Bounding $T_1$.} Our goal is to bound $T_1$ using \Cref{cor:highd-tailbound}. Be begin by bounding the projection of $X - Z$ onto $v$ by the $\ell_2$ norm of $X - Z$:
    \begin{align}
    T_1 =\int_{\ell=0}^{\rho / \sqrt{d}+\gamma} \Pr\Brac{\langle X-\clipbatch, v\rangle \ge \ell} \,\mathrm{d}\ell
    \le \int_{\ell=0}^{\rho / \sqrt{d}+\gamma} \Pr\Brac{\normt{X-\clipbatch} \ge \ell} \,\mathrm{d}\ell\notag
    \end{align}
    For $\ell > 0$, one can check that $\normt{X-\clipbatch} \ge \ell$ implies $\normt{X-\mu} \ge \ell + \rho - \gamma$. (See \Cref{lem:bias-rearranging-terms} in \Cref{sec:bias-missing} for details). Thus,
    \begin{align}
    T_1
    &\le \int_{\ell=0}^{\rho/\sqrt{d}+\gamma} \Pr\Brac{\normt{X-\mu}\ge \ell + \rho - \gamma} \,\mathrm{d}\ell \notag\\
    &\le \int_{\ell=0}^{\rho/\sqrt{d}+\gamma} \Pr\Brac{\normt{X-\mu}\ge \ell + \rho/2} \,\mathrm{d}\ell \tag{by assumption on $\gamma$ and $\rho$}\\
    &\le \int_{\ell=0}^{\rho/\sqrt{d}+\gamma} \Pr\Brac{\normt{X-\mu}\ge \rho/2} \,\mathrm{d}\ell \tag{since $\ell \ge 0$}\\
    &=  \lim_{R\rightarrow 0^+}\int_{\ell=R}^{\rho/\sqrt{d}+\gamma} \tilde{O}\Paren{\frac{d^{k/2}}{m^{k-1}\rho^k}} \,\mathrm{d}\ell \tag{by \Cref{cor:highd-tailbound} and assumption on $\rho$}\\
    &= \tilde{O}\Paren{\frac{d^{\frac{k-1}{2}}}{m^{k-1}\rho^{k-1}}\Paren{1 + \gamma \cdot \frac{d^{\frac{1}{2}}}{\rho}}}.\label{eq:bound-for-first-bias-term}
    \end{align}
    We now focus bounding $T_2$. 
    
    \textbf{Bounding $T_2$.} Our goal is to bound $T_2$ using \Cref{cor:highd-tailbound} for univariate distributions. We start by expanding $\langle X-\mu,v \rangle$:
    \begin{align}
        \langle X-\mu,v \rangle 
        &=\langle X - \clipbatch + \clipbatch - u + u - \mu,v\rangle\notag\\
        &=\langle X-\clipbatch,v\rangle + \langle \clipbatch - u,v \rangle + \langle u - \mu,v \rangle\notag\\
        &\ge \langle X-\clipbatch,v\rangle + \langle \clipbatch - u,v \rangle -\gamma.\label{eq:expanding-dot-product}
    \end{align}
    We use \eqref{eq:expanding-dot-product} to expand $T_2$:
    \begin{align}
        T_2 &= \int_{\ell=\rho / \sqrt{d}+\gamma}^{\infty} \Pr\Brac{\langle X-\clipbatch, v\rangle \ge \ell} \,\mathrm{d}\ell\notag\\
        &\le 
        \int_{\ell=\rho / \sqrt{d}+\gamma}^{\infty} \Pr\Brac{\langle X-\mu,v \rangle \ge \ell + \langle \clipbatch - u,v \rangle -\gamma} \,\mathrm{d}\ell \tag{by \eqref{eq:expanding-dot-product}}
    \end{align}
Note that $X - \clipbatch$ and $\clipbatch - u$ are non-negative scalar multiples of each other. Thus, the inequality $\langle X-\clipbatch,v\rangle \ge 0$ (implied by $\langle X-\clipbatch,v\rangle \ge \ell$) implies $\langle \clipbatch - u,v \rangle \ge 0$. Thus,
    \begin{align}
        T_2
        &\le 
        \int_{\ell=\rho / \sqrt{d}+\gamma}^{\infty} \Pr\Brac{\langle X-\mu, v\rangle \ge \ell - \gamma} \,\mathrm{d}\ell \notag\\
        &\le 
        \int_{\ell'=\rho / \sqrt{d}}^{\infty} \Pr\Brac{\langle X-\mu, v\rangle \ge \ell'} \,\mathrm{d}\ell' \tag{by a change of variables}.
    \end{align}
    Let $Y = \langle X,v\rangle$. Note that $\E[Y] = \E\Brac{\langle X,v \rangle} = \langle \mu,v \rangle$ and also that $$\sigma_k(Y) = \E\Brac{\Abs{Y - \E[Y]}^k}^{1/k} = \E\Brac{\Abs{\langle X-\mu,v \rangle}^k}^{1/k} \le 1$$ by assumption on $\cD$. Thus, we can use \Cref{cor:highd-tailbound} in the case of univariate distributions:
    \begin{align}
        T_2 &\le \int_{\ell'=\rho / \sqrt{d}}^{\infty} \Pr\Brac{\langle X-\mu, v\rangle \ge \ell'} \,\mathrm{d}\ell'\notag\\
        &= \int_{\ell'=\rho / \sqrt{d}}^{\infty} \Pr\Brac{Y - \E[Y] \ge \ell'} \,\mathrm{d}\ell'\notag\\
        &\le \int_{\ell'=\rho / \sqrt{d}}^{\infty} \Pr\Brac{\Abs{Y - \E[Y]} \ge \ell'} \,\mathrm{d}\ell'\notag\\
        &\le 
        \int_{\ell'=\rho / \sqrt{d}}^{\infty} \tilde{O}\Paren{\frac{1}{m^{k-1}(\ell')^k}} \,\mathrm{d}\ell' \tag{by \Cref{cor:highd-tailbound} and assumption on $\rho$}\\
        &\le 
        \tilde{O}\Paren{\frac{\sqrt{d}}{m\rho}}^{k-1}. \label{eq:bound-for-second-term}
    \end{align}
    Combining \eqref{eq:bias-two-sums}, \eqref{eq:bound-for-first-bias-term}, and \eqref{eq:bound-for-second-term} completes the proof.
\end{proof}

\subsection{\itclip}\label{subsec:iterative-clip}
In this section, we introduce the algorithm \itclip (\itcp). As discussed in the introduction of \Cref{sec:hd}, \itclip takes as input an initial coarse estimate $u_0$ to the mean and outputs a new estimate $u^*$. The formal guarantees of \itclip are described in \Cref{thm:iterative-clip}.

\begin{restatable}{theorem}{thmiterativeclip}\label{thm:iterative-clip}
    Let $S = \Set{S_1,\ldots,S_n}$, where each $S_i = \Set{X^{(i)}_1,\ldots,X^{(i)}_m}$ and each $X^{(i)}_j \in \R^d$. Let $u_0 \in \R^d$. For all $\eps, \delta \in (0,1)$, $\mathrm{\itcp}(S;\eps,\delta,u_0)$ satisfies person-level $(\eps,\delta)$-DP. Furthermore, let $\cD$ be a distribution over $\R^d$ with mean $\mu$ and $\sigma_k(\cD) \le 1$. Suppose that $X^{(i)}_j \sim \cD$ for all $i \in [n]$ and $j \in [m]$, and suppose that $\normt{u_0 - \mu} \le \frac{d^{1/2}}{m^{1/2}}$. Let $u^*$ be the output of $\mathrm{\itcp}(S;\eps,\delta,u_0)$. There exists $n_0 = \tilde{O}\Paren{\frac{d\sqrt{\log(1/\delta)}}{\eps}}$, such that, if $n \ge n_0$, then $\normt{u^*-\mu} = \tilde{O}\Paren{\frac{1}{m^{1/2}}}$ with probability at least $0.9$.
\end{restatable}

\paragraph{Overview of \itclip.} 
We give a high-level overview of \itclip (see \Cref{alg:itclip} for the full details). The algorithm is composed of a total of $T$ layers. In each layer $t \in [T]$, the algorithm runs an application of the high-dimensional clip-and-noise algorithm $\textrm{ClipAndNoise}$ (\Cref{alg:hd-clip}) and outputs $u_t \in \R^d$. The name \itclip comes from the fact that the layers are threaded together: the output $u_{t-1}$ is fed as input into layer $t$, and it is used as the center of the clipping ball in the application of $\textrm{ClipAndNoise}$.

To discuss \itclip in more detail, we introduce the following notation. For each $t \in [T]$ and $i \in [n]$, define $\clipbatch{i}{t} = \mathrm{clip}_{\rho_t,u_{t-t}}\Paren{\batch{i}}$ where each $\batch{i} = \frac{1}{m}\sum_{j=1}^mX^{(i)}_j$ and where $\rho_t$ is a value specified at each layer $t$. Let $W_t$ be the noise added by $\textrm{ClipAndNoise}$ in layer $t$. The output of layer $t$ is $$u_t = \frac{1}{n} \sum_{i=1}^n \clipbatch{i}{t} + W_t.$$ As further notation, let $\gamma_t = \normt{u_t - \mu}$. We refer to $\gamma_t$ as \emph{the error of $u_t$}.

\begin{algorithm}[htb]
\caption{Clip-and-Noise}\label{alg:hd-clip}
    \hspace*{\algorithmicindent} \textbf{Input:} $S = \Set{X^{(i)}_j}_{i \in [n], j \in [m]}$ where each $X^{(i)}_j \in \R^d$. Parameters $\eps,\delta,\rho,u$.\\
    \hspace*{\algorithmicindent} \textbf{Output:} $u' \in \R^d$.
    \begin{algorithmic}[1]
    \Procedure{$\textrm{ClipAndNoise}$}{$S$;$\eps$,$\delta$,$\rho$,$u$}
        \For{$i \in [n]$}
            \State $Z_i \leftarrow \clip{\frac{1}{m}\sum_{j=1}^mX^{(i)}_j}$.\label{line:clip}
        \EndFor
        \State Let $W \sim \mathcal{N}\Paren{0,\Paren{\frac{2\rho\sqrt{2\ln (4/\delta)}}{n\eps}}^2\cdot \mathbb{I}_{d \times d}}$.
        \State Return $u' = \frac{1}{n}\sum_{i=1}^nZ_i + W$.\label{line:final-output}
    \EndProcedure
    \end{algorithmic}
\end{algorithm}
\paragraph{The iterative refinements of \itclip.} The layered/threaded design choices of \itclip hinge on the following key idea: Each layer $t$ produces an estimate $u_t$ that is better than that of the previous layer; that is, for all $t \in [T]$, the algorithm produces $\gamma_t$ such that $\gamma_t < \gamma_{t-1}$. In addition, by the time the algorithm terminates, the final estimate $u_T$ will have error $\lesssim \sqrt{1/m}$ (so long as the initial estimate $u_0$ had error $\lesssim \sqrt{d/m}$). 

To see why this is true, it is useful to understand how $\gamma_t$ evolves with $t$. We can think of $\gamma_t$ as being roughly equal to the sum of a bias term and a noise term, both of which come from unraveling the details of $\textrm{ClipAndNoise}$. In our parameter regime, the bias term will be 
\begin{equation}\label{eq:gamma-bias}
    \text{bias}(t) \approx \frac{d^{k/2}\gamma_{t-1}}{m^{k-1}\rho_t^k}
\end{equation} due to \Cref{lem:new-highd-clipping-bias}. When $n \gtrsim \frac{d\sqrt{\log(1/\delta)}}{\eps}$, the noise term will be $$\textrm{noise}(t) \approx \frac{\rho_t}{\sqrt{d}}.$$ Notice that both terms are functions of $\rho_t$, which is a parameter set by the algorithm at each layer. In order to minimize $\gamma_t$, the algorithm sets $\rho_t$ to be the value that balances $\textrm{bias}(t)$ and $\textrm{noise}(t)$, i.e. setting $$\rho_t \approx \frac{\Paren{\gamma_{t-1}}^{\frac{1}{k+1}}d^{\frac{1}{2}}}{m^{\frac{k-1}{k+1}}}$$ (where we pretend for now that the algorithm knows the value of $\gamma_{t-1}$). Plugging this value of $\rho_t$ into $\gamma_t$ gives the following evolution of $\gamma_t$ with respect to $t$: $$\gamma_t \approx \frac{\Paren{\gamma_{t-1}}^{\frac{1}{k+1}}}{m^{\frac{k-1}{k+1}}}.$$ If we assume that $\gamma_0 \le \sqrt{d/m}$, then it is ideal to set \begin{equation}\label{eq:rho-stuff}
    \rho_t = \frac{d^{\frac{1}{2}+\frac{1}{2k(k+1)^{t-1}}}}{m^{1-1/k}}
\end{equation} 
at each round, yielding that $\gamma_t = \frac{d^{\frac{1}{2k(k+1)^{t-1}}}}{m^{1-1/k}}$. After $\log\log d$ layers, it follows that $\gamma_{\log\log d} \lesssim \sqrt{\frac{1}{m}}$.

The gross oversimplification above glosses over one important detail: We can only apply \Cref{lem:new-highd-clipping-bias} in \eqref{eq:gamma-bias} if we assume that $\rho_t \gtrsim \sqrt{d/m}$. Thus, if the right-hand side of \eqref{eq:rho-stuff} ever dips below $\approx \sqrt{d/m}$, we must instead set $\rho_t \approx \sqrt{d/m}$. We then apply $\textrm{ClipAndNoise}$ one final time and then terminate the algorithm. In this event, the completion of the terminating layer actually gives a final error of $\sqrt{1/m}$, in which case the algorithm is happy to terminate. 

The rest of the section is dedicated to formalizing the intuition above.

\begin{algorithm}[htb]
\caption{\itclip}\label{alg:itclip}
    \hspace*{\algorithmicindent} \textbf{Input:} $S = \Set{X^{(i)}_j}_{i \in [n], j \in [m]}$ where each $S_i = \Set{X^{(i)}_j}_{j \in [m]}$ and $X^{(i)}_j \in \R^d$. Parameters $\eps,\delta,u$.\\
    \hspace*{\algorithmicindent} \textbf{Output:} $u^* \in \R^d$.
    \begin{algorithmic}[1]
    \Procedure{\itcp}{$S$;$\eps$,$\delta$,$u$}
        \State $u_0 = u$
        \State $\eps' = \frac{\eps}{\sqrt{\log\log d\cdot \log((2\log\log d)/\delta)}}$, $\delta' = \frac{\delta}{2\log \log d}$
        \State $T = \max\Paren{0,\max\Paren{t \in \mathbb{Z} : \frac{d^{\frac{1}{2}+\frac{1}{2k(k+1)^{t-1}}}}{m^{1-1/k}} > \tilde{O}\Paren{\frac{d^{1/2}}{m^{1/2}}}}} + 1$
        \If{$T > 1$}
            \For{$t = 1,\ldots,T-1$}
                \State Let $S(t) = \Set{S_{\floor{\frac{n}{T}}(t-1) + 1},\ldots,S_{\floor{\frac{n}{T}}t}}$
                \State $u_{t} \leftarrow \textrm{ClipAndNoise}\Paren{S(t);\eps',\delta', \rho_t,u_{t-1}}$, where $\rho_t = \frac{d^{\frac{1}{2}+\frac{1}{2k(k+1)^{t-1}}}}{m^{1-1/k}}$
            \EndFor
        \EndIf
        \State $S(T) = \Set{S_{\floor{\frac{n}{T}}(T-1) + 1},\ldots,S_n}$
        \State $u^* \leftarrow \textrm{ClipAndNoise}\Paren{S(T);\eps',\delta', \rho_{T},u_{T-1}}$, where $\rho_{T} = \tilde{O}\Paren{\frac{d^{1/2}}{m^{1/2}}}$
        \State Return $u^*$
    \EndProcedure
    \end{algorithmic}
\end{algorithm}

\paragraph{Full analysis of \itclip.} 
We decompose $\gamma_{t}$ into \textit{noise}, \textit{sampling error}, and \textit{bias} terms as follows:
\begin{align*}
    \gamma_{t} &= \normt{u_{t} - \mu}\\
    &= \Normt{W_t + \frac{1}{n}\sum_{i=1}^n\clipbatch{i}{t} - \E\Brac{\clipbatch{1}{t}} + \E\Brac{\clipbatch{1}{t}} - \mu}\\
    &\le \normt{W_t} + \Normt{\frac{1}{n}\sum_{i=1}^n\clipbatch{i}{t} - \E\Brac{\clipbatch{1}{t}}} + \Normt{\E\Brac{\clipbatch{1}{t}} - \mu}\\
    &\coloneqq \textrm{noise}(t) + \textrm{sampling}(t) + \textrm{bias}(t).
\end{align*}   
We first introduce \Cref{lem:noise-and-sampling} which bounds $\textrm{noise}(t)$ and $\textrm{sampling}(t)$. Its proof is deferred to \Cref{sec:full-algo-missing}

\begin{restatable}[Noise and Sampling]{lemma}{lemnoiseandsampling}\label{lem:noise-and-sampling}
    Let $\cD$ be a distribution over $\R^d$ with mean $\mu$ and $\sigma_k(\cD) \le 1$. Suppose that, for all $i \in [n]$ and $j \in [m]$, $X^{(i)}_j \sim \cD$. There exists $n_0 = \tilde{O}\Paren{\frac{d\sqrt{\log(1/\delta)}}{\eps}}$ such that, for all $n \ge n_0$, 
    \[
        \mathrm{noise}(t) + \mathrm{sampling}(t) = \tilde{O}\Paren{\frac{\rho_t}{\sqrt{d}}}
    \]
    with probability at least $1-0.1/\log\log d$,
\end{restatable}
We can now give guarantees of each of these three terms at any given layer $t$ of \itclip. 
\begin{lemma}[Error of iteration $t$]\label{lem:noise-samp-bias}
    For all $t$, define $r_t = \frac{d^{\frac{1}{2}+\frac{1}{2k(k+1)^{t-1}}}}{m^{1-1/k}}$. For all $t \in [T]$, let $\rho_t$ and $u_{t}$ be as defined in \itclip, and let $\gamma_t = \normt{u_t-\mu}$. If $\gamma_{t-1} = \tilde{O}\Paren{r_{t-1}/\sqrt{d}}$, then there exists $n_0 = \tilde{O}\Paren{\frac{d\sqrt{\log(1/\delta)}}{\eps}}$ such that for all $n \ge n_0$, $$\gamma_{t} = \tilde{O}\Paren{\rho_t/\sqrt{d}}$$
    with probability at least $1-0.1/\log\log d$.
\end{lemma}
\begin{proof}
    We've seen that 
    \begin{align*}
        \gamma_{t} &\le \textrm{noise}(t) + \textrm{sampling}(t) + \textrm{bias}(t).
    \end{align*}
    Note by \Cref{lem:noise-and-sampling} that the $\textrm{noise}(t) + \textrm{sampling}(t) = \tilde{O}\Paren{\rho_t/\sqrt{d}}$ with probability at least $1-0.1/\log\log d$. Thus, it suffices to show that $\textrm{bias}(t)$ is also $\tilde{O}\Paren{\rho_t/\sqrt{d}}$, which is shown in the following calculation:
    \begin{align*}
        \textrm{bias}(t) 
        &= \tilde{O}\Paren{\frac{d^{k/2}\gamma_{t-1}}{m^{k-1}\rho_t^k} + \frac{d^{k/2}}{m^{k-1}\rho_t^{k-1}}\cdot\frac{1}{d^{1/2}}}\tag{from \Cref{lem:new-highd-clipping-bias}}\\
        &=\tilde{O}\Paren{\frac{d^{k/2}}{m^{k-1}\rho_{t}^k}\cdot\frac{r_{t-1}}{d^{1/2}} + \frac{d^{k/2}}{m^{k-1}\rho_t^{k-1}}\cdot\frac{1}{d^{1/2}}}\tag{by assumption on $\gamma_{t-1}$}\\
        &=\tilde{O}\Paren{\frac{d^{k/2}}{m^{k-1}r_{t}^k}\cdot\frac{r_{t-1}}{d^{1/2}} + \frac{d^{k/2}}{m^{k-1}r_t^{k-1}}\cdot\frac{1}{d^{1/2}}} \tag{since $\rho_t \ge r_t$}\\
        &=\tilde{O}\Paren{\frac{d^{k/2}d^{\frac{1}{2(k+1)^{t-1}}}}{m^{k-1}r_{t}^k}\cdot\frac{r_t}{d^{1/2}} + \frac{d^{k/2}}{m^{k-1}r_t^{k}}\cdot\frac{r_t}{d^{1/2}}} \tag{since $r_{t-1} = r_td^{\frac{1}{2(k+1)^{t-1}}}$}\\
        &=\tilde{O}\Paren{\frac{d^{k/2}d^{\frac{1}{2(k+1)^{t-1}}}}{m^{k-1}r_{t}^k}\cdot\frac{r_t}{d^{1/2}}}\\
        &=\tilde{O}\Paren{\frac{r_t}{\sqrt{d}}}\\
        &= \tilde{O}\Paren{\frac{\rho_t}{\sqrt{d}}} \tag{since $\rho_t \ge r_t$}.
    \end{align*}
\end{proof}

\begin{lemma}\label{lem:clip-is-private}
    \Cref{alg:hd-clip} satisfies person-level $(\eps,\delta)$-DP.
\end{lemma}
\begin{proof}
    For each user, let $S_i = \Set{X^{(i)}_1,\ldots,X^{(i)}_m}$. Define each $Z_i$ as in \Cref{alg:hd-clip}. Consider the function $f(S_1,\ldots,S_n) = \frac{1}{n}\sum_{i=1}^nZ_i$. The $\ell_2$-sensitivity of this function (with respect to person-level privacy) is $O\Paren{\frac{\rho}{n}}$. The proof then immediately follows by applying \Cref{lem:gaussian-mechanism}.
\end{proof}

\thmiterativeclip*
\begin{proof}
    We first prove the privacy guarantees of \Cref{thm:iterative-clip}. Without loss of generality, we can assume that each call to $\mathrm{ClipAndNoise}$ receives as input the entire dataset $X$ of size $n$ and chooses to only compute on a fraction of this dataset. By \Cref{lem:clip-is-private}, each of the $T$ iterations of $\mathrm{ClipAndNoise}$ satisfies person-level $(\eps',\delta')$-DP. Thus, by \Cref{lem:advanced-composition} and the fact that $T \le \log\log d$, \itclip satisfies person-level $(\eps,\delta)$-DP.

    We now prove the accuracy guarantees of $u^*$. Let $\gamma^* = \normt{u^* - \mu}$. Our goal will be to show that $\gamma^* = \tilde{O}\Paren{\rho_{T}/\sqrt{d}}$ as this implies that $\gamma^* = \tilde{O}\Paren{1/m^{1/2}}$. We divide $T$ into two cases: ($1$) when $T > 1$ and ($2$) when $T=1$.

    \paragraph{Case $1$ ($T > 1$):} For all $t \in [T]$, let $E_t$ be the event that $\gamma_t = \tilde{O}\Paren{\rho_t/\sqrt{d}}$, and let $\Bar{E}_t$ be the negation of this event. We have that
    \[
        \Pr\Brac{E_T}
        \ge \Pr\Brac{E_0, E_1,\ldots,E_T}
        = 1 - \Pr\Brac{\exists t\in\{0,\ldots,T\} : \Bar{E}_t}
    \]
    We can make the following realization.
    \begin{equation}\label{eq:bad-events}
        \Pr\Brac{\exists t\in\{0,\ldots,T\} : \Bar{E}_t} = \sum_{t=1}^T\Pr\Brac{\Bar{E}_t \mid E_{t-1},\ldots,E_0} + \Pr[\Bar{E_0}].
    \end{equation}
    We briefly show that, when $T > 1$, $\Pr[\Bar{E}_0] = 0$ (i.e. $\gamma_0 \le r_0/\sqrt{d}$). By assumption, $\gamma_0 \le \frac{d^{1/2}}{m^{1/2}}$. When $T > 1$, $\tilde{O}\Paren{\frac{d^{1/2}}{m^{1/2}}} \le r_1 = r_0/\sqrt{d}$, and thus, $\gamma_0 \le r_0/\sqrt{d}$. 
    
    Plugging in $\Pr[\Bar{E}_0] = 0$, we finish bounding \eqref{eq:bad-events}:
    \begin{align*}
        \Pr\Brac{E_T}
        &\ge 1 - \sum_{t=1}^T\Pr\Brac{\Bar{E}_t \mid E_{t-1},\ldots,E_0}\\
        &= 1 - \sum_{t=1}^T\Pr\Brac{\Bar{E}_t \mid E_{t-1}}\\
        &= 1 - \sum_{t=1}^T\Pr\Brac{\Bar{E}_t \mid \gamma_{t-1} = \tilde{O}\Paren{r_{t-1}/\sqrt{d}}}\tag{since $\rho_t = r_t$ for all $t \le T$}\\
        &\ge 1 - T\cdot \frac{0.01}{\log\log d}\tag{by \Cref{lem:noise-samp-bias}}\\
        &\ge 0.09,
    \end{align*}
    where the second-to-last line comes from the fact that $$T = \max\Paren{0,\max\Paren{t \in \mathbb{Z} : \frac{d^{\frac{1}{2}+\frac{1}{2k(k+1)^{t-1}}}}{m^{1-1/k}} > O\Paren{\frac{d^{1/2}\log^{1/2}m}{m^{1/2}}}}} + 1 \le \log\log d.$$
    \paragraph{Case $2$ ($T = 1$):} When $T=1$,
    \begin{align*}
        \gamma_1 &\le \textrm{noise}(1) + \textrm{sampling}(1) + \textrm{bias}(1)\\
        &= \tilde{O}\Paren{\rho_1/\sqrt{d}} + \textrm{bias}(1) \tag{by \Cref{lem:noise-and-sampling}}
    \end{align*}
    where the last line holds with probability at least $0.9$. Thus it remains to show that $\textrm{bias}(1) = \tilde{O}\Paren{\rho_1/\sqrt{d}}$, which we show in the following calculations:
    \begin{align*}
        \textrm{bias}(1) &= \tilde{O}\Paren{\frac{d^{k/2}\gamma_{0}}{m^{k-1}\rho_1^k} + \frac{d^{k/2}}{m^{k-1}\rho_1^{k-1}}\cdot\frac{1}{d^{1/2}}}\tag{from \Cref{lem:new-highd-clipping-bias}}\\
        &= \tilde{O}\Paren{\frac{d^{k/2}}{m^{k-1}\rho_1^k}\cdot\rho_1 + \frac{d^{k/2}}{m^{k-1}\rho_1^{k-1}}\cdot\frac{1}{d^{1/2}}}\tag{by assumption on $\gamma_0$}\\
        &= \tilde{O}\Paren{\frac{d^{k/2}}{m^{k-1}\rho_1^k}\cdot\rho_1}\\
        &= \tilde{O}\Paren{\frac{d^{k/2}}{m^{k-1}r_1^k}\cdot\rho_1}\tag{since $\rho_1 \ge r_1$}\\
        &= \tilde{O}\Paren{\frac{\rho_1}{\sqrt{d}}}.
    \end{align*}
\end{proof}
\subsection{Full Algorithm}\label{sec:full-algo}
We are now ready to introduce \Cref{alg:hd}. The algorithm first runs the high-dimensional private histogram subroutine \textrm{CoarseEstimate}, which is the natural extension of the one-dimensional coarse estimate. (The details are deferred to the \Cref{subsec:coarse-estimation}.) After running \textrm{CoarseEstimate}, the algorithm then runs \itclip and finishes by running a final round of \textrm{ClipAndNoise}.
\begin{algorithm}[htb]
    \caption{Person-Level Approximate-DP High-Dimensional Mean Estimator}\label{alg:hd}
    \hspace*{\algorithmicindent} \textbf{Input:} $S = \Set{X^{(i)}_j}_{i \in [3n], j \in [m]}$ where each $X^{(i)}_j \in \R^d$. Parameters $\eps,\delta,k > 0$.\\
    \hspace*{\algorithmicindent} \textbf{Output:} $\hat{\mu} \in \R^d$.
    \begin{algorithmic}[1]
    \Procedure{$\textrm{MeanEstimation}$}{$X$; $\eps,\delta,k$}
        \State Let $\rho^* = \max\Paren{\sqrt{\frac{d}{m}},\frac{n^{1/k}\eps^{1/k}d^{\frac{1}{2}-\frac{1}{k}}}{(\log 1/\delta)^{1/(2k)}m^{1-1/k}}}$.
        \State Let $Y \coloneqq \Paren{Y^{(1)},\ldots,Y^{(n)}} = \Paren{X^{(1)},\ldots,X^{(n)}}$
        \State Let $P \coloneqq \Paren{P^{(1)},\ldots,P^{(n)}} = \Paren{X^{(n+1)},\ldots,X^{(2n)}}$
        \State Let $V \coloneqq \Paren{V^{(1)},\ldots,V^{(n)}} = \Paren{X^{(2n+1)},\ldots,X^{(3n)}}$
        \State $u_0 \leftarrow \textrm{CoarseEstimate}(Y;\eps/2,\delta/2,16\sqrt{d/m})$.\label{line:u1}
        \State Let $u^* \leftarrow \mathrm{\itcp}(P;\eps/4,\delta/4)$.\label{line:u2}
        \State Let $\hat{\mu} \leftarrow \textrm{ClipAndNoise} (V;\eps/4,\delta/4,\rho^*,u^*)$.\label{line:mu-hat}
        \State Return $\hat{\mu}$.
    \EndProcedure
    \end{algorithmic}
\end{algorithm}
In the remainder of this section, we prove the following theorem:
\thmapproxdpupperbound*

We first state the guarantees of $\textrm{CoarseEstimate}$.
\begin{restatable}[Approximate DP Coarse Estimation in High Dimensions]{theorem}{thmhighdcoarseestimate}
\label{thm:high-d-coarse-estimate}
Let $n$ be the number of people, $m$ be the number of samples per person, and $\eps> 0, \delta \in (0,1)$.
Assume $P$ is a distribution over $\R^d$ with $k$-th moment bounded by $1$. Then there exists an efficient person-level $(\eps, \delta)$-DP algorithm $\textrm{CoarseEstimate}(X; \eps, \delta, r)$ that takes \iid samples $\set{X_j^{(i)}}_{i, j}$ from $P$, and outputs an estimate of the mean up to accuracy $16^{1/k} \sqrt{d / m} < r$, with success probability $1-\beta$, as long as $n \ge n_0$, for some
\begin{equation*}
n_0 = \tilde{O}\Paren{\frac{\log\paren{1/\beta}}{\log\paren{r \sqrt{m / d}}} 
+ 
\min \Set{
\frac{d\log\paren{1 / \paren{\delta\beta}}}{\eps}
,
\frac{\sqrt{d \log\paren{1/\delta}} \log\paren{1/\delta \beta}}{\eps}
}
} \mcom
\end{equation*}
where $\tilde{O}$ hides lower order logarithmic factors in $d$.
\end{restatable}
Finally, we prove the following lemma about the accuracy of the final application of \textrm{ClipAndNoise}.
\begin{restatable}{lemma}{lemfinalestimation}\label{lem:final-estimation}
       Let $d, m, n \ge 2$, let $\eps, \delta \in (0,1)$, let $k > 2$, and let $\rho^*$, $u^*$, and $\hat{\mu}$ be as defined in \Cref{alg:hd}. Suppose $\cD$ is a distribution over $\R^d$ with mean $\mu$ and $\sigma_k(\cD) \le 1$. If $V^{(i)}_1,\ldots,V^{(i)}_m \stackrel{\text{iid}}{\sim} \cD$ for all $i \in [n]$, and if $\normt{u^* - \mu} \le \frac{\rho^*}{2\sqrt{d}}$, then there exists $$n_0 =  O\Paren{\frac{d}{m\alpha^2} + \frac{d\sqrt{\log 1/\delta}}{m\eps\alpha^{\frac{k}{k-1}}}+\frac{d\sqrt{\log (1/\delta)}}{\sqrt{m}\alpha\eps}}$$ such that if $n \ge n_0$, then with probability at least $0.9$, $$\Normt{\hat{\mu} - \mu}\le \alpha.$$
\end{restatable}
\begin{proof}
     Recall that $\hat{\mu}$ is the output of an application of \Cref{alg:hd-clip}. Define $\mathsf{V}_i \coloneqq \mathrm{clip}_{\rho^*,u^*}\Paren{\frac{1}{m}\sum_{j=1}^mV_j^{(i)}}$ and define $W$ such that
    $$W \sim \mathcal{N}\Paren{0,\Paren{\frac{8\rho^*\sqrt{2\ln (16/\delta)}}{n\eps}}^2\cdot \mathbb{I}_{d \times d}}.$$ We begin by expanding $\normt{\hat{\mu} - \mu}$:
    \begin{align*}
        \normt{\hat{\mu} - \mu}
        &= \Normt{W + \frac{1}{n}\sum_{i=1}^n\mathsf{V}_i - \E\Brac{\mathsf{V_1}} + \E\Brac{\mathsf{V_1}} - \mu}\\
        &\le \underbrace{\normt{W}}_{T_1\coloneqq} + \underbrace{\Normt{\frac{1}{n}\sum_{i=1}^n\mathsf{V}_i - \E\Brac{\mathsf{V_1}}}}_{T_2\coloneqq} + \underbrace{\Normt{\E\Brac{\mathsf{V_1}} - \mu}}_{T_3\coloneqq}.
    \end{align*}
    By assumption on $n$, we can apply \Cref{lem:sampling-error} which says that, with probability at least $0.95$, 
    $$T_2 = \Normt{\frac{1}{n}\sum_{i=1}^n\mathsf{V}_i - \E\Brac{\mathsf{V_1}}} \le \frac{\alpha}{2}.$$
    By the properties of Gaussian random variables, it follows that with probability at least $0.95$,
    \begin{equation*}
        T_1 
        =\Normt{W}
        = O\Paren{\frac{\rho^* \sqrt{d\log (1/\delta)}}{n\eps}}.
    \end{equation*}
    When $\rho^* = \sqrt{d/m}$,
    \begin{equation}
        T_1
        = O\Paren{\frac{d\sqrt{\log (1/\delta)}}{\sqrt{m}n\eps}}
        \le \frac{\alpha}{4},
    \end{equation}
   by assumption on $n$ and $n_0$. Likewise, when $\rho^* = \frac{n^{1/k}\eps^{1/k}d^{\frac{1}{2}-\frac{1}{k}}}{(\log 1/\delta)^{1/(2k)}m^{1-1/k}}$,
    \begin{equation}
         T_1 = O\Paren{\frac{d\sqrt{\log (1/\delta)}}{mn\eps}}^{\frac{k-1}{k}}
         \le \frac{\alpha}{4},
    \end{equation}
    by assumption on $n$ and $n_0$. By \Cref{lem:new-highd-clipping-bias} and assumption on $u^*$,    
    $$\Normt{\E\Brac{\mathsf{V_1}} - \mu} = \tilde{O}\Paren{\frac{\sqrt{d}}{m\rho^*}}^{k-1}.$$
    When $\rho^* = \sqrt{d/m}$, $T_3$ is dominated by $T_1$. Otherwise, when $\rho^* = \frac{n^{1/k}\eps^{1/k}d^{\frac{1}{2}-\frac{1}{k}}}{(\log 1/\delta)^{1/(2k)}m^{1-1/k}}$, $T_3$ and $T_1$ have the same asymptotics. Thus, 
    \begin{equation}
        \normt{\hat{\mu} - \mu} 
        \le T_1 + T_2 + T_3
        \le 2T_1 + T_2
        \le \frac{\alpha}{2} + \frac{\alpha}{2} 
        = \alpha,
    \end{equation}
    which completes the proof.
\end{proof}

We can now prove \Cref{thm:approx-dp-upperbound}
\begin{proof}[Proof of \Cref{thm:approx-dp-upperbound}]
We break the proof into two pieces. First, we prove that \Cref{alg:hd} satisfies person-level $(\eps,\delta)$-DP. Afterwords, we analyze the accuracy of \Cref{alg:hd}.

\textbf{Analysis of privacy.} Without loss of generality, we can assume that each of \Cref{line:u1}, \Cref{line:u2}, and \Cref{line:mu-hat} each receive as input the entire dataset $X$ and choose to only compute on a third of the dataset. By \Cref{thm:high-d-coarse-estimate}, \textrm{CoarseEstimate} satisfies person-level $\Paren{\frac{\eps}{2},\frac{\delta}{2}}$-DP. By \Cref{thm:iterative-clip}, \Cref{alg:itclip}, is $(\frac{\eps}{4},\frac{\delta}{4})$-DP. By \Cref{lem:clip-is-private}, \Cref{line:mu-hat} is person-level $(\frac{\eps}{4},\frac{\delta}{4})$-DP. Thus, by basic composition (\Cref{lem:privacy-composition}), \Cref{alg:hd} is person-level $(\eps,\delta)$-DP.

\textbf{Analysis of accuracy.} By \Cref{thm:high-d-coarse-estimate}, the estimate $u_0$ satisfies $\normt{u_0 - \mu} \le \frac{\sqrt{d}}{2\sqrt{m}}$ with probability at least $0.9$. Conditioned on this event, by \Cref{thm:iterative-clip}, $\normt{u^* - \mu} \le \frac{\rho^*}{2\sqrt{d}}$ with probability at least $0.9$. Conditioned on this event, by \Cref{lem:final-estimation}, $\normt{\hat{\mu} - \mu} \le \alpha$ with probability at least $0.9$. Thus, by a union bound, $\normt{\hat{\mu} - \mu} \le \alpha$ with probability at least $0.7$.
\end{proof}

\section{Mean Estimation in High Dimensions with Pure-DP}
\label{sec:pure_dp_ineff}

The section is split into two parts.
The first focuses on fine estimation, whereas the second integrates coarse estimation to the overall process and presents the full algorithm.

\subsection{Pure-DP Fine Estimation}
\label{subsec:pure_dp_fine}

The high-level overview of our fine estimation algorithm is the following.
We have a distribution $\cD$ over $\bR^d$ such that $\mu \coloneqq \E\limits_{X \sim \cD}[X]$ and $\E\limits_{X \sim \cD}\Brac{\Abs{\iprod{X - \mu, v}}^k} \le 1, \forall \|v\|_2 = 1$.
Since we're in the fine estimation setting, we assume that we have $\|\mu\|_{\infty} \le \alpha$, where $\alpha$ is our target error.
The idea is to cover the space of candidate means in a way that, for every candidate, there exists a point in the cover that is at $\ell_2$-distance at most $\alpha$ from it.
Then, around each point, we construct a ``local cover''.
The goal of the local cover is to help us examine whether our point is a good candidate for an estimate of the true mean by performing binary comparisons.
If the ``central point'' loses a comparison with an element of the local cover, this outcome is treated as a certificate that the true mean is far from that point.
Conversely, if the central point wins all comparisons, this implies that the point is indeed close to the true mean.

Implementing this privately involves instantiating the exponential mechanism in a way that is tailored to the above setting.
Given that each person is contributing a batch consisting of multiple samples, the score of a point is the minimum number of batches that need to be altered so that the candidate point loses at least one comparison with an element of its corresponding local cover.
In the rest of this section, the analysis will first focus on establishing the necessary facts which describe how comparisons behave, and then the focus will shift to establishing the utility guarantees of the exponential mechanism.

We start by analyzing how binary comparisons are performed between the central point and the elements of its local cover.
Our algorithm first calculates the empirical mean of each batch, and then projects the result on the line that connects the two points.
Then, the projections are averaged, and, depending on which point the average is closer to, the outcome of the comparison is determined.
Given that this will be used to instantiate the exponential mechanism, we need to ensure that the extent to which the outcome of the comparison is affected if one or multiple batches are altered is limited, thus guaranteeing bounded sensitivity of the score function.
For that reason, the mean of each batch is truncated around the point that is the candidate for the true mean that is under consideration at each stage of the algorithm's execution.

Analyzing the effect of the truncation is the first step to obtaining our result.
To do so, we will need the a variant of Lemma~$5.1$ from~\cite{KamathSU20}.
The lemma in question quantifies the bias error induced by the truncation, when the truncation center is far from the true mean.

\begin{lemma}
\label{lem:truncation_error}
Let $\cD$ be a distribution over $\R$ with mean $\mu$, and $k$th moment bounded by $1$.
Let $x_0, \rho \in \R$, with $\rho \coloneqq \Theta\Paren{\sqrt{\frac{\paren{k-1}\log m}{m}} + \frac{1}{m \alpha^{\frac{1}{k - 1}}}}$.
For $X \coloneqq (X_1, \dots, X_m) \sim \cD^{\otimes m}$, we define $Z$ be the following random variable:
\[
    Z \coloneqq
    \begin{cases}
        x_0 - \rho, & \text{if } \frac{1}{m} \sum\limits_{i = 1}^m X_i < x_0 - \rho \\
        \frac{1}{m} \sum\limits_{i = 1}^m X_i, & \text{if }\Abs{\frac{1}{m} \sum\limits_{i = 1}^m X_i - x_0} \le \rho \\
        x_0 + \rho, & \text{if }\frac{1}{m} \sum\limits_{i = 1}^m X_i > x_0 + \rho
    \end{cases}.
\]
Then, the following hold:
\begin{enumerate}
    \item If $x_0 > \mu + \frac{\rho}{2}$, we have:
    \[
        \E[Z] \in
        \begin{cases}
            \left[\mu - \frac{\rho}{8},\mu + \frac{\rho}{8}\right], & \text{if }\frac{\rho}{2} < \abs{x_0 - \mu} \leq \frac{17 \rho}{16} \\
            \left[x_0 - \rho, x_0 - \frac{15 \rho}{16}\right], & \text{if }\abs{x_0 - \mu} > \frac{17 \rho}{16}
        \end{cases}.
    \]
    \item If $x_0 < \mu - \frac{\rho}{2}$, we have:
    \[
        \E[Z] \in
        \begin{cases}
            \left[\mu - \frac{\rho}{8}, \mu + \frac{\rho}{8}\right], & \text{if }\frac{\rho}{2} < \abs{x_0 - \mu} \leq \frac{17 \rho}{16} \\
            \left[x_0 + \frac{15 \rho}{16}, x_0 + \rho\right], & \text{if }\abs{x_0 - \mu} > \frac{17 \rho}{16}
        \end{cases}.
    \]
\end{enumerate}
\end{lemma}

\begin{proof}
Without loss of generality, we assume that $x_0 > \mu$, since the other case is symmetric.
We define $a \coloneqq \max\{x_0 - \rho, \mu + \frac{\rho}{16}\}$ and $b \coloneqq x_0 + \rho$.
The previous implies that $\mu < a < b$.
To help us bound the value of $\E[Z]$, we also consider the probability $q \coloneqq \Pr\Brac{\Abs{\frac{1}{m} \sum\limits_{i = 1}^m X_i - \mu} > a - \mu}$.
For the upper bound, we have:
\begin{align}
    \E[Z]
    &= \E\left[\mathds{1}\left\{ \Abs{\frac{1}{m} \sum\limits_{i = 1}^m X_i - \mu} \le a - \mu \right\} Z\right] + \E\left[\mathds{1}\left\{ \Abs{\frac{1}{m} \sum\limits_{i = 1}^m X_i - \mu} > a - \mu \right\} Z\right] \nonumber \\
    &= (1 - q) \E\left[Z \middle| \Abs{\frac{1}{m} \sum\limits_{i = 1}^m X_i - \mu} \le a - \mu \right] + q \E\left[Z \middle| \Abs{\frac{1}{m} \sum\limits_{i = 1}^m X_i - \mu} > a - \mu \right] \nonumber \\
    &\le (1 - q) a + q b \nonumber \\
    &= a + (b - a) q, \label{eq:clipping_bias_ub}
\end{align}
where we used the property of conditional expectation that $\E[X \mathds{1}\{X \in A\}] = \Pr[X \in A] \E[X | X \in A]$.

Indeed, when the event $\Abs{\frac{1}{m} \sum\limits_{i = 1}^m X_i - \mu} > a - \mu \Leftrightarrow \Paren{\frac{1}{m} \sum\limits_{i = 1}^m X_i > a} \vee \Paren{\frac{1}{m} \sum\limits_{i = 1}^m X_i < 2 \mu - a}$ is realized, the value of $Z$ is maximized when one assumes that $\frac{1}{m} \sum\limits_{i = 1}^m X_i \geq x_0 + \rho = b$ holds with probability $q$.
Similarly, when $\Abs{\frac{1}{m} \sum\limits_{i = 1}^m X_i - \mu} \le a - \mu \Leftrightarrow 2 \mu - a \le \frac{1}{m} \sum\limits_{i = 1}^m X_i \le a$ is realized, the value of $Z$ is maximized when $\frac{1}{m} \sum\limits_{i = 1}^m X_i = a$ happens with probability $1 - q$.
Now, we take cases depending on the value of $a$.
If $a = \mu + \frac{\rho}{16} \Leftrightarrow \mu + \frac{\rho}{16} \geq x_0 - \rho \Leftrightarrow x_0 - \mu \le \frac{17}{16} \rho$, (\ref{eq:clipping_bias_ub}) becomes:
\begin{equation}
    \E[Z] \le \mu + \frac{\rho}{16} + 2 q \rho. \label{eq:clipping_bias_ub1}
\end{equation}
Conversely, if $a = x_0 - \rho \Leftrightarrow \mu + \frac{\rho}{16} < x_0 - \rho \Leftrightarrow x_0 - \mu > \frac{17}{16} \rho$, (\ref{eq:clipping_bias_ub}) yields:
\begin{equation}
    \E[Z] \le x_0 - \rho + 2 \rho q. \label{eq:clipping_bias_ub2}
\end{equation}
We observe that, in either of the two cases, in order to establish the desired upper bounds, it suffices to establish that $q \le \frac{1}{32}$.
Since $\rho = \Theta\Paren{\sqrt{\frac{\paren{k-1}\log m}{m}} + \frac{1}{m \alpha^{\frac{1}{k - 1}}}}$, Corollary~\ref{cor:non-uniform-tail-bound} yields:
\begin{align*}
    q = \Pr\Brac{\Abs{\frac{1}{m} \sum\limits_{i = 1}^m X_i - \mu} > \alpha - \mu} 
    &\le \Pr\Brac{\Abs{\frac{1}{m} \sum\limits_{i = 1}^m X_i - \mu} > \frac{\rho}{16}} \\
    &\le_k \frac{1}{m^{k - 1} \Paren{\frac{\rho}{16}}^k}.
\end{align*}
Due to our choice of $\rho$, we get that $q \le \frac{1}{32}$, so the desired upper bounds follow from (\ref{eq:clipping_bias_ub1}) and (\ref{eq:clipping_bias_ub2}).

We now turn our attention to establishing the lower bounds.
The bound $\E[Z] \geq x_0 - \rho$ holds trivially because of the truncation interval considered.
Thus, we focus on the case where $x_0 - \mu < \frac{17}{16} \rho$.
Based on a similar reasoning as the one used in the upper bound, we get:
\begin{align*}
    \E[Z] \geq (1 - q) (2 \mu - a) + q (x_0 - \rho) = (1 - q) \Paren{\mu - \frac{\rho}{16}} + q (x_0 - \rho)
    &= \mu - \frac{\rho}{16} + q (x_0 - \mu) - \frac{15 q \rho}{16} \\
    &\geq \mu - \frac{\rho}{8},
\end{align*}
where we used that $x_0 - \mu > 0$ and $q \le \frac{1}{32}$.
\end{proof}

The next step in the analysis involves reasoning about the \emph{sampling error}.
In particular, we have to consider how much the truncated sample mean of a batch deviates from the true mean of the truncated distribution.

\begin{lemma}
\label{lem:sample_error}
Let $\cD$ be a distribution over $\R$ with mean $\mu$, and $k$th moment bounded by $1$.
Assume we are given $n \geq \cO\Paren{\frac{\log\Paren{\frac{1}{\beta}}}{m \alpha^2}}$ independently-drawn batches of size $m$, i.e., $X^{(i)} \coloneqq \Paren{X^{(i)}_1, \dots, X^{(i)}_m} \sim \cD^{\otimes m}, \forall i \in [n]$.
For $x_0, \rho \in \R$, we define the random variables $Z_i$ as in Lemma~\ref{lem:truncation_error}.\footnote{$x_0, \rho$ do not need to be the same as in that lemma.}
We denote the mean of the above random variables by $\mu_{\mathrm{tr}}$.
Now, let $k = \Theta\Paren{\log\Paren{\frac{1}{\beta}}}$, so that $\frac{n}{k} \geq \frac{10}{m \alpha^2}$.
For $i = 1, \dots, k$, we define $Y^i \coloneqq \Paren{Z_{(i - 1) \frac{n}{k} + 1}, \dots Z_{i \frac{n}{k}}}$, and let $\mu_{\mathrm{tr}, i}$ be the empirical mean of $Y^i$.
Then, for $\Bar{\mu} \coloneqq \mathrm{Median}\Paren{\mu_{\mathrm{tr}, 1}, \dots, \mu_{\mathrm{tr}, k}}$, we will have:
\[
    \Pr[\Abs{\Bar{\mu} - \mu_{\mathrm{tr}}} > \alpha] \le \beta.
\]
In particular, at least $90 \%$ of $\Paren{\mu_{\mathrm{tr}, 1}, \dots, \mu_{\mathrm{tr}, k}}$ will be at most $\alpha$-far from $\Bar{\mu}$, except with probability $\beta$.
\end{lemma}

\begin{proof}
First, we note that Jensen's inequality implies:
\[
    \E\Brac{\Abs{X^i_j - \mu}^2} \le \E\Brac{\Abs{X^i_j - \mu}^k}^{\frac{2}{k}} \le 1.
\]
Additionally, looking at the second moment of the sample mean of each batch, we get:
\[
    \E\Brac{\Abs{\frac{1}{m} \sum\limits_{j = 1}^m X^{(i)}_j - \mu}^2} = \frac{1}{m^2} \sum\limits_{j = 1}^m \E\Brac{\Abs{X^{(i)}_j - \mu}^2} \le \frac{1}{m}.
\]
By Lemma~\ref{lem:second-moment-after-truncation}, we get that:
\[
    \E\Brac{\Abs{Z_i - \E\Brac{Z_i}}^2} \le \E\Brac{\Abs{\frac{1}{m} \sum\limits_{j = 1}^m X^{(i)}_j - \mu}^2} \le \frac{1}{m}.
\]
By Chebyshev's inequality and our assumption that $\frac{n}{k} \geq \frac{10}{m \alpha^2}$, we have that:
\[
    \Pr\Brac{\Abs{\mu_{\mathrm{tr}, i} - \mu_{\mathrm{tr}}} > \alpha} \le 0.1.
\]
Then, by a direct application of the \emph{median trick}, we amplify the probability of success to get the desired high probability guarantee, leading to the bound $n \geq \cO\Paren{\frac{\log\Paren{\frac{1}{\beta}}}{m \alpha^2}}$.
\end{proof}

Having established the above two lemmas which account for the effect of various sources of error, we now proceed to perform a step-by-step analysis of our algorithm.
This involves working in a bottom-up fashion, meaning that we will start by focusing on the binary comparison between one individual candidate mean and an element of its local cover.
We consider two cases based on whether we're truncating around a point that's close to the true mean or not.
For the first case, we have a lemma that establishes that, if we're truncating around a point $p$ that's ``close'' to the true mean, $p$ will win the comparison against any point $q$ that's far from the true mean with high probability, and there will be a lower bound on the number of batches that would have to be changed in order to alter the result of the comparison.

\begin{algorithm}[htb]
    \caption{Binary Mean Comparison}\label{alg:bmc}
    \hspace*{\algorithmicindent} \textbf{Input:} A dataset $X \coloneqq \Paren{X^{\Paren{1}}, \dots, X^{\Paren{n}}} \sim \Paren{\cD^{\otimes m}}^{\otimes n}$, points $p, q \in \bR^d$, target error $\alpha$, target failure probability $\beta$. \\
    \hspace*{\algorithmicindent} \textbf{Output:} \texttt{True, False}.
    \begin{algorithmic}[1]
    \Procedure{$\mathrm{BinMeanComp}_{p, q, \alpha, \beta}$}{$X$}
        \State Let $\rho = \Theta\Paren{\sqrt{\frac{\paren{k-1}\log m}{m}} + \frac{1}{m \alpha^{\frac{1}{k - 1}}}}$.
        \For{$i \in \Brac{n}$}
            \State Consider the $1$-d projection $\left\langle \frac{q - p}{\left\|q - p\right\|_2}, \frac{1}{m} \sum\limits_{j = 1}^m X^{\Paren{i}}_j \right\rangle$, and, for this projection, define $Z_i$ as in Lemma~\ref{lem:truncation_error} with truncation center $x_0 \coloneqq p$ and truncation radius $\rho$.
        \EndFor
        \State Let $p_0 \coloneqq \left\langle \frac{q - p}{\left\| q - p \right\|_2}, p \right\rangle$, and $q_0 \coloneqq \left\langle \frac{q - p}{\left\| q - p \right\|_2}, q \right\rangle$.
        \If {$\Bar{\mu} \le \frac{p_0 + q_0}{2}$} \Comment{$\Bar{\mu}$ is the median of means, as in Lemma~\ref{lem:sample_error}.}
            \State \Return \texttt{True}.
        \Else
            \State \Return \texttt{False}.
        \EndIf
    \EndProcedure
    \end{algorithmic}
\end{algorithm}

\begin{lemma}
\label{lem:score1}
Let $\cD$ be a distribution over $\R^d$ with mean $\mu$, and $k$th moment bounded by $1$.
Assume we are given $n \geq \cO\Paren{\frac{\log\Paren{\frac{1}{\beta}}}{m \alpha^2}}$ independently-drawn batches of size $m$, each denoted by $X^{(i)} \coloneqq \Paren{X^{(i)}_1, \dots, X^{(i)}_m} \sim \cD^{\otimes m}, \forall i \in [n]$.
Also, assume we are given a pair of points $p, q \in \R^d$ such that $\left\|p - \mu\right\|_2 \le \frac{\alpha}{8}$ and $\left\|p - q\right\|_2 > \alpha$.
Then, we have for Algorithm~\ref{alg:bmc}, with probability at least $1 - \beta$, the number of batches an adversary needs to change to make $p$ lose the comparison with $q$ is at least $\Theta\Paren{\frac{n \alpha}{\rho}}$.
\end{lemma}

\begin{proof}
Let $\mu_0 \coloneqq \left\langle \frac{q - p}{\left\|q - p\right\|_2}, \mu \right\rangle$ be the projection of $\mu$ along the line connecting $p$ and $q$.
We have that $\Abs{\mu_0 - p_0} \le \left\|\mu - p\right\|_2 \le \frac{\alpha}{8}$.
The mean $\mu_{0, \mathrm{tr}}$ of the samples $Z_i$ satisfies $\Abs{\mu_{0, \mathrm{tr}} - \mu_0} \le \frac{\alpha}{16}$ (by Corollary~\ref{cor:bias_error_batches_1d} and our choice of truncation radius), whereas we have $\Abs{\Bar{\mu} - \mu_{0, \mathrm{tr}}} \le \frac{\alpha}{16}$, except with probability $\beta$ (by Lemma~\ref{lem:sample_error} for $\alpha \to \frac{\alpha}{16}$).
By triangle inequality, we get overall that $\Abs{p_0 - \Bar{\mu}} \le \frac{\alpha}{8} + \frac{\alpha}{16} + \frac{\alpha}{16} = \frac{\alpha}{4}$.
We note that $\Abs{p_0 - q_0} = \left\|p - q\right\|_2 > \alpha$.
Thus, for $q$ to win the comparison instead of $p$, we need to shift $\Bar{\mu}$ by at least $\frac{\alpha}{4}$.
By Lemma~\ref{lem:sample_error} we know that, for an adversary to achieve this, they need to corrupt at least $0.4 k = \Theta\Paren{\log\Paren{\frac{1}{\beta}}}$ subsamples in order to shift their empirical mean by $\frac{\alpha}{4}$.
Focusing on an individual subsample, corrupting a single batch can shift the mean by at most $\frac{2 \rho k}{n}$.
Thus, in each subsample, the adversary needs to corrupt at least $\frac{n \alpha}{8 \rho k}$ batches, implying at least $\Theta\Paren{\frac{n \alpha}{\rho}}$ batches corrupted overall.
\end{proof}

Our second lemma considers the case of truncating around a point $p$ that's far from the true mean, whereas we have a point $q$ that's close to the true mean.
In this case, there are two possible outcomes, depending on how far the true mean $\mu$ is from $p$.
If $\mu$ (and $q$ with it) is very far from $p$, what might happen is that, due to the effect of the truncation, the truncated mean might end up being closer to $p$ than $q$.
If the previous doesn't happen, we show that $p$ loses with high probability.
Conversely, if it happens, we argue that there exists a set of points $q'$ such that, if we run Algorithm~\ref{alg:bmc} with input $p$ and $q'$, $p$ will lose the comparison with high probability.

\begin{lemma}
\label{lem:score2}
Let $\cD$ be a distribution over $\R^d$ with mean $\mu$, and $k$th moment bounded by $1$.
For $\alpha \le \cO\Paren{\frac{1}{m^{\frac{k - 1}{k}}}}$, we are given $n \geq \cO\Paren{\frac{\log\Paren{\frac{1}{\beta}}}{m \alpha^2}}$ independently-drawn batches of size $m$, each denoted by $X^{(i)} \coloneqq \Paren{X^{(i)}_1, \dots, X^{(i)}_m} \sim \cD^{\otimes m}, \forall i \in [n]$.
Assume we are given a pair of points $p, q \in \R^d$ such that $\left\|p - \mu\right\|_2 > \frac{9 \alpha}{8}$ and $\left\|q - \mu\right\|_2 \le \frac{\alpha}{8}$.
Then, at least one of the following occurs:
\begin{itemize}
    \item $p$ loses the comparison to $q$ with probability at least $1 - \beta$.
    \item Let $\widetilde{\mu} \coloneqq p + \frac{\rho}{2} \cdot \frac{\mu - p}{\left\|\mu - p\right\|_2}$, and let $q'$ be any point in $\mathbb{R}^d$ such that $\left\|q' - \widetilde{\mu}\right\|_2 \le \frac{\alpha}{8}$.
    Then, if we run Algorithm~\ref{alg:bmc} with points $p$ and $q'$ as input, $p$ will lose with probability at least $1 - \beta$.
\end{itemize}
\end{lemma}

\begin{proof}
As in the proof of Lemma~\ref{lem:score1}, we have $\mu_0 \coloneqq \left\langle \frac{q - p}{\left\|q - p\right\|_2}, \mu \right\rangle$ be the projection of $\mu$ along the line connecting $p$ and $q$.
We start by noting that $\left\|p - q\right\|_2 > \alpha$ holds because of the triangle inequality.
Additionally, we have $\Abs{q_0 - \mu_0} \le \left\|q - \mu\right\|_2 \le \frac{\alpha}{8}$.
We now need to consider two cases, depending on how $\Abs{p_0 - \mu_0}$ compares with $\frac{\rho}{2}$.

\begin{itemize}
    \item $\Abs{p_0 - \mu_0} \le \frac{\rho}{2}$ (Figure~\ref{fig:easy_case}).
    By Corollary~\ref{cor:bias_error_batches_1d} and our choice of truncation radius, we get that $\Abs{\mu_{0, \mathrm{tr}} - \mu_0} \le \frac{\alpha}{16}$.
    Additionally, Lemma~\ref{lem:sample_error} yields that $\Abs{\Bar{\mu} - \mu_{0, \mathrm{tr}}} \le \frac{\alpha}{16}$ with probability at least $1 - \beta$.
    Thus, by triangle inequality we get $\Abs{\Bar{\mu} - q_0} \le \frac{\alpha}{4} < \frac{\alpha}{2} \le \frac{\left\|p - q\right\|_2}{2}$.
    This implies that $p$ will lose to $q$, yielding the desired result.
        
    \begin{figure}[htb]
        \centering
        \includegraphics[scale=0.3]{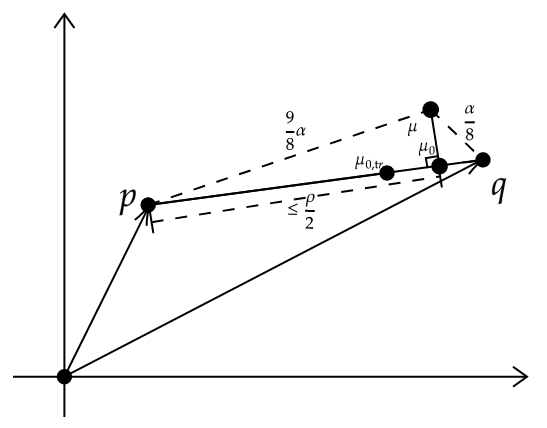}
        \caption{$\left\|p - \mu\right\|_2 > \frac{9 \alpha}{8}, \left\|q - \mu\right\|_2 \le \frac{\alpha}{8}$, and $\Abs{p_0 - \mu_0} \le \frac{\rho}{2}$.}
        \label{fig:easy_case}
    \end{figure}
        
    \item $\Abs{p_0 - \mu_0} > \frac{\rho}{2}$.
    In this case, the triangle inequality implies that $\Abs{p_0 - q_0} > \frac{\rho}{2} - \frac{\alpha}{8} \geq \frac{3 \rho}{8}$, where the last inequality follows from the assumption $\alpha \le \cO\Paren{\frac{1}{m^{\frac{k - 1}{k}}}}$.
    Now, we have to consider two cases depending on how $\Abs{p_0 - q_0}$ compares with $\frac{17 \rho}{16}$.
        
    When $\Abs{p_0 - \mu_0} \le \frac{17 \rho}{16}$ (Figure~\ref{fig:intermediate_case}), Lemma~\ref{lem:truncation_error} implies $\Abs{\mu_0 - \mu_{0, \mathrm{tr}}} \le \frac{\rho}{8}$.
    Additionally, we have by Lemma~\ref{lem:sample_error} that $\Abs{\Bar{\mu} - \mu_{0, \mathrm{tr}}} \le \frac{\alpha}{16}$ with probability at least $1 - \beta$.
    The triangle inequality yields that $\Abs{q_0 - \Bar{\mu}} \le \frac{\alpha}{8} + \frac{\rho}{8} + \frac{\alpha}{16} = \frac{\rho}{8} + \frac{3 \alpha}{16} \le \frac{3 \rho}{16} < \frac{\left\|p - q\right\|_2}{2}$.
    This immediately yields that $p$ loses the comparison to $q$.

    \begin{figure}[htb]
        \centering
        \includegraphics[scale=0.3]{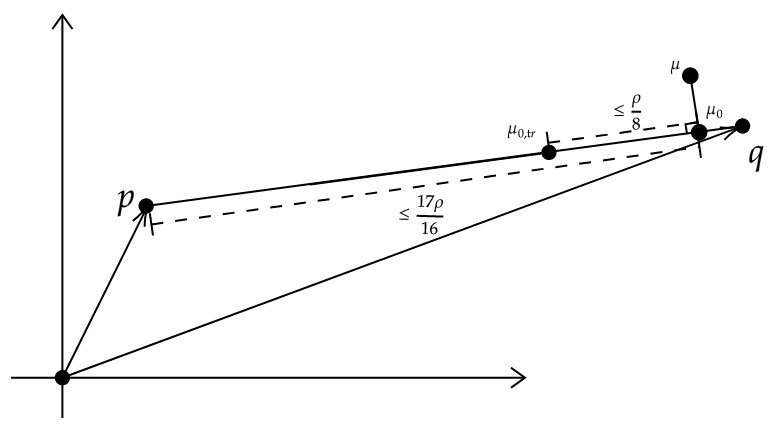}
        \caption{$\Abs{p_0 - \mu_0} > \frac{\rho}{2}$, and $\Abs{p_0 - \mu_0} \le \frac{17 \rho}{16}$.}
        \label{fig:intermediate_case}
    \end{figure}

    When $\Abs{p_0 - \mu_0} > \frac{17 \rho}{16}$, we get that $\Abs{p_0 - \mu_{0, \mathrm{tr}}} \in \left(\frac{15 \rho}{16}, \rho\right]$.
    We note that, by definition, we have $q_0 > p_0$, so $q_0 - p_0 = \left\|q - p\right\|_2 \geq \frac{3 \rho}{8}$.
    Combining this with the fact that $\Abs{q_0 - \mu_0} \le \frac{\alpha}{8}$, we get that $\mu_0 - p_0 \geq \frac{3 \rho}{8} - \frac{\alpha}{8} > 0 \implies \mu_0 > p_0 \implies \mu_{0, \mathrm{tr}} > p_0 + \frac{15 \rho}{16}$.
    By Lemma~\ref{lem:sample_error}, we have that $\Abs{\Bar{\mu} - \mu_{0, \mathrm{tr}}} \le \frac{\alpha}{16}$ with probability at least $1 - \beta$, implying that $\Bar{\mu} > p_0 + \frac{15}{16} \Paren{\rho - \alpha}$.
    If $\Abs{q_0 - \Bar{\mu}} \le \frac{15}{16} \Paren{\rho - \alpha}$, then $q$ wins and $p$ loses.
    However, if that does not happen, we have to work differently.
    
    For the case where $\Abs{q_0 - \Bar{\mu}} > \frac{15}{16} \Paren{\rho - \alpha}$, let $q'$ be any point that is at $\ell_2$-distance at most $\frac{\alpha}{8}$ from $\widetilde{\mu} = p + \frac{\rho}{2} \cdot \frac{\mu - p}{\left\|\mu - p\right\|_2}$ (Figure~\ref{fig:hard_case}).

    \begin{figure}[htb]
        \centering
        \includegraphics[scale=0.3]{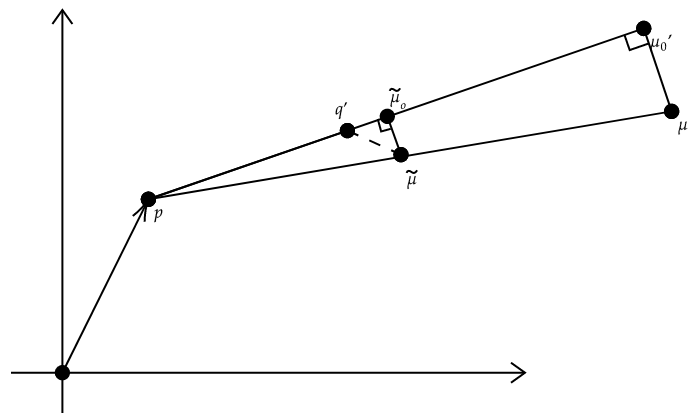}
        \caption{$\Abs{q_0 - \Bar{\mu}} > \frac{15}{16} \Paren{\rho - \alpha}$.}
        \label{fig:hard_case}
    \end{figure}

    We project everything along the line that connects $p$ and $q'$, so we get $p_0' \coloneqq \left\langle \frac{q' - p}{\left\|q' - p\right\|_2}, p \right\rangle, q_0' \coloneqq \left\langle \frac{q' - p}{\left\|q' - p\right\|_2}, q' \right\rangle, \mu_0' \coloneqq \left\langle \frac{q' - p}{\left\|q' - p\right\|_2}, \mu \right\rangle$, and $\widetilde{\mu}_0 \coloneqq \left\langle \frac{q' - p}{\left\|q' - p\right\|_2}, \widetilde{\mu} \right\rangle$.
    By triangle inequality, we have $\Abs{q'_0 - p'_0} = \left\|q' - p\right\|_2 \geq \left\|\widetilde{\mu} - p\right\|_2 - \left\|q' - \widetilde{\mu}\right\|_2 \geq \frac{\rho}{2} - \frac{\alpha}{8} \geq \frac{3 \rho}{8}$.
    Additionally, by triangle similarity, we have:
    \begin{align*}
        &\qquad\quad \frac{\Abs{\widetilde{\mu}_0 - p_0'}}{\Abs{\mu_0' - p_0'}} = \frac{\left\|p - \widetilde{\mu}\right\|_2}{\left\|p - \mu\right\|_2} \\
        &\Leftrightarrow \Abs{\mu_0' - p_0'} = \Abs{\widetilde{\mu}_ 0 - p_0'} \frac{\left\|p - \mu\right\|_2}{\left\|p - \widetilde{\mu}\right\|_2} = \Abs{\widetilde{\mu}_ 0 - p_0'} \frac{\Abs{p_0 - \mu_0}}{\left\|p - \widetilde{\mu}\right\|_2} = \Abs{\widetilde{\mu}_ 0 p- p_0'} \frac{\frac{17}{16} \rho}{\frac{\rho}{2}} = \frac{17}{32} \Abs{\widetilde{\mu}_ 0 - p_0'},
    \end{align*}
    where, by the Pythagorean theorem $\Abs{p_0' - \widetilde{\mu}_0} = \sqrt{\left\|p - \widetilde{\mu}\right\|_2^2 - \left\| \widetilde{\mu} - \Paren{p + \widetilde{\mu}_0 \frac{q' - p}{\|q' - p\|_2}}\right\|_2^2} \geq \sqrt{\left\|p - \widetilde{\mu}\right\|_2^2 - \left\| \widetilde{\mu} - q'\right\|_2^2} \geq \sqrt{\left\|p - \widetilde{\mu}\right\|_2^2 - \Paren{\frac{\alpha}{8}}^2} = \sqrt{\left(\frac{\rho}{2}\right)^2 - \Paren{\frac{\alpha}{8}}^2}$.
    As a result of all the previous, we get that $\Abs{\mu_0' - p_0'}$ is sufficiently large, so we can work similarly to the previous cases, but $p$ is guaranteed to lose this time. 
\end{itemize}
\end{proof}

Assuming that $\left\|\mu\right\|_{\infty} \le \alpha$, our goal is to use the above lemma to analyze the comparisons between a point $p$ with $\left\|p\right\|_{\infty} \le \alpha$ and the elements of a local cover of the $\ell_{\infty}$-ball of radius $2 \alpha$ that is centered around $p$.
We will remove from the cover all the points that are at $\ell_2$-distance at most $\alpha$ from $p$.
Our goal is to plug the above into the exponential mechanism in order to determine whether a point $p$ is close to the true mean of our distribution.
For that reason, we will need to define a \emph{score function}, which we do formally below:

\begin{definition}[$\mathrm{Score}$ of a Point]
\label{def:score}
Let $X \coloneqq \Paren{X^{\Paren{1}}, \dots, X^{\Paren{n}}}$ be a dataset consisting of $n$ batches $X^{\Paren{i}} \in \bR^{m \times d}, \forall i \in \Brac{d}$, and $p$ be a point in $\R^d$, and $\alpha > 0$.
We define $\mathrm{Score}_{X, D, \alpha}\Paren{p}$ of $p$ with respect to a domain $D \subset \R^d$ to be the minimum number of batches of $X$ that need to be changed to get a dataset $\wb{X}$ so that there exists $q \in D$, such that Algorithm~\ref{alg:bmc} outputs \texttt{False}.
If for all $q \in D \setminus \{p\}$ and all $X$, Algorithm~\ref{alg:bmc} outputs \texttt{True}, then we define $\mathrm{Score}_{X, D, \alpha}\Paren{p} = n \alpha$.
\noindent If the context is clear, we abbreviate the quantity to $\mathrm{Score}_X\paren{p}$.
\end{definition}

At this point, we recall the following lemma from~\cite{KamathSU20}:

\begin{lemma}
\label{lem:score_sensitivity}
The $\mathrm{Score}$ function satisfies the following:
\[
    \Delta_{\mathrm{Score}, 1} \le 1.
\]
\end{lemma}

We note that this lemma was established in the setting of \emph{item-level privacy}, where the score is defined in terms of points that have to be changed for a point $p$ to lose a comparison, but the result still holds if instead of individual points we consider batches.

Algorithm~\ref{alg:test} implements the above reasoning.
The algorithm outputs a value that corresponds to the score of $p$, as defined in Definition~\ref{def:score}.
Its performance guarantees will be analyzed by using Lemmas~\ref{lem:score1} and~\ref{lem:score2}.

\begin{algorithm}[htb]
    \caption{Test Candidate}\label{alg:test}
    \hspace*{\algorithmicindent} \textbf{Input:} A dataset $X \coloneqq \Paren{X^{\Paren{1}}, \dots, X^{\Paren{n}}} \sim \Paren{\cD^{\otimes m}}^{\otimes n}$, a point $p \in \bR^d$, target error $\alpha$, target failure probability $\beta$. \\
    \hspace*{\algorithmicindent} \textbf{Output:} $s \in \mathbb{N}$.
    \begin{algorithmic}[1]
    \Procedure{$\mathrm{TestCan}_{p, \alpha, \beta}$}{$X$}
        \For{$i \in \Brac{d}$}
            \State Let $J_{p, i} \coloneqq \left\{p_i - 2 \alpha, p_i - 2 \alpha + \frac{\alpha}{4 \sqrt{d}}, \dots, p_i + 2 \alpha - \frac{\alpha}{4 \sqrt{d}}, p_i + 2 \alpha\right\}$.
        \EndFor
        \State Let $J_p \coloneqq \left(\bigotimes\limits_{i \in [d]} J_{p, i}\right) \setminus \left\{x \in \mathbb{R}^d \colon \left\|x - p\right\|_2 \le \alpha \right\}$.
        \State \Return $\mathrm{Score}_{X, J_p, \alpha}\left(X\right)$.
    \EndProcedure
    \end{algorithmic}
\end{algorithm}

\begin{lemma}
\label{lem:test_candidate}
Let $\cD$ be a distribution over $\R^d$ with mean $\left\|\mu\right\|_{\infty} \le \alpha \le \cO\Paren{\frac{1}{m^{\frac{k - 1}{k}}}}$, and $k$th moment bounded by $1$.
Assume we are given $n \geq \cO\Paren{\frac{d \log\Paren{\frac{d}{\beta}}}{m \alpha^2}}$ independently-drawn batches of size $m$, each denoted by $X^{(i)} \coloneqq \Paren{X^{(i)}_1, \dots, X^{(i)}_m} \sim \cD^{\otimes m}, \forall i \in [n]$.
For any $p \in \mathbb{R}^d$ such that $\left\|p\right\|_{\infty} \le \alpha$, we have:
\begin{itemize}
    \item If $\left\|p - \mu\right\|_2 \le \frac{\alpha}{8}$, we have for Algorithm~\ref{alg:test} that, with probability at least $1 - \beta$, it will output $\mathrm{Score}_{X, J_p, \alpha}\left(X\right) \geq \Theta\Paren{\frac{n \alpha}{\rho}}$.
    \item If $\left\|p - \mu\right\|_2 > \frac{9 \alpha}{8}$, we have for Algorithm~\ref{alg:test} that, with probability at least $1 - \beta$, it will output $\mathrm{Score}_{X, J_p, \alpha}\left(X\right) = 0$.
\end{itemize}
\end{lemma}

\begin{proof}
The proof follows almost immediately from Lemmas~\ref{lem:score1} and~\ref{lem:score2}.

For the first case, the fact that Algorithm~\ref{alg:test} constructs $J_p$ in a way that ensures that points $q$ considered are at least $\alpha$-far from $p$ suffices for us to be able to apply Lemma~\ref{lem:score1}.
Given that we want to apply the lemma for all points $q$, we need to set $\beta \to \frac{\beta}{\left(\Theta\left(\sqrt{d}\right)\right)^{d}}$.
Indeed, by a union bound, the probability of getting the wrong result for at least one comparison is upper bounded by $\beta$, leading to the bound $n \geq \cO\Paren{\frac{d \log\Paren{\frac{d}{\beta}}}{m \alpha^2}}$.

For the second case, we note that $J_p$ is an $\frac{\alpha}{8 \sqrt{d}}$-cover with respect to the $\ell_{\infty}$-distance of the $\ell_{\infty}$-ball with radius $2 \alpha$ that's centered at $p$ (modulo the $\ell_2$-ball centered around $p$, which we've removed).
We want to argue that there must exist a point $q \in J_p$ such that $\left\| q - \mu \right\|_{\infty} \le \frac{\alpha}{8 \sqrt{d}}$.
This would follow immediately from the definition of $J_p$, had the points that lie at the $\ell_2$-ball of radius $\alpha$ that's centered at $p$ not been removed.
Thus, we need to argue that any point $q$ which satisfies $\left\| q - \mu \right\|_{\infty} \le \frac{\alpha}{8 \sqrt{d}}$ must also satisfy $\left\| p - q \right\|_2 > \alpha$.
We have that $\left\| q - \mu \right\|_{\infty} \le \frac{\alpha}{8 \sqrt{d}} \implies \left\| q - \mu \right\|_2 \le \frac{\alpha}{8}$.
Since $\left\|p - \mu\right\|_2 > \frac{9 \alpha}{8}$, the triangle inequality yields that $\left\|p - q\right\|_2 > \frac{9 \alpha}{8} - \frac{\alpha}{8} = \alpha$, leading to the desired result.
Additionally, the previous guarantee that the conditions of Lemma~\ref{lem:score2} are satisfied, and setting again $\beta \to \left(\Theta\left(\sqrt{d}\right)\right)^d$ completes the proof.
\end{proof}

We now have all the necessary tools to present the complete algorithm for fine estimation.
Algorithm~\ref{alg:fine_est} constructs an $\ell_{\infty}$-cover of the set $\{ x \in \mathbb{R}^d \colon \left\| x \right\|_{\infty} \le \alpha \}$ that has granularity $\frac{2 \alpha}{\sqrt{d}}$, and then calculates the score of all points of the cover and samples from the exponential mechanism.

\begin{algorithm}[htb]
    \caption{Fine Estimation}\label{alg:fine_est}
    \hspace*{\algorithmicindent} \textbf{Input:} A dataset $X \coloneqq \Paren{X^{\Paren{1}}, \dots, X^{\Paren{n}}} \sim \Paren{\cD^{\otimes m}}^{\otimes n}$, target error $\alpha$, target failure probability $\beta$, privacy parameter $\eps$. \\
    \hspace*{\algorithmicindent} \textbf{Output:} $\widehat{\mu} \in \bR^d$.
    \begin{algorithmic}[1]
    \Procedure{$\mathrm{FineEst}_{\alpha, \beta, \eps}$}{$X$}
        \For{$i \in \Brac{d}$}
            \State Let $J_i \coloneqq \left\{- \alpha, - \alpha + \frac{2 \alpha}{\sqrt{d}}, \dots, \alpha - \frac{2 \alpha}{\sqrt{d}}, \alpha\right\}$.
        \EndFor
        \State Let $J \coloneqq \bigotimes\limits_{i \in [d]} J_i$.
        \For{$p \in J$}
            \State Run $\mathrm{TestCan}_{p, \frac{8 \alpha}{9}, \frac{\beta}{2 \left(\Theta\left(\sqrt{d}\right)\right)^d}}(X)$ to calculate $\mathrm{Score}_{X, J_p, \alpha}\left(X\right)$.
        \EndFor
        \State Run the Exponential Mechanism with score function $\{\mathrm{Score}_{X, J_p, \alpha}\left(X\right)\}_{p \in J}$ and privacy budget $\eps$, to get a point $\widehat{\mu}$.
        \State \Return $\widehat{\mu}$.
    \EndProcedure
    \end{algorithmic}
\end{algorithm}

\begin{lemma}
\label{lem:fine_est}
Let $\cD$ be a distribution over $\R^d$ with mean $\left\|\mu\right\|_{\infty} \le \alpha \le \cO\Paren{\frac{1}{m^{\frac{k - 1}{k}}}}$, and $k$th moment bounded by $1$.
Assume we are given $n \geq \cO_k\Paren{\frac{d \log\Paren{\frac{d}{\beta}}}{m \alpha^2} + \frac{d \log\Paren{\frac{d}{\beta}} \sqrt{\log(m)}}{\sqrt{m} \alpha \eps} + \frac{d \log\Paren{\frac{d}{\beta}}}{m \alpha^{\frac{k}{k - 1}} \eps}}$ independently-drawn batches of size $m$, each denoted by $X^{(i)} \coloneqq \Paren{X^{(i)}_1, \dots, X^{(i)}_m} \sim \cD^{\otimes m}, \forall i \in [n]$.
Then, for $\eps, \beta > 0$, there exists an $\eps$-DP mechanism (Algorithm~\ref{alg:fine_est}) which, given $\Paren{X^{(1)}, \dots X^{(n)}}$ as input, outputs a point $\widehat{\mu}$ such that $\left\|\widehat{\mu} - \mu \right\|_2 \le \alpha$ with probability at least $1 - \beta$.
\end{lemma}

\begin{proof}
The privacy guarantee is an immediate consequence of the privacy guarantee of the exponential mechanism (Lemma~\ref{lem:exp_mech}), so we focus on the accuracy guarantee.
By the definition of the set $J$, we have that there must exist a point $p \in J$ such that $\left\|p - \mu\right\|_{\infty} \le \frac{\alpha}{\sqrt{d}} \implies \left\|p - \mu\right\|_2 \le \alpha$.
Thus, what we must do is argue that, thanks to our choice of parameters, and the guarantees of the exponential mechanism imply that a point with $\left\|p - \mu\right\|_2 \le \alpha$ will be chosen with high probability.
First, we note that Algorithm~\ref{alg:fine_est} uses Algorithm~\ref{alg:test} with target error $\frac{8 \alpha}{9}$, and probability $\frac{\beta}{2 \left(\Theta\left(\sqrt{d}\right)\right)^d}$.
Due to the number of batches $n$ that we have, as well as the guarantees of Lemma~\ref{lem:test_candidate} we have that, for each point $p \in J$, except with probability $\frac{\beta}{2}$, the score of a point $p$ with $\left\|p - \mu\right\|_2 \le \frac{\alpha}{9}$ will be $\Theta\left(\frac{n \alpha}{\rho}\right)$, whereas the score of a point $p$ with $\left\|p - \mu\right\|_2 > \alpha$ will be $0$.
We have by Lemma~\ref{lem:score_sensitivity} that the sensitivity of our score function is at most $1$.
Thus, thanks to the guarantees of the exponential mechanism (Lemma~\ref{lem:exp_mech}), we have that, except with probability at least $\frac{\beta}{2}$, we have:
\begin{align*}    
    \mathrm{Score}(X, \widehat{\mu})
    &\geq \mathrm{OPT}_{\mathrm{Score}}(X) - \frac{2 \Delta_{\mathrm{Score}, 1}}{\eps} (\ln(|S|) + t) \\
    &\geq \Theta\left(\frac{n \alpha}{\rho}\right) - \frac{2}{\eps} \left(\Theta\left(d \log(d)\right) + \log\left(\frac{2}{\beta}\right)\right) \\
    &\geq \Theta\left(\frac{n \alpha}{\sqrt{\frac{\paren{k-1}\log m}{m}} + \frac{1}{m \alpha^{\frac{1}{k - 1}}}}\right) - \frac{2}{\eps} \left(\Theta\left(d \log(d)\right) + \log\left(\frac{2}{\beta}\right)\right)
\end{align*}
We now need to consider two cases, depending on which term dominates in the denominator.
If $\sqrt{\frac{\paren{k-1}\log m}{m}} \le \frac{1}{m \alpha^{\frac{1}{k - 1}}}$, the previous can be lower-bounded by:
\[
    \Theta\left(n \alpha^{\frac{k}{k - 1}} m\right) - \frac{2}{\eps} \left(\Theta\left(d \log(d)\right) + \log\left(\frac{2}{\beta}\right)\right)
\]
which, by assumption, is greater than $0$.

Conversely, if $\sqrt{\frac{\paren{k-1}\log m}{m}} > \frac{1}{m \alpha^{\frac{1}{k - 1}}}$, we get the lower bound:
\[
    \Theta\left(n \alpha \sqrt{\frac{m}{(k - 1) \log(m)}} \right) - \frac{2}{\eps} \left(\Theta\left(d \log(d)\right) + \log\left(\frac{2}{\beta}\right)\right),
\]
which, by assumption again, is greater than $0$.

Thus, except with probability $\beta$, the exponential mechanism will output a point with score greater than $0$, implying that the point will be at distance at most $\alpha$ from the true mean.
\end{proof}

\subsection{Coarse Estimation and the Full Algorithm}
\label{subsec:everything_together}

After focusing on fine estimation in the previous section, we can reason about coarse estimation, and tie everything together.
Our coarse estimator is not a new algorithm: it consists of a component-wise application of the single-dimensional mean estimator that is implied by Theorem~\ref{cor:mean-estimation-bounded}.
In particular, given sufficiently many samples from a distribution with $k$th moment bounded by $1$ that has mean $\mu$, the estimator of Theorem~\ref{cor:mean-estimation-bounded} outputs a $\widehat{\mu}$ such that $|\widehat{\mu} - \mu| \le \alpha$.
Now, given samples from a distribution in $d$ dimensions with mean $\mu$ and $k$th moment bounded by $1$, we can apply this algorithm independelty for each coordinate, and obtain an estimate $\widehat{\mu}$ such that $\|\widehat{\mu} - \mu\|_{\infty} \le \alpha$.
This allows us to reduce to the case where $\|\mu\|_{\infty} \le \alpha$, which was addressed in Section~\ref{subsec:pure_dp_fine}.
Algorithm~\ref{alg:full_algo} presents the full pseudocode that handles both coarse and fine estimation.

\begin{algorithm}[htb]
    \caption{Person-Level Pure-DP High-Dimensional Mean Estimation}\label{alg:full_algo}
    \hspace*{\algorithmicindent} \textbf{Input:} A dataset $X \coloneqq \Paren{X^{\Paren{1}}, \dots, X^{\Paren{2 n}}} \sim \Paren{\cD^{\otimes m}}^{\otimes 2 n}$, target error $\alpha$, target failure probability $\beta$, privacy parameter $\eps$. \\
    \hspace*{\algorithmicindent} \textbf{Output:} $\widehat{\mu} \in \bR^d$.
    \begin{algorithmic}[1]
    \Procedure{$\mathrm{PersonLevelMeanEst}_{\alpha, \beta, \eps}$}{$X$}
        \For{$k \in \Brac{d}$}
            \State Let $Y_i \coloneqq \Paren{\Paren{X^{(i)}_1}_k, \dots, \Paren{X^{(i)}_m}_k} \in \bR^n, \forall i \in [n]$ be the batch consisting of the $k$-th component of each element of the $i$-th batch.
            \State Run the algorithm of Theorem~\ref{cor:mean-estimation-bounded} over $\Paren{Y_1, \dots, Y_n}$ with target error $\alpha$, target failure probability $\frac{\beta}{2 d}$, and privacy parameter $\frac{\eps}{d}$, to obtain $\mu_{\mathtt{coarse}, i}$.
        \EndFor
        \State Let $\mu_{\mathtt{coarse}} \coloneqq \left(\mu_{\mathtt{coarse}, 1}, \dots, \mu_{\mathtt{coarse}, d}\right)$.
        \For{$i \in \{n+1, \dots, 2n\}$}
            \State Let $Z_{i - n}$ be the batch that is obtained by subtracting $\mu_{\mathtt{coarse}}$ from each element of $X^{(i)}$.
        \EndFor
        \State Let $Z \coloneqq \left(Z_1, \dots, Z_n\right)$
        \State Let $\widehat{\mu} \coloneqq \mathrm{FineEst}_{\alpha, \frac{\beta}{2}, \eps}(Z) + \mu_{\mathtt{coarse}}$.
        \State \Return $\widehat{\mu}$.
    \EndProcedure
    \end{algorithmic}
\end{algorithm}

\begin{theorem}
\label{thm:mean_est_pure_dp}
Let $\cD$ be a distribution over $\R^d$ with mean $\mu$ such that $\|\mu\|_2 \le R$, and $k$th moment bounded by $1$.
Assume we are given $n \geq \cO_k\Paren{\frac{d \log\Paren{d}}{m \alpha^2 \beta} + \frac{d \log\Paren{\frac{d}{\beta}} \sqrt{\log(m)}}{\sqrt{m} \alpha \eps} + \frac{d \log\Paren{\frac{d}{\beta}}}{m \alpha^{\frac{k}{k - 1}} \eps} + \frac{d \log \Paren{\frac{d R m}{\beta}}}{\eps}}$ independently-drawn batches of size $m$, each denoted by $X^{(i)} \coloneqq \Paren{X^{(i)}_1, \dots, X^{(i)}_m} \sim \cD^{\otimes m}, \forall i \in [n]$.
Then, for $\eps, \alpha, \beta > 0$ with $\alpha \le \cO\Paren{\frac{1}{m^{\frac{k - 1}{k}}}}$, there exists an $\eps$-DP mechanism (Algorithm~\ref{alg:fine_est}) which, given $\Paren{X^{(1)}, \dots X^{(n)}}$ as input, outputs a point $\widehat{\mu}$ such that $\left\|\widehat{\mu} - \mu \right\|_2 \le \alpha$ with probability at least $1 - \beta$.
\end{theorem}

\begin{proof}
We start by establishing the privacy guarantee first.
By the privacy guarantee of Theorem~\ref{cor:mean-estimation-bounded}, and basic composition (Lemma~\ref{lem:basic-composition}), we get that we have $\eps$-DP over the batches $(X^{(1)}, \dots, X^{(n)})$.
The guarantee is not affected when we later construct $\mu_{\mathtt{coarse}}$ and subtract it from the other datapoints, due to closure under post-processing (Lemma~\ref{lem:postprocessing}).
By the privacy guarantee of Lemma~\ref{lem:fine_est}, we are guaranteed $\eps$-DP for the batches $X^{(n + 1)}, \dots X^{(2 n)}$.
The overall privacy guarantee then follows from parallel composition.

Now, it remains to establish the accuracy guarantee.
By the guarantees of Theorem~\ref{cor:mean-estimation-bounded}, we get that, except with probability $\frac{\beta}{2}$, we will have $\|\mu_{\mathtt{coarse}} - \mu\|_{\infty} \le \alpha$.
Thus, by the guarantees of Lemma~\ref{lem:fine_est}, we get that, except with probability $\frac{\beta}{2}$, we will get a $\widehat{\mu}$ that satisfies $\left\|\widehat{\mu} - (\mu - \mu_{\mathtt{coarse}})\right\|_2 = \left\|(\widehat{\mu} + \mu_{\mathtt{coarse}}) - \mu\right\|_2 \le \alpha$.
The desired accuracy guarantee follows directly from the last inequality and a union bound over failure events.
\end{proof}

\section{Approximate-DP Lower Bounds for Bounded \texorpdfstring{$k$}{k}-th Moments}
\label{sec:approx_dp_lb}

In this section, we establish lower bounds for private mean estimation of distributions with bounded $k$th moments in the person-level setting, under the constraint of $(\eps, \delta)$-DP.
Our bounds nearly match the upper bounds of Section~\ref{sec:hd}.
We start by stating the main theorem of the section.\footnote{We note that this differs from Theorem~\ref{thm:lower-bound} in that it omits the $\log(1/\delta)/\varepsilon$ term. 
This is a known lower bound for the set of all point-mass distributions in the item-level case, and observe that the problem is unchanged in the person-level setting (regardless of the value of $m$).}

\begin{theorem}
\label{thm:main_lb_approx_dp}
Let $C_1 > 1$ be a sufficiently large absolute constant.
Fix $k \geq 2$, and let $C(k) \coloneqq 2^k C^{(0)}(k) + 1$, where $C^{(0)}(k)$ is a constant that is sufficiently large so that $\mathbb{E}\left[\left|\mathcal{N}(\mu, \sigma^2)\right|^k\right] \le C^{(0)}(k) \left( |\mu| + \sigma \right)^k$ is satisfied.\footnote{It is a standard fact that such a $C^{(0)}(k)$ exists.}
Suppose $\alpha, \eps \le 1$, and $\delta \le \min\left\{\widetilde{\Omega}\left(\frac{\alpha \eps}{d}\right)^{C_1}, \Omega\left(\min\left\{\frac{1}{n m}, \frac{\sqrt{d}}{n m \sqrt{\log\left(\frac{n m}{\sqrt{d}}\right)}}\right\}\right)\right\}$.
Suppose that $M$ is an $(\eps, \delta)$-DP mechanism that takes as input $X \coloneqq \left(X^{(1)}, \dots, X^{(n)}\right)$ with $X^{(i)} \coloneqq \left(X^{(i)}_1, \dots, X^{(i)}_m\right) \in \mathbb{R}^{m \times d}, \forall i \in [n]$.
Let $\mathcal{D}$ be any distribution over $\mathbb{R}^d$ such that $\mathbb{E}\left[\left|\left\langle \mathcal{D} - \mathbb{E}\left[\mathcal{D}\right], v \right\rangle \right|^k \right] \le C(k), v \in \mathbb{S}^{d - 1}$, and assume that $X^{(i)}_j \sim_{\text{i.i.d.}} \mathcal{D}, \forall i \in [n], j \in [m]$.
If, for any such $\mathcal{D}$, we have that $\mathbb{P}\left[\left\|M(X) - \mathbb{E}\left[\mathcal{D}\right] \right\|_2 \le \alpha \right] \geq \frac{2}{3}$, it must hold that $n \geq \widetilde{\Omega}\left( \frac{d}{m \alpha^2} + \frac{d}{\sqrt{m} \alpha \eps} + \frac{d}{m \alpha^{k/{k - 1}} \eps} \right)$.
\end{theorem}

The proof of Theorem~\ref{thm:main_lb_approx_dp} can be split into two parts.
The first part involves establishing the first and the last term of the sample complexity, whereas the second part focuses on the middle term.
The former is significantly simpler since, as we will see, it comes as a direct consequence of a reduction from the item-level setting.
In particular, we recall the following result from~\cite{Narayanan23}:

\begin{proposition}
\label{prop:shyam_lb}
[Theorem~$5.1$ from~\cite{Narayanan23}].
Let $C_1 > 1$ be a sufficiently large absolute constant.
Fix $k \geq 2$, and let $C(k) \coloneqq 2^k C^{(0)}(k) + 1$, where $C^{(0)}(k)$ is a constant that is sufficiently large so that $\mathbb{E}\left[\left|\mathcal{N}(\mu, \sigma^2)\right|^k\right] \le C^{(0)}(k) \left( |\mu| + \sigma \right)^k$ is satisfied.
Suppose $\alpha, \eps \le 1$, and $\delta \le \widetilde{\Omega}\left(\frac{\alpha \eps}{d}\right)^{C_1}$.
Suppose that $M$ is an $(\eps, \delta)$-DP mechanism that takes as input $X \coloneqq \left(X_1, \dots, X_n\right)$ with $X_i \in \mathbb{R}^d, \forall i \in [n]$.
Let $\mathcal{D}$ be any distribution over $\mathbb{R}^d$ such that $\mathbb{E}\left[\left|\left\langle \mathcal{D} - \mathbb{E}\left[\mathcal{D}\right], v \right\rangle \right|^k \right] \le C(k), \forall v \in \mathbb{S}^{d - 1}$, and assume that $X_i \sim_{\text{i.i.d.}} \mathcal{D}, \forall i \in [n]$.
If, for any such $\mathcal{D}$, we have that $\mathbb{P}\left[\left\|M(X) - \mathbb{E}\left[\mathcal{D}\right] \right\|_2 \le \alpha \right] \geq \frac{2}{3}$, it must hold that $n \geq \widetilde{\Omega}\left( \frac{d}{\alpha^2} + \frac{d}{\alpha^{\frac{k}{k - 1}} \eps} \right)$.
\end{proposition}

To understand why \Cref{prop:shyam_lb} implies a lower bound for the batch setting, we need the following lemma.
The lemma describes how, given oracle access to a mechanism for the person-level setting, we can construct a mechanism for the item-level setting.
This allows us to reduce instances of the item-level problem to the person-level.
Hence, lower bounds on the number of items in the item-level setting imply lower bounds on the number of people in the person-level setting.

\begin{lemma}
\label{lem:reduction_simple}
Let $\mathcal{D}$ be any distribution over $\mathbb{R}^d$.
Let $\eps, \delta, \alpha \geq 0$.
Assume that any $(\eps, \delta)$-DP mechanism $M \colon \bR^{n' \times d} \to \mathbb{R}^d$ which, given $X \coloneqq \left(X_1, \dots, X_{n'}\right) \sim \mathcal{D}^{\otimes n'}$, outputs $M(X)$ such that $\mathbb{P}\left[\left\|M(X) - \mathbb{E}\left[\mathcal{D}\right] \right\|_2 \le \alpha \right] \geq \frac{2}{3}$ requires at least $n' \geq n_0$ samples.
Then, for any $(\eps, \delta)$-DP mechanism $\Bar{M} \colon \left(\mathbb{R^{m \times d}}\right)^n \to \mathbb{R}^d$ which, given $Y \coloneqq \left(Y^{(1)}, \dots, Y^{(n)}\right) \sim \left(\mathcal{D}^{\otimes m}\right)^{\otimes n}$, outputs $\Bar{M}(Y)$ such that $\mathbb{P}\left[\left\|M(Y) - \mathbb{E}[\mathcal{D}]\right\|_2 \le \alpha\right] \geq \frac{2}{3}$, it must hold that $n \geq \frac{n_0}{m}$.
\end{lemma}

\begin{proof}
Let $\Bar{M}$ be an $(\eps, \delta)$-DP mechanism for the person-level setting that takes $n$ batches of size $m$ as input, and satisfies the desired accuracy guarantee.
Given oracle access to this mechanism, we will show how to construct a mechanism $M$ for the item-level setting which shares the same accuracy guarantees.
Assume that we have a dataset $X$ of size $n' \coloneqq n m$ that has been drawn i.i.d.\ from $\mathcal{D}$.
We partition the dataset into $n$ batches of size $m$.
This results in a dataset $Y$ of size $n$, where each datapoint is an i.i.d.\ sample from $\mathcal{D}^{\otimes m}$.
As a result, we can feed $Y$ into $\Bar{M}$, and the resulting mechanism will be $M$.
$M$ inherits the accuracy and privacy guarantees of $\Bar{M}$.
Consequently, any lower bound on $n_0$ on $n'$ also implies a lower bound on $n$, yielding the desired result.
\end{proof}

As a direct consequence of Proposition~\ref{prop:shyam_lb} and Lemma~\ref{lem:reduction_simple}, we get the following corollary which accounts for two out of three terms that appear in the lower bound of Theorem~\ref{thm:main_lb_approx_dp}.

\begin{corollary}
\label{cor:partial_lb_approx_dp}
Let $C_1 > 1$ be a sufficiently large absolute constant.
Fix $k \geq 2$, and let $C(k) \coloneqq 2^k C^{(0)}(k) + 1$, where $C^{(0)}(k)$ is a constant that is sufficiently large so that $\mathbb{E}\left[\left|\mathcal{N}(\mu, \sigma^2)\right|^k\right] \le C^{(0)}(k) \left( |\mu| + \sigma \right)^k$ is satisfied.
Suppose $\alpha, \eps \le 1$, and $\delta \le \widetilde{\Omega}\left(\frac{\alpha \eps}{d}\right)^{C_1}$.
Suppose that $M$ is an $(\eps, \delta)$-DP mechanism that takes as input $X \coloneqq \left(X^{(1)}, \dots, X^{(n)}\right)$ with $X^{(i)} \coloneqq \left(X^{(i)}_1, \dots, X^{(i)}_m\right) \in \mathbb{R}^{m \times d}, \forall i \in [n]$.
Let $\mathcal{D}$ be any distribution over $\mathbb{R}^d$ such that $\mathbb{E}\left[\left|\left\langle \mathcal{D} - \mathbb{E}\left[\mathcal{D}\right], \forall v \right\rangle \right|^k \right] \le C(k), v \in \mathbb{S}^{d - 1}$, and assume that $X^{(i)}_j \sim_{\text{i.i.d.}} \mathcal{D}, \forall i \in [n], j \in [m]$.
If, for any such $\mathcal{D}$, we have that $\mathbb{P}\left[\left\|M(X) - \mathbb{E}\left[\mathcal{D}\right] \right\|_2 \le \alpha \right] \geq \frac{2}{3}$, it must hold that $n \geq \widetilde{\Omega}\left( \frac{d}{m \alpha^2} + \frac{d}{m \alpha^{\frac{k}{k - 1}} \eps} \right)$.
\end{corollary}

It remains to argue about the term $\frac{d}{\sqrt{m} \alpha \eps}$ that appears in Theorem~\ref{thm:main_lb_approx_dp}.
To do so, we need to invoke a result that is implicit in~\cite{LevySAKKMS21}.
Theorem~$6$ in that work is concerned with proving a lower bound on the rates of Stochastic Convex Optimization with person-level privacy.
Establishing the result involves reducing from Gaussian mean estimation under person-level privacy, for which~\cite{LevySAKKMS21} proves a lower bound in Appendix E.2.
We explicitly state the result here.

\begin{proposition}
\label{prop:lb_gaussians_ulp}
Given $X \sim \left(\cN(\mu, \mathbb{I})^{\otimes m}\right)^{\otimes n}$ with $\mu \in [\pm 1]^d$, for any $\alpha = \cO\paren{\sqrt{d}}$ and any $\paren{\eps, \delta}$-DP mechanism $M \colon \left(\mathbb{R}^{m \times d}\right)^n \to [\pm 1]^d$ with $\eps, \delta \in [0, 1]$, and $\delta \le \Omega\left(\min\left\{\frac{1}{n m}, \frac{\sqrt{d}}{n m \sqrt{\log\left(\frac{n m}{\sqrt{d}}\right)}}\right\}\right)$ that satisfies $\underset{X, M}{\mathbb{E}}\left[\left\|M\paren{X} - \mu\right\|_2^2\right] \le \alpha^2$, it holds that $n = \Omega\left(\frac{d}{\sqrt{m} \alpha \eps}\right)$.
\end{proposition}

We note that the result in~\cite{LevySAKKMS21} comes from a reduction-based approach which works by leveraging the lower bound for Gaussian mean estimation in the item-level setting that was shown in~\cite{KamathLSU19}.
We provide a direct proof of Proposition~\ref{prop:lb_gaussians_ulp} in Appendix~\ref{sec:direct_proof_fing}, which might be of independent interest.

Having said all the above, Theorem~\ref{thm:main_lb_approx_dp} now follows directly by combining Corollary~\ref{cor:partial_lb_approx_dp} and Proposition~\ref{prop:lb_gaussians_ulp}.

\addcontentsline{toc}{section}{References}
\bibliographystyle{alpha}
\bibliography{biblio,gbiblio}
\appendix

\newpage
\section{Missing Proofs from Section~\ref{sec:hd}}\label{sec:extra-proofs}
\subsection{Missing Proofs from Section~\ref{sec:highdim-overview}}\label{sec:highdim-overview-missing}
\begin{lemma}
Let $m, j$ be integers satisfying $1 \le j \le m$. It follows that $$\binom{m}{j} \le \left(\frac{em}{j}\right)^j.$$
\label{lem:binom-upperbound}
\end{lemma}

\begin{lemma}
	Let $a > 0$ and $b \ge 0$ be constants. The function $(ax^b)^x$ defined on $x > 0$ is convex.
	\label{lem:convex}
\end{lemma}
\begin{proof}
Note that
\[
\frac{d^2}{dx^2} (ax^b)^x = (ax^b)^x\Paren{\Paren{\ln (ax^b) + b}^2 + \frac{b}{x}} \ge 0
\]
when $a,x > 0$, $b \ge 0$, completing the proof.
\end{proof}

For the next lemma, recall that the function $f(\ell)$ is defined as $$f(\ell) = \Theta\Paren{\frac{d^{k/2} \cdot 2^{\ell(k - 1)}\log^km}{m^{k-1}t^k}}^{2^{\ell-1}}.$$
\lemfisconvex*
\begin{proof}
    Set $x = 2^{\ell}$, $a = \frac{m^{1/2}d^{k/4}\log^{k/2} m}{m^{k/2}t^{k/2}}$, and $b = \frac{k-1}{2}$. Observe that $a, b, x > 0$. Thus,
    \begin{align*}
    	f(\ell)
    	&= \Paren{\frac{md^{k/2}2^{\ell(k-1)}\log^k m}{m^kt^k}}^{2^{\ell-1}}\\
    	&= \Paren{\frac{m^{1/2}d^{k/4}\log^{k/2} m}{m^{k/2}t^{k/2}}\cdot 2^{\ell\Paren{\frac{k-1}{2}}}}^{2^{\ell}}\\
    	&= (ax^b)^x,
    \end{align*}
    which we know by Lemma~\ref{lem:convex} is convex with respect to $x > 0$, and thus is convex with respect to $\ell$.
\end{proof}

\subsection{Missing Proofs from Section~\ref{sec:highdim-smallt}}\label{sec:highdim-smallt-missing}
\begin{restatable}{lemma}{lemmavectorbernstein}
    Let $k > 2$ and $t > 0$. Let $X_i$ be a random vector in $\mathbb{R}^d$ with $\E[X_i] = 0$ and $k$-th moment bounded by $1$. Let $r_1 = \frac{d}{t}$. Let $Y_i = X_i\cdot \mathbb{I}\Brac{\normt{X_i} < r_1}$. There exists $t_2 = \Omega(\sqrt{d})$ such that, for all positive $t \le t_2$,
    $$\normt{\E\Brac{Y_i}} \le \min\Paren{\frac{t}{6},\frac{d}{t}}.$$
    \label{lem:bounding-expectation-in-the-tails}
\end{restatable}
\begin{proof}[Proof of Lemma~\ref{lem:bounding-expectation-in-the-tails}]
    Since $\E[X_i] = 0$, we have that
    \[
        \E[X_i] = \E\Brac{X_i\cdot \mathbb{I}\Brac{\normt{X_i} < r_1}} + \E\Brac{X_i\cdot \mathbb{I}\Brac{\normt{X_i} \ge r_1}}
    \]
    which implies that
    \begin{align*}
        \normt{\E\Brac{Y_i}} &= \normt{\E\Brac{X_i\cdot \mathbb{I}\Brac{\normt{X_i} < r_1}}}\\
        &= \normt{\E[X_i] - \E\Brac{X_i\cdot \mathbb{I}\Brac{\normt{X_i} \ge r_1}}}\\
        &= \normt{\E\Brac{X_i\cdot \mathbb{I}\Brac{\normt{X_i} \ge r_1}}}
    \end{align*}
    So now,
    \begin{align*}
        &\normt{\E\Brac{X_i\cdot \mathbb{I}\Brac{\normt{X_i} \ge r_1}}}\\
        &\le \E\Brac{\normt{X_i}\cdot \mathbb{I}\Brac{\normt{X_i} \ge r_1}}\\
        &=\int_{x=0}^{\infty} \Pr\Brac{ \normt{X_i}\cdot \mathbb{I}\Brac{\normt{X_i} \ge r_1} > x}\,\mathrm{d}x\\
        &=\int_{x=0}^{r_1} \Pr\Brac{ \normt{X_i}\cdot \mathbb{I}\Brac{\normt{X_i} \ge r_1} > x}\,\mathrm{d}x \\
        &\phantom{oooooo}+ \int_{x=r_1}^{\infty} \Pr\Brac{ \normt{X_i}\cdot \mathbb{I}\Brac{\normt{X_i} \ge r_1} > x}\,\mathrm{d}x\\
        &= \int_{x=0}^{r_1} \Pr\Brac{\normt{X_i} \ge r_1}\,\mathrm{d}x + \int_{x=r_1}^{\infty} \Pr\Brac{ \normt{X_i} > x}\,\mathrm{d}x\\
        &= \int_{x=0}^{r_1} \Paren{\frac{d^{1/2}}{r_1}}^k\,\mathrm{d}x + \int_{x=r_1}^{\infty}\Paren{\frac{d^{1/2}}{x}}^k\,\mathrm{d}x \tag{by Lemma~\ref{lem:tail-bound-norm}}\\
        &= \frac{d^{k/2}}{r_1^{k-1}} + \frac{d^{k/2}}{r_1^{k-1}(k-1)}\\
        &\le \frac{2d^{k/2}}{r_1^{k-1}} \tag{since $k \ge 2$}\\
        &= \frac{2t^{k-1}}{d^{k/2-1}}.\\
    \end{align*}
   To finish the proof, note that
    \begin{equation*}
        \frac{2t^{k-1}}{d^{k/2-1}}
        \le 2t\Paren{\frac{t}{\sqrt{d}}}^{k-2}
        \le \frac{t}{6},
    \end{equation*}
    which uses the assumption that $t \le t_2$ for some $t_2 \le 12^{\frac{-1}{k-2}}\sqrt{d}$.
    Also note that
    \begin{equation*}
        \frac{2t^{k-1}}{d^{k/2-1}}
        \le \frac{2t^{k}}{d^{k/2}}\cdot\frac{d}{t}
        \le \frac{d}{t}
    \end{equation*}
    which uses the assumption that $t \le t_2$ for some $t_2 \le 2^{-1/k}\sqrt{d}$.
\end{proof}

\begin{lemma}
    Let $k > 2$ and $t > 0$. Let $X_i$ be a random vector in $\mathbb{R}^d$ with $\E[X_i] = 0$ and $k$-th moment bounded by $1$. Let $r_1 = \frac{d}{t}$. Let $Y_i = X_i\cdot \mathbb{I}\Brac{\normt{X_i} < r_1}$. Let $Z_i = Y_i - \E\Brac{Y_i}$. There exists $t_2 = \Omega(\sqrt{d})$ such that, for all positive $t \le t_2$,
    \begin{align}
    \E[Z_i] &= 0,\label{eq:bernstein-constraint-1}\\
    \E\Brac{\snormt{Z_i}} &\le d\text{, and}\label{eq:bernstein-constraint-2}\\
    \normt{Z_i} &\le 2r_1. \label{eq:bernstein-constraint-3}
\end{align}
\label{lem:bernstein-buildup-1}
\end{lemma}
\begin{proof}
    Note that \eqref{eq:bernstein-constraint-1} is satisfied by design of the $Z_i$'s. We now show that \eqref{eq:bernstein-constraint-2} holds. Note that
    \begin{equation*}
    \E\Brac{\snormt{Z_i}}
    =
    \E\Brac{\snormt{Y_i - \E\Brac{Y_i}}}
    \le
    \E\Brac{\snormt{Y_i}}
    \le 
    \E\Brac{\snormt{X_i}},
    \end{equation*}
    where the first inequality above comes from the fact that setting $y = \E\Brac{Y_i}$ minimizes $\E\Brac{\snormt{Y_i - y}}$.
    Observe that
    \begin{equation}
    \label{eq:bound-snormt-Xi}
    \E\Brac{\normt{X_i}^2} 
    =
    \E\Brac{\sum_{j=1}^d \iprod{X_i, e_j}^2}
    =
    \sum_{j=1}^d \E\Brac{\iprod{X_i, e_j}^2}
    \le
    \sum_{j=1}^d \E\Brac{\iprod{X_i, e_j}^k} 
    \le 
    d,
    \end{equation}
    where the last inequality comes from the assumption that, 
    for every $u$ with $\normt{u} = 1$, $\E\Brac{\iprod{X_i, u}^k} \le 1$. 
    Thus, $\E\Brac{\snormt{Z_i}} \le d$ and so \eqref{eq:bernstein-constraint-2} is satisfied.
    
    Now we show that the last constraint, \eqref{eq:bernstein-constraint-3}, also holds. We start by expanding $\normt{Z_i}$ for some $i$.
    \begin{align}
    \normt{Z_i}
    &=
    \normt{X_i\cdot \mathbb{I}\Brac{\normt{X_i} < r_1} - \E\Brac{Y_i}} \notag\\
    &\le
    \normt{X_i\cdot \mathbb{I}\Brac{\normt{X_i} < r_1}} + \normt{\E\Brac{Y_i}} \notag\\
    &\le
    r_1 + \normt{\E\Brac{Y_i}} \notag\\
    &\le 2r_1.\tag{by Lemma~\ref{lem:bounding-expectation-in-the-tails}}
    \end{align}
\end{proof}
\begin{lemma}
    Let $k > 2$ and $t > 0$. Let $X_i$ be a random vector in $\mathbb{R}^d$ with $\E[X_i] = 0$ and $k$-th moment bounded by $1$. Let $r_1 = \frac{d}{t}$. Let $Y_i = X_i\cdot \mathbb{I}\Brac{\normt{X_i} < r_1}$. Let $Z_i = Y_i - \E\Brac{Y_i}$. There exists $t_2 = \Omega(\sqrt{d})$ such that, for all positive $t \le t_2$,
    $$\Pr \Brac{\Bignormt{\sum_{i= 1}^m Z_i} \ge m \Paren{\frac{t}{3} - \normt{\E\Brac{Y_i}}}} \le \exp\Paren{-\Theta\Paren{\frac{mt^2}{d}}}$$
    \label{lem:final-berstein}
\end{lemma}
\begin{proof}
    We apply the Vectorized Bernstein's Inequality (i.e. Lemma~\ref{lem:vector-bernstein}) by setting $t' = t / 3 - \normt{\E\Brac{Y_i}}$, $r = 2r_1$, and $\sigma^2 = d$. One can check using Lemma~\ref{lem:bounding-expectation-in-the-tails} that $0 < t' < \sigma^2/r$ is satisfied for all $0 < t < \sqrt{d}/12$. Thus,
    \begin{align*}
    \Pr \Brac{\Bignormt{\sum_{i= 1}^m Z_i} \ge m \paren{t / 3 - \normt{\E\Brac{Y_i}}}}
    &= \Pr \Brac{\Bignormt{\sum_{i= 1}^m Z_i} \ge mt'}\\
    &\le \exp\Paren{- m \frac{(t')^2}{8 \sigma^2} + \frac{1}{4}} \tag{by Lemma~\ref{lem:vector-bernstein}}\\
    &= \exp\Paren{- m \frac{(t / 3 - \normt{\E\Brac{Y_i}})^2}{8d} + \frac{1}{4}}\\
    &\le \exp\Paren{- m \frac{(t / 3 - t/6)^2}{8d} + \frac{1}{4}} \tag{by Lemma~\ref{lem:bounding-expectation-in-the-tails}}\\
    &\le \exp\Paren{-\Theta\Paren{\frac{mt^2}{d}}}.
    \end{align*}
\end{proof}

For the next lemma, recall that the function $f(\ell)$ is defined as $$f(\ell) = \Theta\Paren{\frac{d^{k/2} \cdot 2^{\ell(k - 1)}\log^km}{m^{k-1}t^k}}^{2^{\ell-1}}.$$
\begin{restatable}{lemma}{lemellislarge}
Let $m, d, t > 0$ and $k > 2$. Let $r_2 = mt/(3\log m)$ and let $r_1 = d/t$. There exists $t_1 = O\Paren{\sqrt{\frac{d\log m}{m}}}$ and $t_2 = \Omega\Paren{\sqrt{d}\log^{\frac{-1}{k-2}}m}$ such that, for all $t \in [t_1,t_2]$,
    $$f(\log(r_2/r_1)) = \tilde{O}\Paren{\frac{d^{k/2}}{m^{k-1}t^k}}.$$
    \label{lem:ell-is-large}
\end{restatable}
\begin{proof}
    Define $g(t) = f\Paren{\log \Paren{\frac{mt^2}{3d\log m}}}$. Note that $$\min_{t \in [t_1,t_2]} \tilde{O}\Paren{\frac{d^{k/2}}{m^{k-1}t^k}} = \tilde{O}\Paren{\frac{1}{m^{k-1}}}.$$ Thus, it suffices to show that
    \begin{equation}
        \max_{t \in [t_1,t_2]} g(t) \le \tilde{O}\Paren{\frac{1}{m^{k-1}}}.
    \end{equation}
    We break the proof into two parts. First, we show that 
    \begin{equation}
        \max_{t \in [t_1,t_2]} g(t) \in \{g(t_1),g(t_2)\}.
        \label{eq:g-is-max}
    \end{equation}
    Second, we show that
    \begin{equation}
        g(t_1),g(t_2) \le \tilde{O}\Paren{\frac{1}{m^{k-1}}}.
        \label{eq:g-is-bounded}
    \end{equation}
    To show \eqref{eq:g-is-max}, it suffices to show that $g(t)$ is convex with respect to $t$ for $t > 0$. Note that
    \begin{align*}
    	g(t)
    	&= \Theta\Paren{\frac{md^{k/2}\Paren{\frac{mt^2}{3d\log m}}^{k-1}\log^km}{m^kt^k}}^{\frac{mt^2}{6d\log m}}\\
    	&= \Theta\Paren{\frac{d^{-k/2 + 1}t^{k-2}\log m}{3^{k-1}}}^{\frac{mt^2}{6d\log m}}\\
    	&= \Theta\Paren{\log m\Paren{\frac{t^2}{d}}^{\frac{k-2}{2}}}^{\frac{mt^2}{6d\log m}}.
    \end{align*}
    Note that
    \[
        \Theta\Paren{\log m\Paren{\frac{t^2}{d}}^{\frac{k-2}{2}}}^{\frac{mt^2}{6d\log m}} = \Paren{\Paren{\frac{\log m}{d^{\frac{k-2}{2}}}}^{\frac{m}{6d\log m}}\Paren{t^2}^{\frac{k-2}{2}\cdot \frac{m}{6d\log m}}}^{t^2} = (ax^b)^x,
    \]
    for $x = t^2$, $a = \Paren{\frac{\log m}{d^{\frac{k-2}{2}}}}^{\frac{m}{6d\log m}}$ and $b = \frac{k-2}{2}\cdot \frac{m}{6d\log m}$. Since $a,b,x > 0$, then by Lemma~\ref{lem:convex}, $(ax^b)^x$ is convex with respect to $x$, and so $g$ is also convex for $t > 0$. Thus, \eqref{eq:g-is-max} holds.

    We move on to showing \eqref{eq:g-is-bounded}. Let $t_1 = \sqrt{\frac{cd\log m}{m}}$ for some $c \ge \frac{12(k-1)}{k-2}$. When $t = t_1$,
    \begin{align*}
        g(t_1) 
        &= \Theta\Paren{\log m\Paren{\frac{t_1^2}{d}}^{\frac{k-2}{2}}}^{\frac{mt_1^2}{6d\log m}}\\
        &= \Theta\Paren{\log m\Paren{\frac{\log m}{m}}^{\frac{k-2}{2}}}^{\frac{c}{6}}\\
        &= \tilde{O}\Paren{\frac{1}{m^{\Paren{\frac{k-2}{2}\cdot\frac{c}{6}}}}}\\
        &= \tilde{O}\Paren{\frac{1}{m^{k-1}}}.
    \end{align*}
    Let $t_2 = \frac{\sqrt{d}\log^{\frac{-1}{k-2}}m}{2}$. When $t = t_2$,
    \begin{align*}
        g(t_2) &= \Theta\Paren{\log m\Paren{\frac{t_2^2}{d}}^{\frac{k-2}{2}}}^{\frac{mt_2^2}{6d\log m}}\\
        &= O\Paren{\frac{1}{2}}^{O(m)}\\
        &= \tilde{O}\Paren{\frac{1}{m^{k-1}}}.
    \end{align*}
\end{proof}

\subsection{Missing Proofs from Section~\ref{sec:bias}}\label{sec:bias-missing}
\begin{lemma}
    Let $\ell > 0$. Let $\normt{u - \mu} < \rho$.  It follows that $$\normt{X-\mu} \ge \normt{X-\clip{X}} + \rho - \normt{\mu - u}.$$
    \label{lem:bias-rearranging-terms}
\end{lemma}
\begin{proof}
    Let $r = \normt{u - \mu}$, and let $\mathbb{B}^d(u,r)$ be the $\ell_2$ ball in $d$ dimensions of radius $r$ that is centered at $u$. Note that, by assumption, $r < \rho$, and so it follows that $\mu \in \mathbb{B}^d(u,r) \subseteq \mathbb{B}^d(u,\rho)$. Let $$\mu^* = \argmin_{\mu' \in \mathbb{B}^d(u,r)} \normt{X-\mu'}.$$ 
    One can check that $\mu^*$ is a point on the surface of $\mathbb{B}^d(u,r)$, and also that $u$, $\mu^*$, $\clip{X}$, and $X$ are collinear and arranged sequentially on the line in that order. We use these facts to relate the quantities $\normt{X-\mu}$ and $\normt{X-\clip{X}}$:
    \begin{align}
        \normt{X-\mu} 
        &\ge \normt{X-\mu^*}\notag\\
        &= \normt{X-u} - \normt{\mu^* - u}\notag\\
        &= \normt{X-\clip{X}} + \normt{\clip{X}-u} - \normt{\mu^* - u}\notag\\
        &\ge \normt{X-\clip{X}} + \rho - \normt{\mu - u} \label{eq:foobar}
    \end{align}
    Note that $\normt{X-\clip{X}} \ge \ell$ thus implies $\normt{X-\mu} \ge \ell + \rho - \normt{\mu - u}$, completing the proof.
\end{proof}
\subsection{Missing Proofs from Section~\ref{sec:full-algo}}\label{sec:full-algo-missing}
To prove \Cref{lem:noise-and-sampling}, we use the following lemma about Gaussian random variables:
\begin{lemma}[Laurent and Massart~\cite{LaurentM00}]\label{lem:gaussians}
    Let $d \in \N$. Consider the random vector $W \sim \cN(0,\sigma^2\mathbb{I}_d)$. Then, for every $s \ge 0$,
    \[
        \Pr\Brac{\normt{W}^2 \ge \sigma^2\Paren{d + 2d\sqrt{s}+2ds}} \le \exp{(-sd)}.
    \]
\end{lemma}
We also use \Cref{lem:sampling-error} which gives an analysis of the error incurred due to the sampling process.
\begin{restatable}{lemma}{lemsamplingerror}\label{lem:sampling-error}
    Suppose $\cD$ over $\R^d$ satisfies $\sigma_k(\cD) \le 1$. For all $i \in [n]$ and $j \in [m]$, let $X^{(i)}_j \stackrel{\text{iid}}{\sim} \cD$, and for some $\rho \ge 0$, $u \in \R^d$ let $$Z_i = \clip{\frac{1}{m}\sum_{j=1}^mX^{(i)}_j}.$$ Let $\alpha, \beta > 0$. There exists $n_0 = O\Paren{\frac{d}{m\beta\alpha^2}}$ such that for all $n \ge n_0$ and for all $\rho$ and $u$,
    \begin{equation}
        \Pr\Brac{\Normt{\frac{1}{n}\sum_{i=1}^n Z_i - \E\Brac{\frac{1}{n}\sum_{i=1}^n Z_i}} \ge \alpha} \le \beta.
    \end{equation}
\end{restatable}
\begin{proof}
    Let $Y^{(i)} = \frac{1}{m}\sum_{j=1}^mX^{(i)}_j \in \R^d$, and let $Y^{(i)}(\ell) \in \R$ be the $\ell$th coordinate of $Y^{(i)}$. We begin by expanding the following expression of expectation:
    \begin{align*}
        &\E\Brac{\Normt{\frac{1}{n}\sum_{i=1}^n Z_i - \E\Brac{\frac{1}{n}\sum_{i=1}^n Z_i}}^2}\\
        &= \frac{1}{n^2}\E\Brac{\Normt{\sum_{i=1}^n(Z_i-\E\Brac{Z_i})}^2}\\
        &\le \frac{1}{n^2}\E\Brac{\sum_{i=1}^n\Normt{Z_i-\E\Brac{Z_i}}^2}\\
        &= \frac{1}{n}\E\Brac{\Normt{Z_1-\E\Brac{Z_1}}^2} \\
        &\le \frac{1}{n}\E\Brac{\Normt{\frac{1}{m}\sum_{j=1}^mX_j^{(1)}-\E\Brac{\frac{1}{m}\sum_{j=1}^mX_j^{(1)}}}^2}\\
        &= \frac{1}{n}\E\Brac{\sum_{\ell=1}^d\Paren{Y^{(1)}(\ell)-\E\Brac{Y^{(1)}(\ell)}}^2}\\
        &= \frac{1}{n}\sum_{\ell=1}^d\E\Brac{\Paren{Y^{(1)}(\ell)-\E\Brac{Y^{(1)}(\ell)}}^2}\\
        &\le \frac{d}{nm}.
    \end{align*}
    Where the last line follows from the assumption that $\sigma_k(\cD) \le 1$. By Chebyshev's Inequality,
    \begin{equation*}
        \Pr\Brac{\Normt{\frac{1}{n}\sum_{i=1}^n Z_i - \E\Brac{\frac{1}{n}\sum_{i=1}^n Z_i}} \ge \alpha} = O\Paren{\frac{d}{mn\alpha^2}} \le \beta.
    \end{equation*}
\end{proof}
\lemnoiseandsampling*
\begin{proof}
    Recall from \textrm{ClipAndNoise} that $W_t \sim \cN(0,\sigma^2\mathbb{I}_d)$ with $\sigma = \tilde{O}\Paren{\rho_t/(n\eps')}$. Using \Cref{lem:gaussians}, let $s = \frac{\log(\log\log d / 0.05)}{d}$. With probability at least $1-0.05/(\log\log d)$,
    \[
        \textrm{noise}(t) = \normt{W_t}
        =\tilde{O}\Paren{\frac{\sqrt{d}\rho_t}{n\eps'}}.
    \]
    Thus, there exists $n_0 = \tilde{O}\Paren{\frac{d
    \sqrt{\log{(1/\delta)}}}{\eps}}$, such that, for $n \ge n_0$,
    it follows that with probability at least $1-0.05/(\log\log d)$, 
    $$\normt{W} \le \frac{\rho_t}{4\sqrt{d}}.$$
    In order to bound $\textrm{sampling}(t)$, we apply \Cref{lem:sampling-error} directly, setting $\alpha = \frac{\rho_t}{2\sqrt{d}}$ and $\beta = 0.05/\log\log d$. In particular, if $$n \ge \frac{d\log\log d}{0.05} \ge \frac{d\log\log d}{0.05m}\cdot \frac{4d}{\rho_t^2} \ge \frac{d}{m\beta\alpha},$$
    then $\textrm{sampling}(t) > \frac{\rho_t}{2\sqrt{d}}$ with probability at most $\frac{0.05}{\log\log d}$.
\end{proof}

\subsection{Coarse Estimation}\label{subsec:coarse-estimation}
In this subsection, we provide an algorithm to obtain a coarse estimate of the mean, up to accuracy $O\paren{\sqrt{d / m}}$ in high dimensions under approximate differential privacy, using $1$-dimensional estimates (\Cref{thm:user-level-coarse-esitmation}).

\thmhighdcoarseestimate*
\begin{proof}
Run $d$ instances of \Cref{alg:user-level-coarse-estimate} with $r'= r/\sqrt{d}, \eps' = \eps / \sqrt{d \log\paren{1/\delta}}, \delta' = \delta / d$, on each coordinate of the dataset $X$, and output the coordinate-wise estimate of the mean. By \Cref{thm:user-level-coarse-esitmation} (one dimensional coarse estimation), and \Cref{lem:advanced-composition} (advanced composition), we know that this output will be person-level $\paren{O\paren{\eps}, O\paren{\delta}}$-DP. Moreover, the output will have accuracy $r$, with probability $1-\beta$, as long as $n \ge n_1$, for some 
\begin{equation*}
n_1 = O \Paren{\frac{\log\paren{d/\beta}}{\log\paren{r \sqrt{m / d}}}
+
\frac{\sqrt{d}\sqrt{\log\paren{1/\delta}}\log\paren{d^2/\delta \beta}}{\eps}
} \mper
\end{equation*}
Alternatively, using basic composition (\Cref{lem:privacy-composition}), and applying \Cref{alg:user-level-coarse-estimate} with $r ' = r / \sqrt{d}, \eps' = \eps / d, \delta' = \delta / d$, on each coordinate of $X$, and outputting the coordinate-wise estimate of the mean, we obtain a person-level $\paren{\eps, \delta}$-DP algorithm that has accuracy $r$, with success probability $1 - \beta$, as long as $n \ge n_2$, for some
\begin{equation*}
n_2 = O \Paren{\frac{\log\paren{d/\beta}}{\log\paren{r \sqrt{m / d}}}
+
\frac{d\log\paren{d^2/\delta \beta}}{\eps}
} \mper
\end{equation*}
\end{proof}

\section{Lower Bounds Under Pure Differential Privacy}
\label{sec:pure-lbs}
Suppose $\cD$ is a distribution over $\R^d$ with $k$-th moment bounded from above by $1$. Assume $n$ people each take $m$ samples from the distribution $\cD$. In this section we show lower bounds against the number of people required to estimate the mean of the distribution $\cD$ up to accuracy $\alpha$ with success probability $1-\beta$, under person-level $\eps$-differential privacy. 

We can describe the problem in an equivalent language: take $n$ samples from the tensorized distribution $\cD^{\tensor m}$. How many samples are required in order to estimate the mean of $\cD$ up to accuracy $\alpha$ with success probability $1-\beta$, under item-level $\eps$ differential privacy?

Technically, we apply packing lower bounds to distributions of the form $\cD^{\tensor m}$. Specifically, we repeatedly apply the following theorem, that shows sample complexity lower bounds against learning a set of distributions that are all $\gamma$ far from a central distribution and have disjoint parameters that can be learned with probability at least $1-\beta$. Moreover, we use some facts about Kullback-Leibler divergence in order to analyze the tensorized distribution.

\begin{theorem}[Theorem~7.1 in \cite{HopkinsKM22}]
\label{thm:lb-meta}
Let $\mathcal{P}=\left\{P_1, \ldots, P_m\right\}$ be a set of distributions, and $P_O$ be a distribution such that for every $P_i \in \mathcal{P},\left\|P_i-P_O\right\|_{\mathrm{TV}} \leq \gamma$. Let $\mathcal{G}=\left\{G_1, \ldots, G_m\right\}$ be a collection of disjoint subsets of some set $\mathcal{Y}$. If there is an $\varepsilon$-DP algorithm $M$ such that $\mathbb{P}_{X \sim P_i^n}\left[M(X) \in G_i\right] \geq 1- \beta$ for all $i \in[m]$, then
$$
n \geq \Omega\left(\frac{\log m+\log (1 / \beta)}{\gamma\left(e^{2 \varepsilon}-1\right)}\right) .
$$
Note that for the usual regime $\varepsilon \leq 1$, we can replace the $e^{2 \varepsilon}-1$ in the denominator with $\varepsilon$.
\end{theorem}

We instantiate \Cref{thm:lb-meta} three times in order to prove each term of the lower bound. The first term is the non-private cost, the second term corresponds to the cost of estimating the mean of a Gaussian, the third term corresponds to the cost of estimating the mean of point distributions with bounded $k$-th moment, and the last term corresponds to the cost of estimating up to a coarse estimate.

\begin{theorem}[Pure DP Lower Bound for Bounded $k$-th Moments]
\label{thm:main-lb}
Suppose $k \ge 2$, and $\cD$ is a distribution with $k$-th moment bounded by $1$ and mean $\mu$. Moreover, suppose $\mu$ is in a ball of radius $R$, where $\alpha \le \min\set{\frac{6}{25} R, \paren{1/25}^{(k-1)/k}}$, and assume $\eps \le 1$. Any $\eps$-DP algorithm that takes $n$ samples from $\cD^{\tensor m}$ and outputs $\hat{\mu}$ such that $\normt{\mu - \hat{\mu}} \le \alpha$ with probability $1-\beta$ requires
\begin{equation*}
n = \Omega_k\Paren{
\frac{d + \log\paren{1/\beta}}{\alpha^2 m}
+
\frac{d + \log \paren{1/\beta}}{\alpha\sqrt{m} \eps}
+\frac{d+ \log\paren{1/\beta}}{\alpha^{\frac{k}{k-1}} m \eps}
+
\frac{d\log \paren{R / \alpha} + \log\paren{1/\beta}}{\eps}
}.
\end{equation*}
many samples, where $\Omega_k$ hides multiplicative factors that only depend on $k$.
\end{theorem}
\begin{proof}
The first term is the cost of estimating the mean of a distribution with $nm$ samples without privacy.
For the other terms putting together \Cref{cor:normalized-lb-gaussian}, \Cref{lem:lb-bounded-moments}, and \Cref{lem:lb-range} finishes the proof. 
\end{proof}

\subsection{Preliminaries}

First, we state some facts about Kullback–Leibler divergence and total variation distance. In our analysis we intend to analyze the total variation distance of the tensorized distribution $\cD^{\tensor m}$.  However, analyzing the total variation distance of the tensorized distribution directly, is difficult, as the total variation distance does not behave well under tensorization. Thereofore, instead we use the Kullback-Leibler divergence to analyze the total variation distance of the tensorized distribution through Pinkser's inequality.

We want to analyze the tensorized distribution $\cD^{\tensor m}$. The following fact states that the Kullback-Leibler divergence tensorizes well.
\begin{fact}[Tensorization of KL Divergence]
\label{fact:kl-tensorizes}
Assume $\set{P_i}_{i=1}^n$, and $\set{Q_i}_{i=1}^n$ are distributions over $\cX$ and $P = \tensor_{i=1}^n P_i$ and $Q = \tensor_{i=1}^n Q_i$. Then 
\begin{equation*}
\mathrm{KL}
(P \Vert Q)=\sum_{i=1}^n \mathrm{KL}
\left(P_i \| Q_i\right) \mper
\end{equation*}
\end{fact}
We apply Pinsker's inequality to relate Kullback-Leibler divergence and total variation distance.
\begin{fact}[Pinsker's Inequality]
\label{fact:pinsker}
Let $P$ and $Q$ be two distributions over $\cX$. Then
\begin{equation*}
\mathrm{d}_{\mathrm{TV}}(P, Q) \leq \sqrt{\frac{\mathrm{KL}(P \Vert Q)}{2}}
\mper
\end{equation*}
\end{fact}

Next, for two arbitrary Gaussians their closed form KL divergence is known, which we use in our analysis.

\begin{fact}[KL Divergence of Gaussians \cite{Rmusssen05}]
\label{fact:kl-gaussian}

For two arbitrary Gaussians 
$\mathcal{N}\left(\mu_1, \Sigma_1\right), \mathcal{N}\left(\mu_2, \Sigma_2\right)$
we have
\begin{equation*}
\mathrm{KL}\left(\mathcal{N}\left(\mu_1, \Sigma_1\right) \| \mathcal{N}\left(\mu_2, \Sigma_2\right)\right)=\frac{1}{2}\left(\operatorname{tr}\left(\Sigma_1^{-1} \Sigma_2-I\right)+\left(\mu_1-\mu_2\right)^{\top} \Sigma_1^{-1}\left(\mu_1-\mu_2\right)-\log \operatorname{det}\left(\Sigma_2 \Sigma_1^{-1}\right)\right) \mper
\end{equation*}
Specifically for identity covariance,
\begin{equation*}
\mathrm{KL}\left(\mathcal{N}\left(\mu_1, I\right) \| \mathcal{N}\left(\mu_2, I\right)\right)= \frac{\snormt{\mu_1-\mu_2}}{2} \mper
\end{equation*}
\end{fact}

Next we state a basic fact about packing number $M(\Theta,\|\cdot\|, \epsilon)$ and covering number $N(\Theta,\|\cdot\|, \epsilon)$. This fact shows and upper and lower bound for the packing and covering number of a unit norm ball.

\begin{fact}[Metric Entropy of Norm Balls \cite{Wu16}]
\label{fact:covering-number}
Let $B_2(1)$ be the unit $\ell_2$ ball.
Consider $N\left(B_2(1),\normt{\cdot}, \epsilon\right).$ When $\epsilon \ge 1, N\left(B_2(1),\normt{\cdot}, \epsilon\right)= 1$. When $\epsilon<1$, 
\begin{equation*}
\left(\frac{1}{\epsilon}\right)^d
\leq 
N\left(B_2(1),\normt{\cdot}, \epsilon\right)
\leq 
M\left(B_2(1),\normt{\cdot}, \epsilon\right)
\le
\left(1+\frac{2}{\epsilon}\right)^d \leq\left(\frac{3}{\epsilon}\right)^d
\mper
\end{equation*}
Hence $d \log \frac{1}{\epsilon} \leq 
\log N\left(B_2(1),\normt{\cdot},\epsilon\right) 
\leq
\log M\left(B_2(1),\normt{\cdot}, \epsilon\right)
\le
d \log \frac{3}{\epsilon}$.
\end{fact}

\subsection{Proof of Theorem~\ref{thm:main-lb}}
In this section we prove the lower bounds for each term in \cref{thm:main-lb}. First we prove the second term which corresponds to learning a Gaussian.
\begin{lemma}[Pure DP Gaussian Mean Estimation Lower Bound]
\label{lem:lb-gaussian}
Suppose $\mu$ is in a ball of radius $R$, where $\alpha \le R/4$, and suppose $\eps \le 1$.
Any  $\eps$-DP algorithm that takes $n$ samples from $\cN\paren{\mu , I_d}^{\tensor m}$ and outputs $\hat{\mu}$ such that $\normt{\mu - \hat{\mu}} \le \alpha$ with probability $1-\beta$ requires 
\begin{equation*}
n = \Omega\Paren{\frac{d + \log \paren{1/\beta}}{\alpha \sqrt{m} \eps}} \mcom
\end{equation*}
many samples.
\end{lemma}
\begin{proof}
We want to use \Cref{thm:lb-meta}. 
Let $P_O = \mathcal{N} \paren{0, I_d}^{\tensor m}$,  $P_i = \cN \paren{\lambda v_i, I_d}^{\tensor m}$, for $v_i$ in the unit ball, such that $\normt{v_i - v_j} \ge 2\alpha / \lambda$, and for $\lambda$ to be determined later. We know that a set $I$ such that $\log\abs{I} \ge d \log\paren{\lambda / 2 \alpha}$ of such $v_i's$ exists from \Cref{fact:covering-number}.
Let $G_i = \lambda v_i + \alpha B_d$, for such a set $I$ and $i \neq j$, we will have that $G_i$'s are disjoint.
Now, from \Cref{fact:kl-gaussian}, and $\normt{v_i}\le 1$ we know that $\mathrm{KL}\paren{\cN\paren{0, I_d} \| \cN\paren{\lambda v_i, I_d}} \le \lambda^2 / 2$. Tensorization of KL (\Cref{fact:kl-tensorizes}) and Pinsker's inequality (\Cref{fact:pinsker}) imply that
\begin{equation*}
\normtv{P_O - P_i} \le \sqrt{\frac{\mathrm{KL}\paren{P_i \| P_O}}{2}} \le \sqrt{\frac{m \cdot \mathrm{KL}\paren{\cN\paren{\lambda v_i, I_d}, \cN\paren{0, I_d}}}{2}} \le \sqrt{m} \lambda /2
\mper
\end{equation*}
Taking $\lambda = 4 \alpha$ we obtain 
\begin{equation*}
n = \Omega\Paren{\frac{d + \log \paren{1/\beta}}{\alpha \sqrt{m} \eps}} \mper
\end{equation*}
\end{proof}
However, note that the the $k$-th moment of $\cN\paren{\mu, I_d}$, is not bounded by $1$, but instead it is bounded from above by $(k-1)!!$, where the $!!$ denotes the double factorial that is the product of all numbers from $1$ to $k-1$ with the same parity as $k-1$. We state the corollary below to extend the lemma above to our setting where the $k$-th moment is bounded by $1$.
\begin{corollary}[Pure DP Normalized Gaussian Mean Estimation Lower Bound]
\label{cor:normalized-lb-gaussian}
Let $\sigma_k^k$ be the $k$-th moment of $\cN\paren{0, 1}$.
Suppose $\mu$ is in a ball of radius $R$, where $\alpha \le R/4$, and suppose $\eps \le 1$. Any $\eps$-DP algorithm that takes $n$ samples from $\cN\paren{\mu, I_d /  \sigma_k^2}^{\tensor m}$, and outputs $\hat{\mu}$ such that $\normt{\mu - \hat{\mu}} \le \alpha$ with probability $1-\beta$ requires
\begin{equation*}
n = \Omega_k\Paren{\frac{d + \log \paren{1/\beta}}{\alpha \sqrt{m} \eps}} \mcom
\end{equation*}
where $\Omega_k$ hides multiplcative factors that only depend on $k$.
Moreover, $\cN\paren{\mu, I_d / \sigma_k^2}^{\tensor m}$ has $k$-th moment bounded by $1$.
\end{corollary}
\begin{proof}
We apply \Cref{lem:lb-gaussian}. Assume there exists an $\eps$-DP algorithm that takes $n$ samples from $\cN\paren{\mu, I_d / \sigma_k^2}^{\tensor m}$, and outputs an estimate of $\mu$ up to accuracy $\alpha$, with probability $1-\beta$. Now we can take the same algorithm and apply it to $Y_i$'s where $Y_i = X_i / \sigma_k$ and $X_i\distras{}\cN\paren{\mu \sigma_k, I_d}^{\tensor m}$. Multiplying this estimate by $\sigma_k$ would give an $\alpha \sigma_k$ estimate of the mean of $\cN\paren{\mu \sigma_k, I_d}$. However, we know that such an estimate requires at least $\Omega\paren{\paren{d+ \log\paren{1/\beta}}/\paren{\alpha \sigma_k\sqrt{m} \eps}}$ many samples. Since $\sigma_k$ is a constant that only depends on $k$ we must have
\begin{equation*}
n = \Omega_k \Paren{\frac{d + \log\paren{1/\beta}}{\alpha \sqrt{m} \eps}} \mcom
\end{equation*}
as desired.
\end{proof}

Next we prove the third term which we show through constructing point mass distributions that have bounded $k$-th moements. 
\begin{lemma}[Pure DP Point Distribution with Bounded Moments Lower Bound]
\label{lem:lb-bounded-moments}
Suppose $k\ge 2$, and $\cD$ is a distribution with $k$-th moment bounded by $1$ and mean $\mu$.
Moreover, suppose $\mu$ is in a ball of radius $R$, where $\alpha \le \min\set{\frac{6}{25} R, \paren{1/25}^{(k-1)/k}}$, and assume $\eps \le 1$.
Any  $\eps$-DP algorithm that takes $n$ samples from $\cD^{\tensor m}$ and outputs $\hat{\mu}$ such that $\normt{\mu - \hat{\mu}} \le \alpha$ with probability $1-\beta$ requires 
\begin{equation*}
n = \Omega\Paren{\frac{d + \log \paren{1/\beta}}{\alpha^{\frac{k}{k-1}} m \eps}}
\end{equation*}
many samples.
\end{lemma}
\begin{proof}
We want to use \Cref{thm:lb-meta}.
A set $I$, with $\log \abs{I} \ge d \log(2)$ of indices and vectors $v_i$ in the unit ball exist such that $\normt{v_i - v_i} \ge 1/2$.

Suppose $Q_O$ is a distribution with all of the mass on the origin, and $Q_i$ is equal to $0$ with probability $1-\lambda$, and is equal to 
$\frac{1}{6 \alpha^{\frac{1}{k-1}}} v_i$ , with probability $\lambda$, where $\lambda = 25 \alpha^{\frac{k}{k-1}}$. Then $\E_{x \distras{} Q_i}\brac{x} = \frac{25}{6} \alpha v_i$, and for all unit $u$, $\E_{x \distras{} Q_i} \brac{\abs{\iprod{x, u}}^k} \le \frac{25}{6^k} < 1$, therefore $Q_i$'s satisfy the moment assumption. Moreover if we take $G_i = \frac{26}{6} \alpha v_i + \alpha B$, then $G_i$'s are disjoint.
Let $P_i = Q_i^{\tensor m}$, then
\begin{equation*}
\normtv{P_O - P_i} = \normtv{Q_O^{\tensor m} - Q_i^{\tensor m}} \le \Pr_{x \distras{} Q_i^{\tensor m}}\brac{x \neq 0^{\tensor m}} = 
1 - \Pr_{x \distras{} Q_i^{\tensor m}}\brac{x = 0^{\tensor m}}
\le
1 - \paren{1 - \lambda}^m \le m \lambda,
\end{equation*}
where the last inequality comes from the fact that the function $f(y) = 1- my +(1-x)^m, y \in \brac{0,1}$ is an increasing function and is minimized at $0$, and therefore $f(y) \ge 0, \forall y \in \brac{0,1}$. Applying \Cref{thm:lb-meta}
finishes the proof.
\end{proof}

Lastly, we prove the last term which corresponds to the cost of finding a coarse estimate. We show this by choosing a set of point mass distributions over the ball of radius $R$.

\begin{lemma}[Pure DP Lower Bound Range Parameter]
\label{lem:lb-range}
Suppose $k \ge 2$, and $\cD$ is a distribution with $k$-th moment bounded by $1$ and mean $\mu$. Moreover, suppose $\mu$ is in a ball of radius $R$, where $\alpha < R/2 $, and assume $\eps \le 1$. Any $\eps$-DP algorithm that takes $n$ samples from $\cD^{\tensor m}$ and outputs $\hat{\mu}$ such that $\normt{\mu - \hat{\mu}} \le \alpha$ with probability $1-\beta$ requires
\begin{equation*}
n = \Omega\Paren{\frac{d \log \paren{R/\alpha} + \log\paren{1/\beta}}{\eps}}
\end{equation*}
many samples.
\end{lemma}
\begin{proof}
We want to use \Cref{thm:lb-meta}. \Cref{fact:covering-number} implies that there exists a set $I$ of indices such that $\log \abs{I} \ge d \log \paren{R/2\alpha}$ and $v_i$ in the ball of radius $R$ such that $\normt{v_i - v_j} \ge 2\alpha$, for every $i \neq j \in I$. Let $Q_O$ be the distribution with all of the mass on the origin, and $Q_i$'s be the distributions with all of the mass on $v_i$, then $\normtv{Q_O - Q_i} = 1$ for every $i \in I$. Let $P_i = Q_i^{\tensor m}$, then $\normtv{P_O - P_i} = 1$ for every $i \in I$. Let $G_i = v_i + \alpha B$, then $G_i$'s are disjoint. Applying \Cref{thm:lb-meta} finishes the proof.
\end{proof}

\section{Non Uniform Berry-Esseen}
\label{sec:recreating-nube}
In this section we restate and prove the Non Uniform Berry-Esseen Theorem \cite{Michel76}.

\begin{theorem}[\cite{Michel76}]
Let $k \ge 3$.
Assume $X$ is a distribution with mean $0$ and $k$-th moment bounded by $1$ . Then for any $t \geq \sqrt{\paren{k -1} \log m}$, we have

\begin{equation*}
\Pr\Brac{\sum X_i / \sqrt{m} \ge t} \le_k m^{-k/2+1} t^{-2k} + m \Pr\Brac{X \ge rm^{1/2} t},
\end{equation*}
where $r = 1 / (2 (k-1) k)$.
\end{theorem}

\begin{proof}
Let $Y_i$ be the truncation of $X_i$ at $rm^{1/2} t$, i.e. $Y_i = \min \paren{X_i, rm^{1/2} t}$, where $r = 1 / (2 (k-1) k)$. Then
\begin{equation*}
\Pr\Brac{\sum X_i / \sqrt{m} \ge t} \le \Pr\Brac{\sum Y_i / \sqrt{m} \ge t} + m \Pr\Brac{X \ge r\sqrt{m} t}.
\end{equation*}
Therefore, we need to show that
\begin{equation*}
\Pr\Brac{\sum Y_i / \sqrt{m} \ge t} \le_k m^{-k/2+1} t^{-2k}.
\end{equation*}
Let $h = m^{-1/2} t^{-1} \cdot \paren{(k-2) \log m + 2 \paren{k-1}k \log t}$, and  $f(x) = \exp\paren{h m^{1/2} x}$. $f$ is a non-negative non-decreasing function. We apply Markov's inequality.
\begin{align*}
\Pr\Brac{\sum Y_i / \sqrt{m} \ge t} 
&=
\Pr \Brac{f \Paren{\sum Y_i / \sqrt{m}} \ge f(t)} \\
&\le
\frac{\E \Brac{f \Paren{\sum Y_i / \sqrt{m}}}}{f(t)} \\
&=
\frac{\E \Brac{\prod f \Paren{Y_i / \sqrt{m}}}}{f(t)} \\
&=
\frac{\E \Brac{f \Paren{Y_i / \sqrt{m}}}^m}{f(t)} \\
&=
\frac{\E\Brac{\exp\paren{h Y_i}}^m}{\exp\paren{h m^{1/2} t}} \\
&=
\E\Brac{\exp\paren{h Y_i}}^m \cdot \exp\paren{-h m^{1/2} t} \\
&\le
\E\Brac{\exp\paren{h Y_i}}^m \cdot m ^{-k + 2} \cdot t^{-2 \paren{k-1} k}
\end{align*}

Now we have to bound $\E\Brac{\exp\paren{h Y_i}}$. We use the Taylor expansion of $\exp$, and the fact that $Y_i \le rm^{1/2} t$, to get

\begin{equation*}
\exp\paren{h Y} \le 1 + hY + \frac{h^2 Y^2}{2} + \frac{h^3 |Y|^3}{6} \cdot \max_{y' \le rm^{1/2}t}\exp\paren{h  y'} \le 1 + hY + \frac{h^2 Y^2}{2} + \frac{h^3 |Y|^3}{6} \cdot \exp\paren{h rm^{1/2}t}.
\end{equation*}
Taking expectation, and using the fact that $\E\brac{Y_i^2} \le \E\brac{X_i^2} \le 1$, 
\begin{align*}
\E\Brac{\exp\paren{h Y_i}}
&\le
1 + h \abs{\E\Brac{Y_i}} + \frac{h^2 \E\Brac{Y_i^2}}{2} + \frac{h^3 \E\Brac{|Y_i|^3}}{6} \cdot \exp\paren{h rm^{1/2}t} \\
&\le
1 +  h\abs{\E\Brac{Y_i}} + \frac{h^2}{2} + \frac{h^3 \E\Brac{|Y_i|^3}}{6} \cdot \exp\paren{h rm^{1/2}t}.
\end{align*}
Now we bound $\abs{\E\Brac{Y_i}}$ and $\E\Brac{|Y_i|^3}$. We bound $\abs{\E\Brac{Y_i}}$ using \Cref{lem:clipping-bias}.
\begin{align*}
\Abs{\E\Brac{Y_i}} &\le 
\int_{r m^{1/2} t}^\infty \Pr\Brac{X \ge x} \diff x \\
&\le
\int_{r m^{1/2} t}^\infty \Pr\Brac{\Abs{X} \ge x} \diff x \\
&\le
\int_{r m^{1/2} t}^\infty \frac{\E\Brac{X^2}}{x^2} \diff x \\
&\le 
\frac{1}{r m^{1/2} t}.
\end{align*}
In order to bound $\E\Brac{|Y_i|^3}$, note that $\E\Brac{\Abs{Y_i}^3} \le \E\Brac{\Abs{X_i}^3} \le 1$. Therefore, since $t \ge 1$, we have
\begin{align*}
\E\Brac{\exp\paren{hY_i}}
&\le
1 + \frac{h}{rm^{1/2}t} + \frac{h^2}{2} + \frac{h^3}{6} \cdot \exp\paren{h rm^{1/2}t} \\
&=
1 + \frac{h}{rm^{1/2}t} + \frac{h^2}{2} + \frac{h^3}{6} \cdot 
m^{(k-2) / 2(k-1)k} t \\
&\le
1 + \frac{h^2}{2} + b_1 m^{-1},
\end{align*}
where $b_1$ is a constant depending only on $k$. Therefore,
\begin{equation*}
\E\Brac{\exp\paren{hY_i}} \le 1 + \frac{h^2}{2} + b_1 m^{-1} \le \exp\paren{h^2 / 2 + b_1 m^{-1}},
\end{equation*}
and
\begin{equation*}
\E\Brac{\exp\paren{hY_i}}^m \le \exp\paren{h^2 m / 2 + b_1}.
\end{equation*}
Note that 
\begin{align*}
h^2 m / 2 &\le \frac{\paren{\paren{k-2} \log m + 2\paren{k-1} k \log t}^2}{2t^2} \\
&= \frac{\paren{k-2}^2 \log^2 m}{2t^2} + \frac{\paren{k-2} \log m \cdot 2\paren{k-1} k \log t}{t^2} + \frac{4\paren{k-1}^2 k^2 \log^2 t}{2t^2} \\
&=  \frac{\paren{k-2} \log m}{t^2} \cdot \frac{\paren{k-2} \log m}{2} + \frac{\paren{k-1} \log m}{t^2} \cdot 2\paren{k-2} k \log t + 
\frac{\log^2 t}{t^2} \cdot 2\paren{k-1}^2 k^2 \\
&\le
\frac{k-2}{2} \log m + 2\paren{k-1} k \log t + b_2,
\end{align*}
where $b_2$ is a constant depending only on $k$. The last inequality follows from the assumption that $t \ge \sqrt{\paren{k-1} \log m}$, and the fact that $t \ge \log t$. Therefore,
\begin{equation*}
\E\Brac{\exp\paren{hY_i}}^m \le \exp\paren{\frac{k-2}{2} \log m + 2\paren{k-2} k \log t + b_3},
\end{equation*}
where $b_3$ is a constant depending only on $k$. Therefore,
\begin{equation*}
\Pr\Brac{\sum Y_i / \sqrt{m} \ge t} \le \exp\paren{\frac{k-2}{2} \log m + 2\paren{k-2} k \log t + b_3} \cdot m ^{-k + 2} \cdot t^{-2 \paren{k-1} k} \le_k m^{-k/2+1} t^{-2k},
\end{equation*}
and we are done.
\end{proof}

\begin{corollary}[Non-uniform Tail Bound for $k$-th Moment Bounded by $1$ - large $t$]
Let $k\ge 3$. Assume $X$ is a distribution with mean $0$ and $k$-th moment bounded by $1$. Then for any $t \geq \sqrt{\paren{k -1} \log m}$, we have
\begin{equation*}
\Pr\Brac{\sum X_i / \sqrt{m} \ge t} \le_k m^{-k/2+1} t^{-k},
\end{equation*}
and consequently, for any $t \geq \sqrt{\frac{\paren{k -1} \log m}{m}}$, we have
\begin{equation*}
\Pr\Brac{\sum X_i / m \ge t} \le_k m^{-k +1} t^{-k}.
\end{equation*}
\end{corollary}
\begin{proof}
Follows from the previous theorem and the fact that $\Pr\Brac{X \ge rm^{1/2} t} \le r^{-k} m^{-k/2} t^{-k}$.
\end{proof}

\section{Direct Proof of Proposition~\ref{prop:lb_gaussians_ulp}}
\label{sec:direct_proof_fing}

In this appendix, we focus on proving Proposition~\ref{prop:lb_gaussians_ulp} directly using fingerprinting.
To do that, we will use the general lemma for exponential families from~\cite{KamathMS22}, which we state here (in a slightly simplified form):

\begin{proposition}
\label{prop:lower_bound}
Let $p_{\eta}$ be a distribution over $S \subseteq \bR^d$ belonging to an exponential family $\cE\Paren{T, h}$ with natural parameter vector $\eta \in \cH \subseteq \bR^k$.
Also, let $\eta^{\Paren{1}}, \eta^{\Paren{2}} \in \cH$ and let $I_j \coloneqq \left[\eta^{\Paren{1}}_j, \eta^{\Paren{2}}_j\right], \forall j \in \left[k\right]$ be a collection of intervals and $R \coloneqq \eta^{\Paren{2}} - \eta^{\Paren{1}}, m \coloneqq \frac{\eta^{\Paren{1}} + \eta^{\Paren{2}}}{2}$ be the corresponding width and midpoint vectors, respectively.
Assume that $\bigotimes\limits_{j \in \brac{k}} I_j \subseteq \cH$ and $\eta \sim \cU\Paren{\bigotimes\limits_{j \in \brac{k}} I_j}$.
Moreover, assume that we have a dataset $X \sim \Paren{p_{\eta}^{\otimes n} \middle| \eta}$ and an independently drawn point $X_i' \sim \Paren{p_{\eta} \middle| \eta}$ and $X_{\sim i}$ denotes the dataset where $X_i$ has been replaced with $X_i'$.
Finally, let $M \colon S^n \to \bigotimes\limits_{j \in \brac{k}} \left[\pm \frac{R_j}{2}\right]$ be an $\Paren{\eps, \delta}$-DP mechanism with $\eps \in \brac{0, 1}, \delta \geq 0$ and $\underset{X, M}{\mathbb{E}}\left[\left\|M\Paren{X} - \Paren{\eta - m}\right\|_2^2\right] \le \alpha^2 \le \frac{\left\|R\right\|_2^2}{24}$.
Then, for any $T_0 > 0$, it holds that:
\[
    n \left\{2 \delta T_0 + 2 \alpha \eps \underset{\eta}{\mathbb{E}}\left[\sqrt{\left\|\Sigma_T\right\|_2}\right] + 2 \underset{\eta}{\mathbb{E}}\left[\int\limits_{T_0}^{\infty} \underset{X_i}{\mathbb{P}}\left[\left\|T\Paren{X_i} - \mu_T\right\|_2 > \frac{4 t}{\left\|R\right\|_{\infty}^3 \sqrt{k}}\right] \, dt\right]\right\} \geq \frac{\left\|R\right\|_2^2}{24}. \label{eq:main_ineq}
\]
\end{proposition}

We note now that $\cN\Paren{\mu, \mathbb{I}}^{\otimes m}$ is an exponential family.
Indeed, the density of $\cN\Paren{\mu, \mathbb{I}}^{\otimes m}$ is:
\[
    \frac{1}{\Paren{2 \pi}^{\frac{m d}{2}}} e^{- \frac{1}{2} \sum\limits_{i \in [m]} \left\|x_i - \mu\right\|_2^2},
\]
which can be written in the form $h(x) e^{\langle \eta, T(x) \rangle - Z(\eta)}$ for $x \coloneqq (x_1, \dots, x_m) \in \Paren{\mathbb{R}^d}^m$ with:
\begin{align*}
       h\paren{x} &= \frac{1}{\Paren{2 \pi}^{\frac{m d}{2}}} e^{- \frac{1}{2} \sum\limits_{i \in [m]} \left\|x_i\right\|_2^2}, \\
       T\paren{x} &= \sum\limits_{i = 1}^m x_i, \\
             \eta &= \mu, \\
    Z\paren{\eta} &= \frac{m \left\|\eta\right\|_2^2}{2}.
\end{align*}

At this point, we recall the following lemma from~\cite{KamathMS22}.

\begin{lemma}
\label{lem:tail_prob_mean_ub1}[Lemma C.5 from~\cite{KamathMS22}].
Let $X \sim \cN\paren{\mu, \mathbb{I}}$.
Assuming that $T_0 \geq 2 d$, we have:
\[
    \int\limits_{T_0}^{\infty} \underset{X}{\mathbb{P}}\left[\left\|X - \mu\right\|_2 > \frac{t}{2 \sqrt{d}}\right] \, dt \le \sqrt{2 \pi d} e^{- \frac{\Paren{\sqrt{T_0^2 - 2 d^2} - \sqrt{2} d}^2}{8 d}}.
\]
\end{lemma}

Using this lemma, we prove the following result.

\begin{lemma}
\label{lem:tail_prob_mean_ub2}
Let $X \sim \cN\paren{\mu, \mathbb{I}}^{\otimes m}$.
Assuming that $T_0 \geq 2 m d$, we have:
\[
    \int\limits_{T_0}^{\infty} \underset{X}{\mathbb{P}}\left[\left\|T(X) - \mu_T\right\|_2 > \frac{t}{2 \sqrt{d}}\right] \, dt \le \sqrt{2 \pi d} e^{- \frac{\Paren{\sqrt{\Paren{\frac{T_0}{m}}^2 - 2 d^2} - \sqrt{2} d}^2}{8 d}}.
\]
\end{lemma}

\begin{proof}
We note that we have:
\begin{align*}    
    \int\limits_{T_0}^{\infty} \underset{X}{\mathbb{P}}\left[\left\|T(X) - \mu_T\right\|_2 > \frac{t}{2 \sqrt{d}}\right] \, dt
    &= \int\limits_{T_0}^{\infty} \underset{X_1, \dots, X_m}{\mathbb{P}}\left[\left\|\sum\limits_{i = 1}^m X_i - m \mu\right\|_2 > \frac{t}{2 \sqrt{d}}\right] \, dt \\
    &= \int\limits_{T_0}^{\infty} \underset{X \sim \cN(0, 1)}{\mathbb{P}}\left[\left\|X - \mu\right\|_2 > \frac{t}{2 m \sqrt{d}}\right] \, dt \\
    &= \int\limits_{\frac{T_0}{m}}^{\infty} \underset{X \sim \cN(0, 1)}{\mathbb{P}}\left[\left\|X - \mu\right\|_2 > \frac{t}{2 \sqrt{d}}\right] \, dt \\
    &\le \sqrt{2 \pi d} e^{- \frac{\Paren{\sqrt{\Paren{\frac{T_0}{m}}^2 - 2 d^2} - \sqrt{2} d}^2}{8 d}},
\end{align*}
where the last inequality follows directly from Lemma~\ref{lem:tail_prob_mean_ub1}.
\end{proof}

\begin{proof}[Proof of Proposition~\ref{prop:lb_gaussians_ulp}]
Assume that we have $\mu \sim \cU\Paren{[\pm 1]^d}$, and $X \sim \Paren{\Paren{\cN(\mu, \mathbb{I})^{\otimes m}}^{\otimes n} \middle| \mu}$.
Also, let $M \colon \Paren{\mathbb{R}^{m \times d}}^n \to [\pm 1]^d$ be an $\Paren{\eps, \delta}$-DP mechanism with:
\[
    \underset{X, M}{E}\left[\left\|M\paren{X} - \mu\right\|_2^2\right] \le \alpha^2 \le \frac{d}{6} = \frac{\|R\|_2^2}{24}.
\]
By Proposition~\ref{prop:lower_bound} for batches of Gaussian samples $\cN(\mu, \mathbb{I})^{\otimes m}$ with $\mu_j \in [\pm 1], \forall j \in [d]$, we get:
\[
    n \Paren{2 \delta T_0 + 2 \sqrt{m} \alpha \eps + 2 \sqrt{2 \pi d} e^{- \frac{\Paren{\sqrt{\Paren{\frac{T_0}{m}}^2 - 2 d^2} - \sqrt{2} d}^2}{8 d}}} \geq \frac{d}{6}. 
\]
The rest of the proof follows the same steps as the proof in Appendix C.2 of~\cite{KamathMS22}, so we do not repeat it here in its entirety.
The core idea is to show that, for $T_0 = 2 m \sqrt{d} \sqrt{\ln\Paren{\frac{1}{\delta}} + \Paren{\sqrt{\ln\Paren{\frac{1}{\delta}}} + \sqrt{d}}^2}$, and:
\[
    \delta \le \min\left\{\frac{1}{144 \sqrt{2} n m}, \frac{\sqrt{d}}{288 \sqrt{2} n m \sqrt{\ln\Paren{\frac{144 \sqrt{2} n m}{\sqrt{d}}}}}\right\},
\]
we have that:
\[
    \delta T_0 \geq 2 \sqrt{2 \pi d} e^{- \frac{\Paren{\sqrt{\Paren{\frac{T_0}{m}}^2 - 2 d^2} - \sqrt{2} d}^2}{8 d}} \text{ and } 3 n \delta T_0 \le \frac{d}{12}.
\]
\end{proof}

\end{document}